\documentclass[11pt]{article}
\usepackage{hyperref}
\usepackage{amssymb}
\usepackage{amsmath}
\usepackage{mdframed}
\usepackage{subfig}
\usepackage{cancel}
\usepackage{enumerate}
\usepackage{color}
\usepackage{mathtools}
\usepackage{pdfpages}
\usepackage{mathrsfs}
\usepackage{enumitem}
\usepackage{mathtools}
\usepackage{graphicx}
\usepackage{mdframed}
\usepackage{amsfonts}
\usepackage{cancel}
\usepackage{placeins}
\usepackage{amsthm}
\usepackage{tikz}
\usetikzlibrary{shapes,arrows,automata}
\usepackage{xcolor}
\usetikzlibrary{arrows.meta}

\definecolor{goldenpoppy}{rgb}{0.99, 0.76, 0.0}
\definecolor{richblack}{rgb}{0.06, 0.05, 0.03}
\definecolor{cadmiumred}{rgb}{0.89, 0.0, 0.13}
\definecolor{fuchsia}{rgb}{0.3, 0.0, 0.3}
\definecolor{green(ncs)}{rgb}{0.0, 0.52, 0.32}
\tikzstyle{species_T} = [circle,radius=0.1cm, text centered, draw=black, fill=goldenpoppy]
\tikzstyle{species_R} = [circle,radius=0.1cm, text centered, draw=black, fill=white]
\tikzstyle{species_C} = [circle,radius=0.1cm, text centered, draw=black, fill=fuchsia]
\tikzstyle{dots} = [circle,radius=0.1cm,text centered]
\tikzstyle{arrowA} = [-{Latex[length=2mm]},white,dashed]
\tikzstyle{arrowB} = [-{Latex[length=2mm]},green(ncs)]
\tikzstyle{arrowC} = [-{Latex[length=2mm]},cadmiumred]
\tikzstyle{inhibit} = [thick,-|,black!100]
\tikzstyle{loosely dashed}= [dash pattern=on 3pt off 6pt]

\newtheorem{definition}{Definition}[section]
\newtheorem{theorem}{Theorem}[section]

\newtheorem{lemma}{Lemma}[section]
\newtheorem{corollary}{Corollary}[section]

\let\Item\item

\newcommand{\bcol}[1]{{\color{black}#1}}

\mdfdefinestyle{theoremstyle}{%
frametitlerule=true,%
innertopmargin=\topskip,
}
\mdtheorem[style=theoremstyle]{algorithm}{Algorithm}[section]

\begin{document}
\vspace*{-1cm}

\centerline{{\huge Mapping dynamical systems into chemical reactions}}

\medskip
\bigskip

\centerline{
\renewcommand{\thefootnote}{$1$}
{\Large Tomislav Plesa \footnote{
Department of Applied Mathematics and Theoretical Physics, University of Cambridge,
Centre for Mathematical Sciences, Wilberforce Road, Cambridge, CB3 0WA, UK; 
e-mail: tp525@cam.ac.uk}
}}

\medskip
\bigskip

\noindent
{\bf Abstract}: Polynomial dynamical systems (DSs) 
can model a wide range of physical processes.
A special subset of these DSs 
that can model chemical reactions 
under mass-action kinetics 
is called chemical dynamical systems (CDSs).
A fundamental problem, central to synthetic biology,
is to map polynomial DSs into 
dynamically similar CDSs.
In this paper, we introduce the 
\emph{quasi-chemical map} (QCM)
that can systematically solve this problem.
The QCM introduces suitable
state-dependent perturbations into any
given polynomial DS
which then becomes a CDS under sufficiently 
large translations of variables.
This map preserves robust features, 
such as generic equilibria and limit cycles, 
and generic bifurcations,
as well as some temporal properties, 
such as periods of oscillations.
Furthermore, the resulting CDSs
are at most one degree higher than the original DSs.
We showcase the QCM by 
designing relatively simple CDSs
with oscillations, chaos and bifurcations, 
and addressing Hilbert's 16th problem in chemistry.

\section{Introduction}
First-order autonomous ordinary-differential equations with polynomials
on the right-hand side, which we call 
polynomial \emph{dynamical systems} (DSs)~\cite{Perko}, 
can model time-evolution of a 
range of chemical and biological processes~\cite{Feinberg,Janos}.
In particular, for chemical reactions under mass-action kinetics,
the dependent variables in the corresponding DSs
can be interpreted as (non-negative) chemical concentrations, 
and each distinct monomial can be interpreted as a chemical reaction~\cite{Feinberg,Janos}.
In this paper, this subset of DSs is called 
\emph{chemical dynamical systems} (CDSs). 
For example, $\mathrm{d} x/\mathrm{d} t = (1 - x)$ is a CDS which
describes how concentration $x = x(t)$ of a chemical species $X$ changes in time $t$,
with monomials $1$ and $-x$ interpreted respectively as production 
$\varnothing \to X$ and degradation $X \to \varnothing$, 
where $\varnothing$ denotes some neglected species.
On the other hand, DS 
$\mathrm{d} y/\mathrm{d} t = -1$ is not chemical, 
because monomial $-1$ drives $y = y(t)$ to negative values,
and therefore cannot be interpreted as a chemical reaction.

Problems arising in chemistry and biology are often \emph{direct}:
given a CDS, the task is to deduce some of its properties, 
such as number and stability of time-independent solutions
(equilibria) or isolated periodic solutions (limit cycles). 
Conversely, in the field of synthetic biology~\cite{Synthetic}, 
and the sub-field of DNA computing~\cite{DNA} in particular, 
the defining problem is an \emph{inverse} one: given a property, 
the task is to design a CDS
that displays the property~\cite{Me1,Me2,Me3,Me4,Me5,RNCRN}. 
This challenging problem can be approached by designing 
CDSs from scratch, or by suitably modifying existing
DSs with desired features arising from e.g. 
mechanical or electronic applications.
Key to both of the two approaches
are special maps that can transform
general DSs into CDSs, 
while preserving desired dynamical features~\cite{Me1};
in this paper, we call these transformations \emph{chemical maps}.

Any given polynomial DS can always be mapped
to a CDS via a suitable state-dependent
rescaling of time~\cite{Time_change1,Time_change2};
we call the underlying chemical map the \emph{time-warp map}. 
Under this map, all positive solutions of the underlying
DS are preserved; however, this comes at a cost.
Firstly, since time is re-scaled, 
temporal properties, such as periods of oscillations, are 
in general significantly changed.
Secondly, the right-hand side of the CDSs obtained via
this map is of higher degree compared 
to the original DSs, 
and this increase in degree scales with the number of equations.
Mathematically, these two properties may be undesirable
if one wants to design lower-degree or less time-distorted CDSs.
Experimentally, the time-warp map can have limited applicability, 
because time-scales at which dynamical phenomena occur are
often of great importance, and this map distorts them. 
Furthermore, the significantly higher-degree terms 
in the resulting CDSs are 
experimentally more costly to implement~\cite{Me6,DNA}.

Another map put forward for designing CDSs
is the \emph{$\mathbf{x}$-factorable map}~\cite{Samardzija}, 
which involves multiplying the right-hand side of
the differential equations by their respective dependent variables.
The key advantage of this map 
is that the resulting CDSs are 
of only one degree higher than the original DSs.
However, despite numerous successful examples~\cite{Samardzija}, 
to the best of author's knowledge, 
there are no rigorous results in the literature 
that systematically characterize how the 
$\mathbf{x}$-factorable map influences various dynamical features.
 
To bridge the gap, in this paper we introduce
a novel chemical map, which we call the
\emph{quasi-chemical map} (QCM). This map systematically
introduces appropriate state-dependent perturbations
on the right-hand side of any given polynomial DS
which then becomes a CDS
under sufficiently large translations of the dependent variables.
The QCM displays a number of desirable properties:
firstly, the resulting CDSs are 
\emph{at most} one degree higher than the original DSs.
Secondly, this map preserves all dynamical features
that are robust to small perturbations, 
such as generic equilibria and limit cycles.
Finally, the QCM approximately preserves temporal properties, 
such as periods of oscillations.
In addition, we show that, in a special case, 
the QCM can be seen as 
a correction of the $\mathbf{x}$-factorable map,
thus also clarifying validity of the latter map.

The paper is organized as follows. 
In Section~\ref{sec:background}, we provide
some background on DSs and CDSs; more details
are available in Appendix~\ref{app:background}.
In Section~\ref{sec:chemical_maps}, we
define chemical maps, 
and discuss two particular examples: 
the time-warp and $\mathbf{x}$-factorable maps;
more details are presented in Appendix~\ref{app:timechange}.
In Section~\ref{sec:quasi_chemical}, 
we introduce the QCM
and prove its properties; 
a generalization is
presented in Appendix~\ref{app:quasi_chemical_g}.
In Section~\ref{sec:quasi_chemical_applications}, 
we apply the QCM to design CDSs with exotic dynamics
and bifurcations, and to address Hilbert's 16th problem in chemistry;
further details and examples are presented in 
Appendices~\ref{app:linearoscillator}--\ref{app:homoclinic}.
Finally, we provide a summary and 
discussion in Section~\ref{sec:discussion};
further details are presented in
Appendix~\ref{app:nonpoly}. 

\section{Background theory} 
\label{sec:background}
In this section, we present notation 
and background theory used in the paper.

\textbf{Notation}.
The space of integers is denoted by $\mathbb{Z}$,
while the space of real numbers by $\mathbb{R}$.
Subscripts $\ge$ and $>$ restrict these spaces
to respectively non-negative and positive numbers;
for example, $\mathbb{Z}_{\ge}$ is the space
of non-negative integers, while 
$\mathbb{R}_{>}$ is the space of positive real numbers.
Absolute value of $x \in \mathbb{R}$ is denoted by $|x|$.
Euclidean column vectors are denoted by 
$\mathbf{x} = (x_1, x_2, \ldots, x_N)^{\top} \in \mathbb{R}^{N}$, 
where $\cdot^{\top}$ is the transpose operator.
The $1$-norm of 
$\mathbf{x} \in \mathbb{R}^{N}$ is given by 
$\|\mathbf{x}\| = |x_1| + |x_2| + \ldots + |x_N|$.

\subsection{Dynamical systems} \label{sec:dynamical_sys}
Consider ordinary-differential equations given by
\begin{align}
\frac{\mathrm{d} y_i}{\mathrm{d} t} 
& = f_i(y_1,y_2, \ldots, y_N), 
\; \; \; 
\textrm{where } f_i \in  \mathbb{P}_n(\mathbb{R}^N,\mathbb{R}),
\; \; \; i = 1, 2, \ldots, N,
\label{eq:dyn_components} 
\end{align}
where $\mathbb{P}_n(\mathbb{R}^N,\mathbb{R})$
is the space of all polynomial functions 
$f : \mathbb{R}^N \to \mathbb{R}$ of degree at most $n$.
We interpret $t \ge 0$ as time, 
and note that $f_i = f_i(y_1,y_2,\ldots,y_N)$
is autonomous, i.e. does not depend on time explicitly.
Defining 
$\mathbf{y} = (y_1,y_2, \ldots, y_N)^{\top} \in \mathbb{R}^N$
and $\mathbf{f}(\mathbf{y}) = (f_1(\mathbf{y}),f_2(\mathbf{y}), \ldots, 
f_N(\mathbf{y}))^{\top} \in \mathbb{R}^N$, 
system~(\ref{eq:dyn_components}) can be written succinctly 
in a vector form as
\begin{align}
\frac{\mathrm{d} \mathbf{y}}{\mathrm{d} t} & = 
\mathbf{f}(\mathbf{y}),
\; \; \; \textrm{where } 
\mathbf{f} \in \mathbb{P}_n(\mathbb{R}^N,\mathbb{R}^N),
\label{eq:dyn} 
\end{align}
where $\mathbb{P}_n(\mathbb{R}^N, \mathbb{R}^N)$
is the space of all vector-valued functions
$\mathbf{f} : \mathbb{R}^N \to \mathbb{R}^N$ 
whose components are polynomials of degree at most $n$.
We call~(\ref{eq:dyn}) a \emph{dynamical system} (DS)
with state $\mathbf{y}$ and \emph{vector field} $\mathbf{f}$.
We say that DS~(\ref{eq:dyn}) is \emph{polynomial}
with degree $n$ and dimension $N$.
In this paper, we focus on polynomial DSs;
see Section~\ref{sec:discussion} and Appendix~\ref{app:nonpoly}
for a discussion on non-polynomial DSs.

Let $\mathbf{y} = \mathbf{y}(t;\mathbf{y}_0)$ be
the solution of~(\ref{eq:dyn}) 
satisfying the initial condition 
$\mathbf{y}(0) = \mathbf{y}_0 \in \mathbb{R}^N$.  
This solution can be represented as a trajectory in the 
\emph{time-state space} $\mathbb{R}_{\ge} \times \mathbb{R}^N$, i.e. as the set
$\{(t,\mathbf{y}(t; \mathbf{y}_0)) | t \ge 0\}$,
or projected to the \emph{state-space} $\mathbb{R}^N$, i.e. as the set 
$\{\mathbf{y}(t;\mathbf{y}_0) | t \ge 0\}$.

\textbf{Robustness}. 
Let us consider system
\begin{align}
\frac{\mathrm{d} \mathbf{z}}{\mathrm{d} t} & = 
\mathbf{f}(\mathbf{z}) + \mathbf{F}(\mathbf{z}; \mu),
\; \; \; \textrm{where } 
\mathbf{F} \in C^1(\mathbb{R}^N \times \mathbb{R}_{\ge},\mathbb{R}^N)
\textrm{ and } \mathbf{F}(\mathbf{z};0) = \mathbf{0},
\label{eq:dyn3} 
\end{align}
where $C^1 = C^1(\mathbb{R}^N \times \mathbb{R}_{\ge},\mathbb{R}^N)$
is the space of all functions
$\mathbf{F} : \mathbb{R}^N \times \mathbb{R}_{\ge} \to \mathbb{R}^N$ 
that are continuously differentiable in all arguments.
For all sufficiently small $\mu > 0$, DS~(\ref{eq:dyn3})
can be interpreted as a perturbation of DS~(\ref{eq:dyn}):
a ``small" perturbation $\mathbf{F}$
has been added to the vector field $\mathbf{f}$.
Solutions of~(\ref{eq:dyn}) that persist under all such perturbations 
are of special interest.

\begin{definition} $($\textbf{Robust dynamical system}$)$ 
\label{def:robustness}
Given a compact set $\mathbb{K} \subset \mathbb{R}^N$ in the state-space,
and given any vector field $\mathbf{F} \in C^1$ with 
$\mathbf{F}(\mathbf{z};0) = \mathbf{0}$,
assume that there exists $\mu_0 > 0$
such that for all $\mu \in (0,\mu_0)$
the perturbed \emph{DS}~$(\ref{eq:dyn3})$
in $\mathbb{K}$ is qualitatively equivalent 
to the unperturbed \emph{DS}~$(\ref{eq:dyn})$ in $\mathbb{K}$.
Then, \emph{DS}~$(\ref{eq:dyn})$ is 
said to be \emph{robust} in $\mathbb{K}$.
\end{definition} 
\noindent
\bcol{\textbf{Remark}. Robust DSs are also said to be
\emph{structurally stable} in the literature~\cite{Perko,Kuznetsov,Wiggins};
see also e.g.~\cite{Wiggins,Hyperbolicity} for the 
related notion of \emph{hyperbolicity}.}

\noindent
\textbf{Remark}. In Definition~\ref{def:robustness}, 
we say that DSs~(\ref{eq:dyn})
and~(\ref{eq:dyn3}) are \emph{qualitatively equivalent}
in $\mathbb{K} \subset \mathbb{R}^N$ 
if there exists a local change of coordinates
$\mathbf{y} = \boldsymbol{\Phi}(\mathbf{z})$ in $\mathbb{K}$, 
where map $\boldsymbol{\Phi}$ is continuous, 
has a continuous inverse, and preserves
the direction of time~\cite{Perko,Kuznetsov};
see also Appendix~\ref{app:background}. 

\noindent
\textbf{Remark}. One can also define
a \emph{robust solution} of~$(\ref{eq:dyn})$ by 
localizing the set $\mathbb{K}$
from Definition~\ref{def:robustness}.
In particular, a solution of DS~$(\ref{eq:dyn})$ is 
robust if in each of its neighborhood
DS~$(\ref{eq:dyn3})$ with sufficiently small $\mu > 0$
has a qualitatively equivalent solution;
for a more precise definition, 
see Appendix~\ref{app:background}.

Two important classes of solutions of DSs
are hyperbolic equilibria and limit cycles.
An equilibrium $\mathbf{y}^*$ is a time-independent solution of~$(\ref{eq:dyn})$;
it is hyperbolic if other nearby solutions either
exponentially move towards or away from $\mathbf{y}^*$.
A limit cycle $\mathbf{y}_{\tau}$ is a time-periodic 
solution of~$(\ref{eq:dyn})$ with period $\tau > 0$;
hyperbolicity has analogous meaning as for equilibria;
see Appendix~\ref{app:background} for more details.
These two classes of solutions are robust~\cite{Perko,Kuznetsov}.
\bcol{More broadly, hyperbolicity can be defined for any solution of a DS,
and one can show that hyperbolic solutions are robust~\cite{Wiggins}[Theorem 3.6.4].}

\subsection{Chemical dynamical systems} \label{sec:chemical_sys}
Let us consider polynomial DS
\begin{align}
\frac{\mathrm{d} \mathbf{x}}{\mathrm{d} t} & = 
\mathbf{g}(\mathbf{x}),
\; \; \; \textrm{where } 
\mathbf{g} \in \mathbb{P}_n(\mathbb{R}^N,\mathbb{R}^N).
\label{eq:dyn2} 
\end{align}
Assume that, given any non-negative initial condition, 
the solution of~(\ref{eq:dyn2}) is non-negative for all future times; 
then, DS~(\ref{eq:dyn2}) is said to be \emph{non-negative}.
In this case, the non-negative orthant 
is a trapping region for~(\ref{eq:dyn2}), i.e. 
vector field $\mathbf{g}$ points along or 
inwards on the boundary of $\mathbb{R}_{\ge}^N$:
$g_i(x_1,x_2, \ldots, x_N) \ge 0$
if $x_i = 0$ and $x_k \ge 0$ for $k \ne i$.
We now define a special subset of such DSs
that satisfy a stronger non-negativity condition:
each distinct monomial in $g_i$
is non-negative when $x_i = 0$ and $x_k \ge 0$ for $k \ne i$.
This condition ensures that each of the monomials
can be interpreted as a chemical reaction, 
and variables $x_i$ as chemical concentrations~\cite{Feinberg,Janos}.

\begin{definition} $($\textbf{Chemical dynamical system}$)$ \label{def:chem_sys}
Monomial $\alpha x_1^{\nu_{1}} 
x_2^{\nu_{2}} \ldots x_N^{\nu_{N}}$ 
of $g_i \in \mathbb{P}_n(\mathbb{R}^N,\mathbb{R})$, 
with $\alpha \in \mathbb{R}$
and $\nu_1,\nu_2,\ldots,\nu_N \in \mathbb{Z}_{\ge}$,
is \emph{chemical} if $\alpha x_1^{\nu_{1}} 
x_2^{\nu_{2}} \ldots x_N^{\nu_{N}} \ge 0$
for $x_i = 0$ and for all $x_k \ge 0$ with $k \ne i$.
Polynomial $g_i$ is \emph{chemical}
if all of its distinct monomials are chemical.
Vector field $\mathbf{g} = (g_1,g_2,\ldots,g_N)^{\top}$
is \emph{chemical} if all of its component are chemical.
\emph{DS}~$(\ref{eq:dyn2})$ is called a \emph{chemical dynamical system}
(\emph{CDS}) if its vector field is chemical.
Monomials, polynomials, vector fields and dynamical systems
that are not chemical are \emph{non-chemical}. 
Spaces of all $n$-degree polynomial
vector fields that are chemical and non-chemical 
are denoted respectively by 
$\mathbb{P}_n^{\mathcal{C}}(\mathbb{R}^N,\mathbb{R}^N)$ 
and $\mathbb{P}_n^{\mathcal{N}}(\mathbb{R}^N,\mathbb{R}^N)$.
\end{definition}

\textbf{Example}. Consider the two-dimensional quadratic DS
\begin{align}
\frac{\mathrm{d} y_1}{\mathrm{d} t} 
& = f_1(y_1,y_2) = - 1 + 3 y_1, \nonumber \\
\frac{\mathrm{d} y_2}{\mathrm{d} t} 
& = f_2(y_1,y_2) = 5 + 2 y_1 - y_1^2. 
\label{eq:ex_0}
\end{align}
Letting $y_1 = 0$, 
one gets monomial $f_1(0,y_2) = -1$,
which is always negative, and hence non-chemical;
therefore, $f_1$ is non-chemical.
Similarly, letting $y_2 = 0$, 
one obtains polynomial $f_2(y_1,0) = (5 + 2 y_1 - y_1^2)$.
Monomials $5$ and $2 y_1$ are non-negative when $y_1 \ge 0$,
and are hence chemical; on the other hand, 
$- y_1^2 < 0$ when $y_1 > 0$, and is therefore non-chemical; 
hence, $f_2$ is also non-chemical.
Therefore, $\mathbf{f} = (f_1,f_2)^{\top}
\in \mathbb{P}_2^{\mathcal{N}}(\mathbb{R}^2,\mathbb{R}^2)$,
and DS~(\ref{eq:ex_0}) is non-chemical.

\textbf{Example}. Consider the two-dimensional quadratic DS
\begin{align}
\frac{\mathrm{d} y_1}{\mathrm{d} t} 
& = f_1(y_1,y_2) = 1 - 2 y_2 + y_2^2, \nonumber \\
\frac{\mathrm{d} y_2}{\mathrm{d} t} 
& = f_2(y_1,y_2) = 1 - y_1 y_2. 
\label{eq:ex_02}
\end{align}
Monomials $1$ and $y_2^2$ from $f_1(0,y_2) = (1 - 2 y_2 + y_2^2)$
are chemical; however, $-2 y_2$ is non-chemical, 
since $-2 y_2 < 0$ for all $y_2 > 0$, so that 
$f_1$ is non-chemical.
On the other hand, $f_2(y_1,0) = 1 > 0$;
hence, $f_2$ is chemical. 
Therefore, $\mathbf{f} 
\in \mathbb{P}_2^{\mathcal{N}}(\mathbb{R}^2,\mathbb{R}^2)$,
and DS~(\ref{eq:ex_02}) is non-chemical.
However, note that DS~(\ref{eq:ex_02}) is non-negative, 
since $f_1(0,y_2) = (1 - y_2)^2 \ge 0$ and $f_2(y_1,0) = 1 > 0$, 
i.e. $\mathbb{R}_{\ge}^2$ is a trapping region.
This example shows that CDSs are a proper
subset of non-negative DSs.

\textbf{Example}. Consider the two-dimensional cubic DS
\begin{align}
\frac{\mathrm{d} x_1}{\mathrm{d} t} 
& = g_1(x_1,x_2) = 2 x_1, \nonumber \\
\frac{\mathrm{d} x_2}{\mathrm{d} t} 
& = g_2(x_1,x_2) = 5 + 2 x_1 - x_1^2 x_2.
\label{eq:ex_01}
\end{align} 
Polynomial $g_1$ is chemical, since $g_1(0,x_2) = 0$;
similarly, $g_2$ is chemical, since 
both monomials of $g_2(x_1,0) = (5 + 2 x_1)$ 
are non-negative when $x_1 \ge 0$.
Hence, $\mathbf{g} \in 
\mathbb{P}_3^{\mathcal{C}}(\mathbb{R}^2, \mathbb{R}^2)$,
and~(\ref{eq:ex_01}) is a CDS.

\textbf{Chemical reaction networks}. 
To every CDS one can associate a
set of chemical reactions, which are jointly called
a \emph{chemical reaction network} (CRN)~\cite{Janos}.
For example, a CRN associated with CDS~(\ref{eq:ex_01}) is given by
\begin{align}
X_1 & \xrightarrow[]{2} 2 X_1, \; \; \; \; 
\varnothing \xrightarrow[]{5} X_2, \; \; \; \; 
X_1 \xrightarrow[]{2} X_1 + X_2, \; \; \; \; 
2 X_1 + X_2 \xrightarrow[]{1} 2 X_1, \label{eq:ex_01_CRN1}
\end{align}
where $X_1$ and $X_2$ denote the chemical species 
whose concentrations are respectively $x_1$ and $x_2$, 
and $\varnothing$ denotes some neglected species;
the positive numbers above the reaction arrows are 
called \emph{rate coefficients}.
Another CRN of~(\ref{eq:ex_01}) 
is obtained by fusing the first 
and third reaction in~(\ref{eq:ex_01_CRN1}), and reads
\begin{align}
X_1 & \xrightarrow[]{2} 2 X_1 + X_2, \; \; \; \; 
\varnothing \xrightarrow[]{5} X_2, \; \; \; \; 
2 X_1 + X_2 \xrightarrow[]{1} 2 X_1. \label{eq:ex_01_CRN2}
\end{align}
See Appendix~\ref{app:background}
for more details about CRNs.

\section{Chemical maps} 
\label{sec:chemical_maps}
In this section, we define the central objects of this paper:
maps that transform DSs (see Section~\ref{sec:dynamical_sys})
into CDSs (see Section~\ref{sec:chemical_sys}), 
while preserving desired dynamical properties.
Before providing a definition, let us consider a
cautionary example.

\textbf{Example}. 
Consider the two-dimensional quadratic DS
\begin{align}
\frac{\mathrm{d} x_1}{\mathrm{d} t} 
& = - \frac{45}{4} + \frac{1}{8} x_1 + \frac{1}{2} x_2, \nonumber \\
\frac{\mathrm{d} x_2}{\mathrm{d} t} 
& = \frac{135}{8} +  \frac{1}{2} x_1 - \frac{51}{32} x_2 - \frac{1}{20} x_1 x_2
 + \frac{1}{20} x_2^2, \label{eq:ex1_translated2}
\end{align}
which has an unstable 
hyperbolic equilibrium $(x_1^*,x_2^*) = (10,20)$ 
enclosed by a stable hyperbolic limit cycle.
We display these two features respectively as the blue dot
and red curve in state-space in Figure~\ref{fig:1}(a);
also shown as grey arrows
is the underlying vector field.
In Figure~\ref{fig:1}(b), we show the limit 
cycle in the $(t,x_1)$ time-state space.

One can notice from~Figure~\ref{fig:1}(a) that DS~(\ref{eq:ex1_translated2})
is non-chemical, since the vector field can point outside 
the non-negative quadrant on the $x_2$-axis. 
More precisely, (\ref{eq:ex1_translated2}) 
does not satisfy Definition~\ref{def:chem_sys} because
of the non-chemical monomial $(-45/4)$ in the first equation.
A naive approach to transform this DS into a CDS
is to simply multiply $(-45/4)$ by $x_1$, leading to 
\begin{align}
\frac{\mathrm{d} x_1}{\mathrm{d} t} 
& = \left(- \frac{45}{4} + \frac{1}{8}\right) x_1 + \frac{1}{2} x_2, \nonumber \\
\frac{\mathrm{d} x_2}{\mathrm{d} t} 
& =\frac{135}{8} +  \frac{1}{2} x_1 - \frac{51}{32} x_2 - \frac{1}{20} x_1 x_2
 + \frac{1}{20} x_2^2.  \label{eq:ex1_naive}
\end{align}
Even though only one term of the vector field
from~(\ref{eq:ex1_translated2}) has been modified to 
yield~(\ref{eq:ex1_naive}), the underlying dynamics
has been drastically changed:
(\ref{eq:ex1_naive}) has no equilibria
and, consequently, all of its solutions grow unboundedly, 
i.e. we have designed a hazardous CDS. 
This catastrophic phenomenon, involving destruction of all non-negative 
equilibria and blow-up of chemical concentrations, 
plays an important role in molecular control~\cite{Me5}.

\begin{definition} $($\textbf{Chemical map}$)$ \label{def:chemical_map}
Consider \emph{DS}~$(\ref{eq:dyn})$
in a desired compact set $\mathbb{K} \subset \mathbb{R}^N$ 
in state-space.
Consider also \emph{CDS}
\begin{align}
\frac{\mathrm{d} \mathbf{x}}{\mathrm{d} t} & = 
\Psi \mathbf{f}(\mathbf{x}),
\; \; \; \textrm{where } 
\Psi \mathbf{f} \in \mathbb{P}_{m}^{\mathcal{C}}
(\mathbb{R}^{N},\mathbb{R}^{N}).
\label{eq:chem} 
\end{align} 
Assume that \emph{CDS}~$(\ref{eq:chem})$ in some compact set
$\mathbb{K}' \subset \mathbb{R}_{\ge}^N$ in the non-negative orthant 
is qualitatively equivalent to \emph{DS}~$(\ref{eq:dyn})$ in $\mathbb{K}$.
Then $\Psi : \mathbb{P}_n(\mathbb{R}^N,\mathbb{R}^{N}) 
\to \mathbb{P}_{m}^{\mathcal{C}}(\mathbb{R}^{N},\mathbb{R}^{N})$, 
mapping vector field $\mathbf{f}$ to $\Psi \mathbf{f}$,
is a \emph{chemical map} that qualitatively 
preserves \emph{DS}~$(\ref{eq:dyn})$ in $\mathbb{K}$.
\end{definition}

A natural choice for chemical maps are affine maps, 
since then qualitative equivalence is ensured.
However, for a given DS, there is no a-priori 
guarantee that a suitable affine map exists. 
Furthermore, even if it does exist, 
finding this map can be a non-trivial task 
even for two-dimensional DSs~\cite{Escher},
let alone higher-dimensional ones.
For this reason, we do not focus on affine maps alone.
In what follows, 
we apply two different non-affine maps to
transform~(\ref{eq:ex1_translated2}) into CDSs.

\begin{figure}[!htbp]
\vskip  0.2cm
\leftline{\hskip 1.0cm 
(a) Original system \hskip 1.7cm 
(c) Time-warp map \hskip 1.4cm 
(e) $\mathbf{x}$-factorable map}
\vskip  0.2cm
\centerline{
\hskip -0.4cm
\includegraphics[width=0.27\columnwidth]{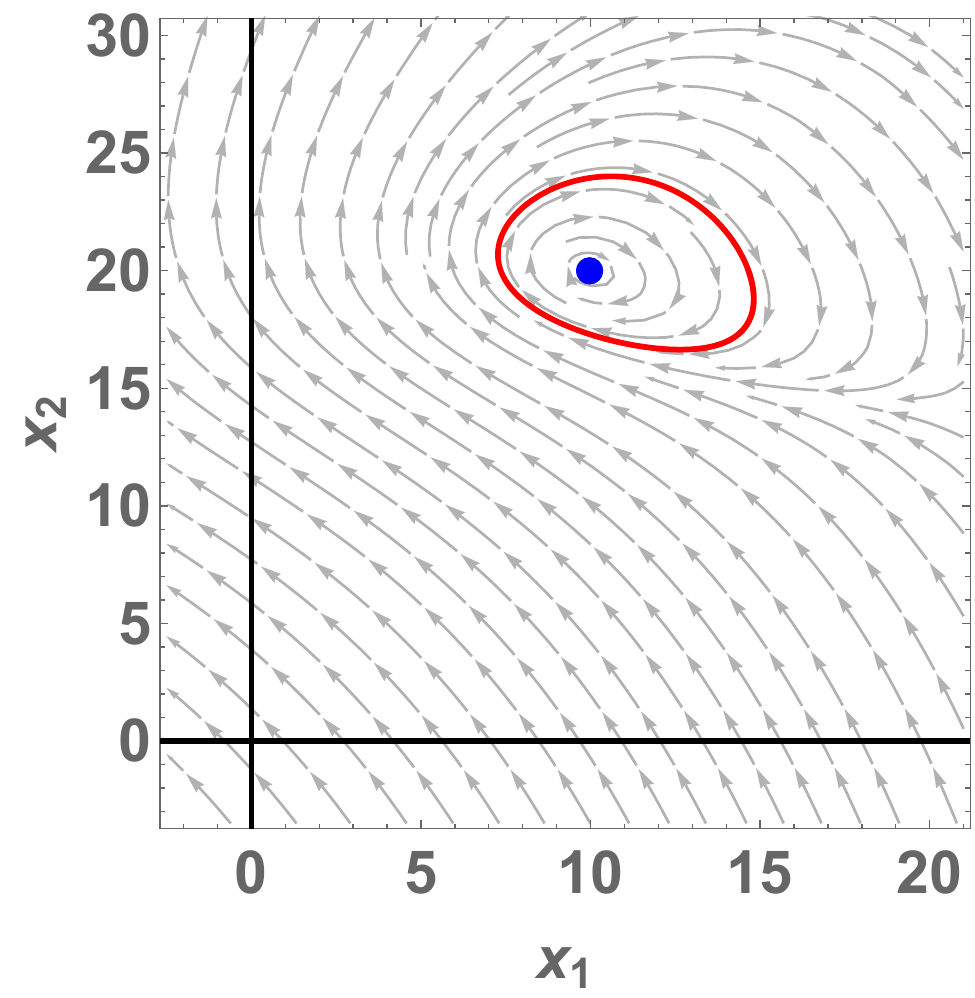}
\hskip 0.5cm
\includegraphics[width=0.27\columnwidth]{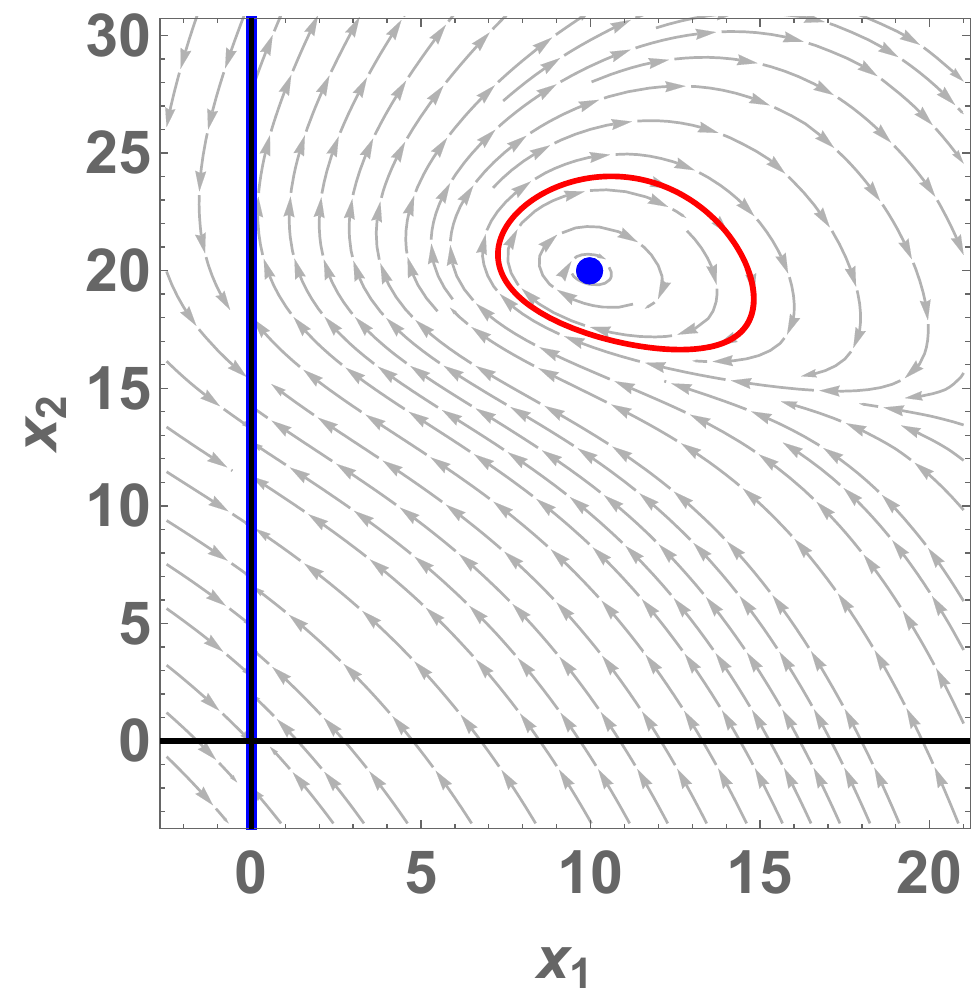}
\hskip 0.5cm
\includegraphics[width=0.27\columnwidth]{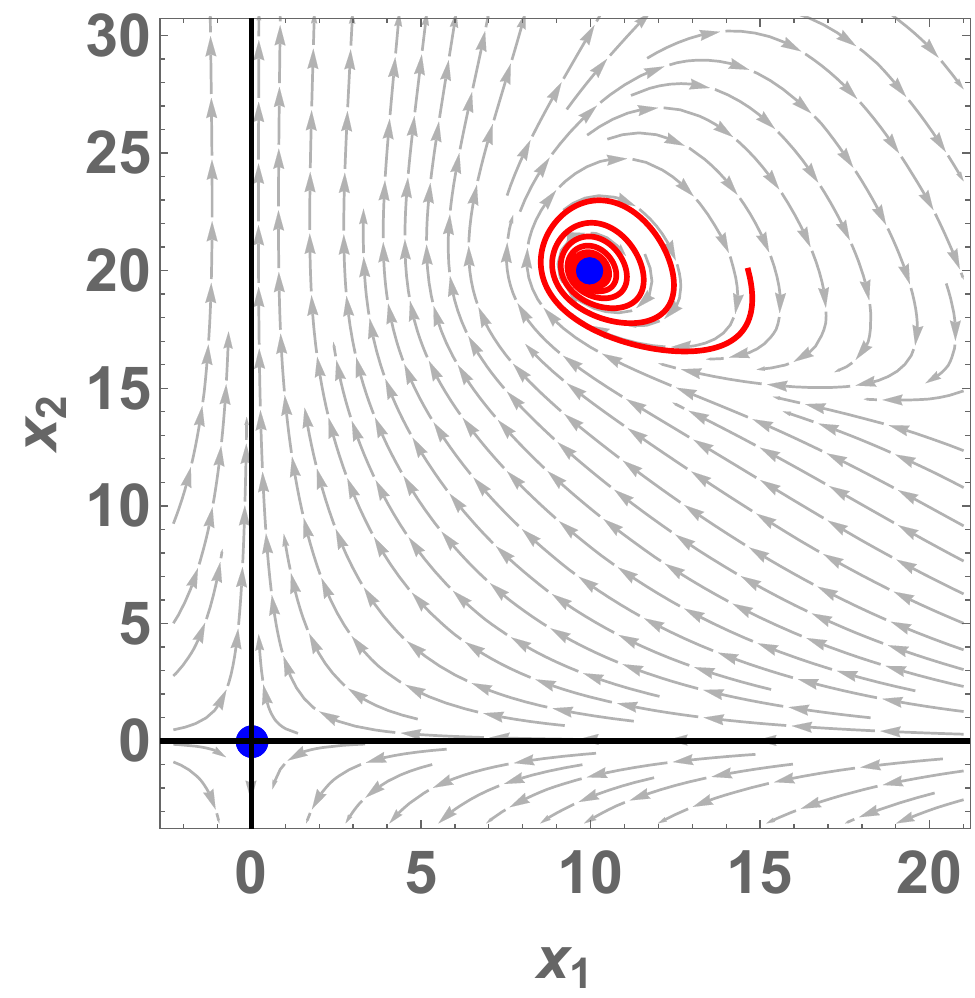}
}
\vskip 0.0cm
\leftline{\hskip 1.0cm (b) \hskip  4.4cm (d) \hskip 4.5cm  (f)}
\vskip 0.2cm
\centerline{
\includegraphics[width=0.3\columnwidth]{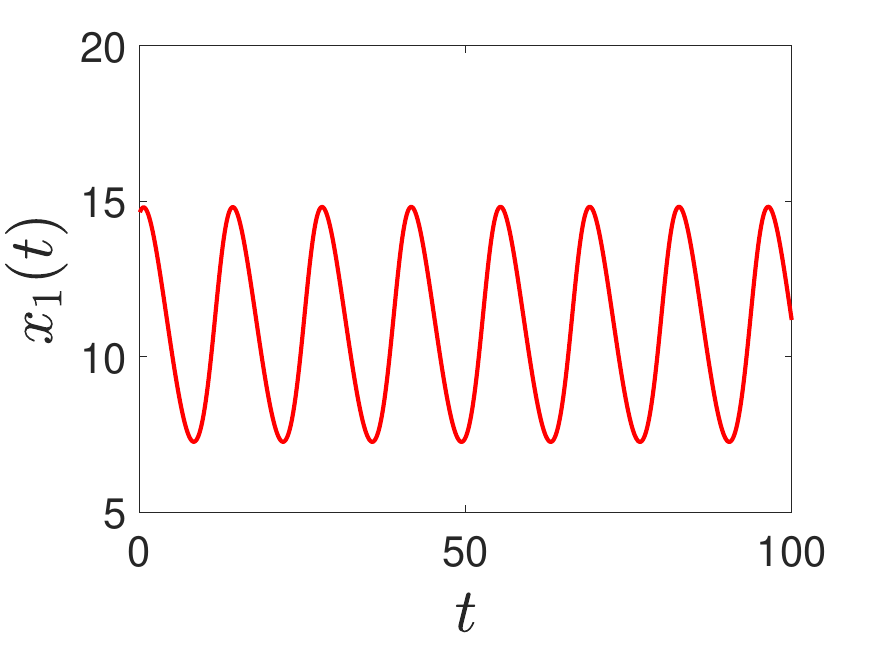}
\includegraphics[width=0.3\columnwidth]{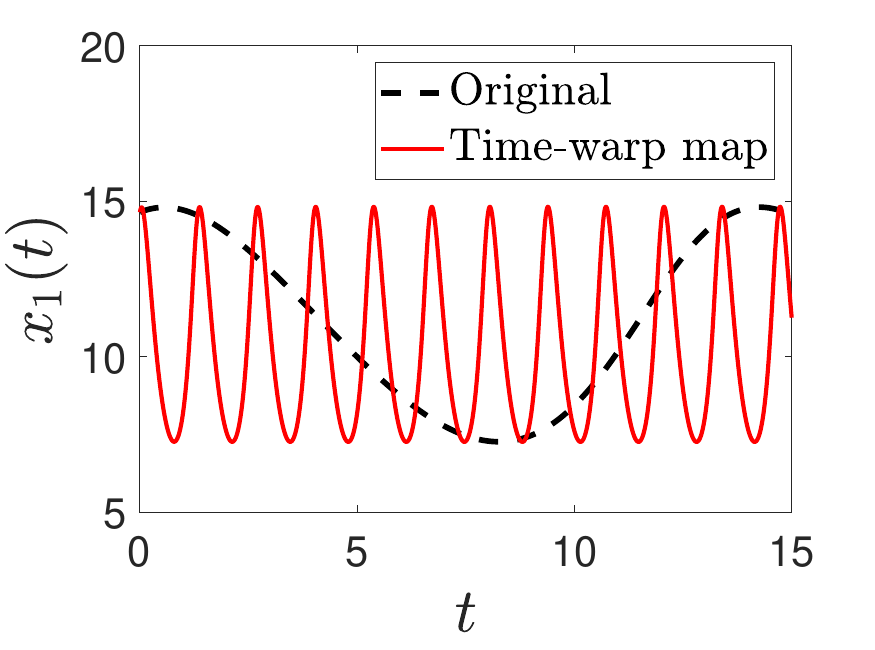}
\includegraphics[width=0.3\columnwidth]{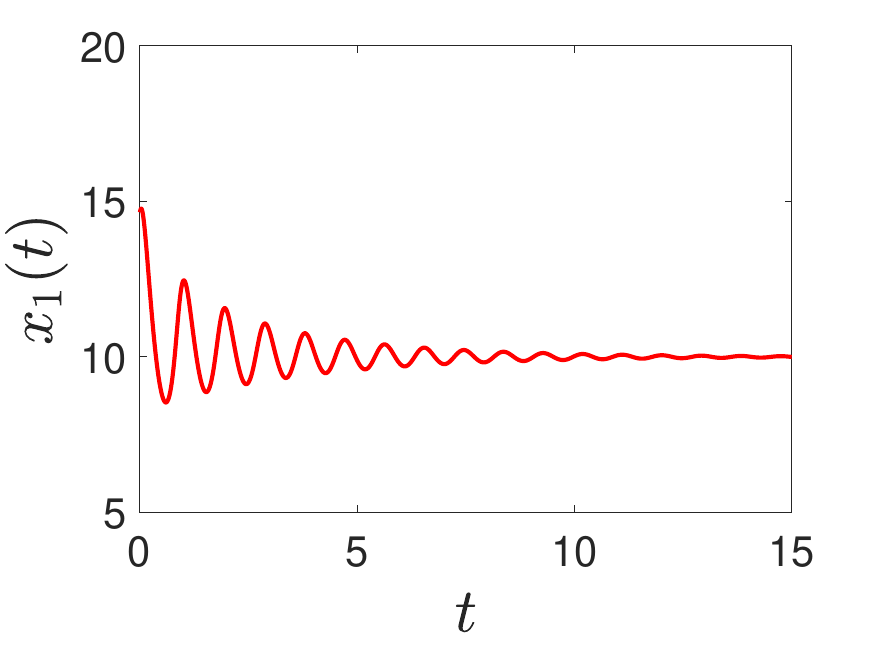}
}
\vskip 0.0cm
\caption{\it{\emph{Application of chemical maps on DS~(\ref{eq:ex1_translated2}).}} 
Panel \emph{(a)} displays state-space
for the original \emph{DS}~$(\ref{eq:ex1_translated2})$, with
the vector field shown as gray arrows,
the equilibrium as the blue dot,
and the limit cycle as the red curve; 
panel \emph{(b)} displays a corresponding $(t, x_1)$-space.
Panels \emph{(c)} and \emph{(d)} show analogous plots 
for \emph{CDS}~$(\ref{eq:ex1_timechange})$, which is obtained by applying 
the time-warp map on~$(\ref{eq:ex1_translated2})$. 
Analogous plots are shown in panels \emph{(e)} and \emph{(f)}
for \emph{CDS}~$(\ref{eq:xfactorable})$,
which is obtained by applying the $\mathbf{x}$-factorable map
on~$(\ref{eq:ex1_translated2})$.
In panel \emph{(e)}, shown as the blue dot 
at the origin is one of the boundary equilibria
introduced by the $\mathbf{x}$-factorable map.
} \label{fig:1}
\end{figure}

\subsection{Time-warp map}
Instead of multiplying only $(-45/4)$ by $x_1$,
as in the naive approach~(\ref{eq:ex1_naive}),
let us instead multiply the whole vector field 
from~(\ref{eq:ex1_translated2}) 
by $x_1$, and denote time by $s$, leading to
the CDS
\begin{align}
\frac{\mathrm{d} x_1}{\mathrm{d} s}
& = x_1 \left(- \frac{45}{4} + \frac{1}{8} x_1
 + \frac{1}{2} x_2 \right), \nonumber \\
\frac{\mathrm{d} x_2}{\mathrm{d} s}
& = x_1 \left(\frac{135}{8} +  \frac{1}{2} x_1 -
 \frac{51}{32} x_2 - \frac{1}{20} x_1 x_2
 + \frac{1}{20} x_2^2 \right).
 \label{eq:ex1_timechange}
\end{align}
We call the map that transforms DS~(\ref{eq:ex1_translated2}) 
into CDS~(\ref{eq:ex1_timechange}) a \emph{time-warp map}.
More generally, the time-warp map can be applied systematically as follows:
assume that the first $K \ge 1$ equations 
from~(\ref{eq:dyn2}) are non-chemical, 
while the remaining $(N-K)$ equations are chemical;
then, multiply \emph{all} $N$ components 
$g_1, g_2, \ldots, g_N$ of the vector field by 
$x_{1} x_{2} \ldots x_{K}$.

This map is equivalent to a state-dependent change of time. 
In particular, consider~(\ref{eq:ex1_translated2})
with an initial condition $\mathbf{x}_0 \in \mathbb{R}_{>}^2$
such that $x_1(t; \mathbf{x}_0) > 0$ over a desired time-interval.
Let us introduce a new time, given by
\begin{align}
s(t) & = \int_0^t \frac{\mathrm{d} \theta}{x_1(\theta; \mathbf{x}_0)}.
\label{eq:time_change_s}
\end{align}
Since $x_1(\theta; \mathbf{x}_0) > 0$ over the desired time-interval, 
$s(t)$ is then well-defined, positive and monotonically increasing;
in particular, it preserves the direction of time $t$, 
and it has an inverse, which we denote by $t(s)$.  
Applying the chain-rule to $x_1 = x_1(t(s))$ 
and $x_2 = x_2(t(s))$ from~(\ref{eq:ex1_translated2}), 
it follows that $\mathrm{d} x_1/\mathrm{d} s = 
(\mathrm{d} t/\mathrm{d} s) \mathrm{d} x_1/\mathrm{d} t 
= x_1 \mathrm{d} x_1/\mathrm{d} t$ 
and similarly $\mathrm{d} x_2/\mathrm{d} s = x_1 \mathrm{d} x_2/\mathrm{d} t$,
which yields~(\ref{eq:ex1_timechange}).
From this consideration, it follows that 
the time-warp map is a chemical one.
In particular, this map qualitatively preserves
 every trajectory that remains within the positive quadrant. 
We verify this fact in Figure~\ref{fig:1}(c), 
which shows state-space for~(\ref{eq:ex1_timechange}).

However, the time-warp map displays three undesirable features.
Firstly, it introduces a continuum of equilibria on the 
$x_2$-axis, which we show in blue in Figure~\ref{fig:1}(c).
Consequently, for initial conditions in certain regions, 
the solutions of~(\ref{eq:ex1_timechange}) asymptotically 
approach the $x_2$-axis - 
a behavior qualitatively different 
from the original DS~(\ref{eq:ex1_translated2}). 
This property arises from the fact that~(\ref{eq:time_change_s})
is not differentiable at time $t$ such that $x_1(t) = 0$.
Secondly, period of oscillations is significantly changed under this map.
In particular, since $x_1(t)$ oscillates around 
the equilibrium $x_1^* = 10$, it follows from~(\ref{eq:time_change_s})
that the period is reduced roughly by a factor of $10$.
This observation is confirmed in Figure~\ref{fig:1}(d): 
we show in red the limit cycle 
of~(\ref{eq:ex1_timechange}) in the $(s,x_1)$-space, 
together with the limit cycle of~(\ref{eq:ex1_translated2}) 
as the black dashed curve.
Finally, the time-warp map in general significantly increases 
degree of DSs, i.e. degree $m$
can be significantly larger than $n$ in Definition~\ref{def:chemical_map}.
For example, if~(\ref{eq:ex1_translated2}) also had a non-chemical term 
in the second equation, then the vector fields would
have to be multiplied by $x_1 x_2$. 
See Appendix~\ref{app:timechange} and e.g.~\cite{Time_change1,Time_change2}
for more details about this map.

\subsection{$\mathbf{x}$-factorable map}
Let us now multiply vector field in the 
first equation of~(\ref{eq:ex1_translated2}) by $x_1$, 
and that in the second equation by $x_2$, thus obtaining CDS
\begin{align}
\frac{\mathrm{d} x_1}{\mathrm{d} t}
& = 
x_1 \left(- \frac{45}{4} + \frac{1}{8} x_1 + \frac{1}{2} x_2 \right), \nonumber \\
\frac{\mathrm{d} x_2}{\mathrm{d} t}
& = 
x_2 \left(\frac{135}{8} +  \frac{1}{2} x_1 - 
\frac{51}{32} x_2 - \frac{1}{20} x_1 x_2
 + \frac{1}{20} x_2^2 \right).
\label{eq:xfactorable}
\end{align}
The map that transforms DS~(\ref{eq:ex1_translated2})
into CDS~(\ref{eq:xfactorable}) 
is called the \emph{$\mathbf{x}$-factorable map}~\cite{Samardzija}.
More generally, the $\mathbf{x}$-factorable map 
can be applied systematically as follows:
multiply the vector field $g_i$ in
equation $i$ from~(\ref{eq:dyn2}) by $x_i$
for all $i = 1, 2, \ldots, N$.
State-space for~(\ref{eq:xfactorable})
is shown in Figure~\ref{fig:1}(e),
while a corresponding $(t,x_1)$-space in Figure~\ref{fig:1}(f).
One can notice that the $\mathbf{x}$-factorable map
has reversed stability of the target equilibrium and 
destroyed the limit cycle.

In contrast to the failure of the $\mathbf{x}$-factorable map 
in Figure~\ref{fig:1}(e)--(f), a number of examples are presented 
in~\cite{Samardzija} indicating that this map can preserve 
equilibria, limit cycles, and even chaotic attractors; 
however, no rigorous results are put forward.
Instead, the authors from~\cite{Samardzija} suggest a heuristic: 
to preserve a dynamical feature in the state-space, this feature
should be translated sufficiently far from the axes 
in the non-negative orthant before the $\mathbf{x}$-factorable map is applied. 
Let us stress that this heuristic 
does not impose any constraints on the direction in which
the target feature is translated; it simply demands that 
the translation is sufficiently large. 
Let us now prove by counter-example that this heuristic can fail
even for hyperbolic equilibria.

\textbf{Counter-example}. 
Consider a perturbed harmonic oscillator
\begin{align}
\frac{\mathrm{d} y_1}{\mathrm{d} t} 
& = \frac{1}{2} y_2 + \frac{25}{2} \varepsilon^2 y_1, \nonumber \\
\frac{\mathrm{d} y_2}{\mathrm{d} t} 
& = -\frac{1}{2} y_1
+ \frac{1}{2} \varepsilon 
\left(- y_1 y_2 + y_2^2 \right)
- \frac{75}{8} \varepsilon^2 y_2. \label{eq:ex1}
\end{align}
Using e.g. theory from~\cite{Perko}[Chapter 4.11], 
one can show that~(\ref{eq:ex1}) has 
an unstable hyperbolic equilibrium 
at the origin surrounded by 
a unique stable hyperbolic limit cycle for all $\varepsilon > 0$ 
sufficiently small.
Under translations
$x_1 = (y_1 + 10)$, $x_2 = (y_2 + 20)$ and with $\varepsilon = 1/10$, 
DS~(\ref{eq:ex1}) becomes~(\ref{eq:ex1_translated2}).

Let us now fix $\varepsilon = 1/10$, translate
the variables via $x_1 = (y_1 + 1/\mu)$ and $x_2 = (y_2 + 2/\mu)$
in~(\ref{eq:ex1}), where $\mu > 0$ is a parameter, 
and then apply the $\mathbf{x}$-factorable map, thus obtaining CDS
\begin{align}
\frac{\mathrm{d} x_1}{\mathrm{d} t} 
& = x_1 \left[\frac{1}{8} \left(x_1 - \frac{1}{\mu} \right)
+ \frac{1}{2} \left(x_2 - \frac{2}{\mu} \right)  \right], \nonumber \\
\frac{\mathrm{d} x_2}{\mathrm{d} t} 
& = x_2 \left[-\frac{1}{2} \left(x_1 - \frac{1}{\mu} \right)
- \frac{1}{20}  \left(x_1 - \frac{1}{\mu} \right) 
\left(x_2 - \frac{2}{\mu} \right)
+ \frac{1}{20}  \left(x_2 - \frac{2}{\mu} \right)^2 
- \frac{3}{32} \left(x_2 - \frac{2}{\mu} \right) \right]. 
\label{eq:xfactorable2}
\end{align}
Note that~(\ref{eq:xfactorable2}) reduces to~(\ref{eq:xfactorable}) 
when $\mu = 1/10$. 
Does the heuristic from~\cite{Samardzija} hold, i.e.
does the equilibrium $(1/\mu,2/\mu)$ from~(\ref{eq:xfactorable2})
qualitatively match the unstable equilibrium $(10,20)$ from~(\ref{eq:ex1_translated2})
if $\mu > 0$ is sufficiently small?
Trace of the Jacobian matrix for~(\ref{eq:xfactorable2}) 
at $(1/\mu,2/\mu)$ is given by
$-1/(16 \mu)$, and the determinant by $61/(128 \mu^2)$. 
Therefore, the equilibrium $(1/\mu,2/\mu)$
is stable for all $\mu > 0$, i.e.
no matter how far it is translated from the boundary of $\mathbb{R}_{\ge}^2$.
Therefore, the heuristic fails, and the $\mathbf{x}$-factorable map, 
not even preserving the hyperbolic equilibrium, is not chemical.

\section{Quasi-chemical map: Theory} 
\label{sec:quasi_chemical}
In Section~\ref{sec:chemical_maps}, we have designed CDSs 
by firstly translating the variables, and then 
multiplying the vector field with appropriate monomials.
In this section, we take a different approach
to designing CDSs:
we first introduce suitable perturbations 
to the vector field, and then employ large translations of variables.
Let us introduce this new approach 
with two examples, before presenting general results.

\textbf{Example}.
Let us fix $\varepsilon = 1/10$ in DS~(\ref{eq:ex1}), 
and perturb its vector field as follows
\begin{align}
\frac{\mathrm{d} z_1}{\mathrm{d} t} 
& = \frac{1}{2} z_2 + \frac{1}{8} z_1
+ \mu z_1 \left[\frac{1}{2} z_2 + \frac{1}{8} z_1 \right], 
\nonumber \\
\frac{\mathrm{d} z_2}{\mathrm{d} t} 
& = -\frac{1}{2} z_1
+ \frac{1}{20}
\left(- z_1 z_2 + z_2^2 \right)
- \frac{3}{32} z_2
+ \frac{\mu}{2} z_2 \left[-\frac{1}{2} z_1
+ \frac{1}{20}
\left(- z_1 z_2 + z_2^2 \right)
- \frac{3}{32} z_2 \right]. \label{eq:ex1_P1}
\end{align}
The target equilibrium and limit cycle of~(\ref{eq:ex1}) are hyperbolic,
and therefore robust; 
see Section~\ref{sec:dynamical_sys}.
Hence, there exists $\mu_0 > 0$ such that
for all $\mu \in (0,\mu_0)$ DS~(\ref{eq:ex1_P1})
displays an equilibrium surrounded by 
a limit cycle that are arbitrarily close and qualitatively equivalent
to those of~(\ref{eq:ex1}). 

\begin{figure}[!htbp]
\vskip  -2.3cm
\leftline{\hskip 0.6cm 
(a) (\ref{eq:ex1_P1c}) with $\mu = 1/10$ \hskip 1.1cm 
(b) (\ref{eq:ex1_P1c}) with $\mu = 1/100$ \hskip 1.1cm 
(c) (\ref{eq:ex1_P1c}) with $\mu = 1/100$}
\vskip  0.2cm
\centerline{
\hskip -0.6cm
\includegraphics[width=0.27\columnwidth]{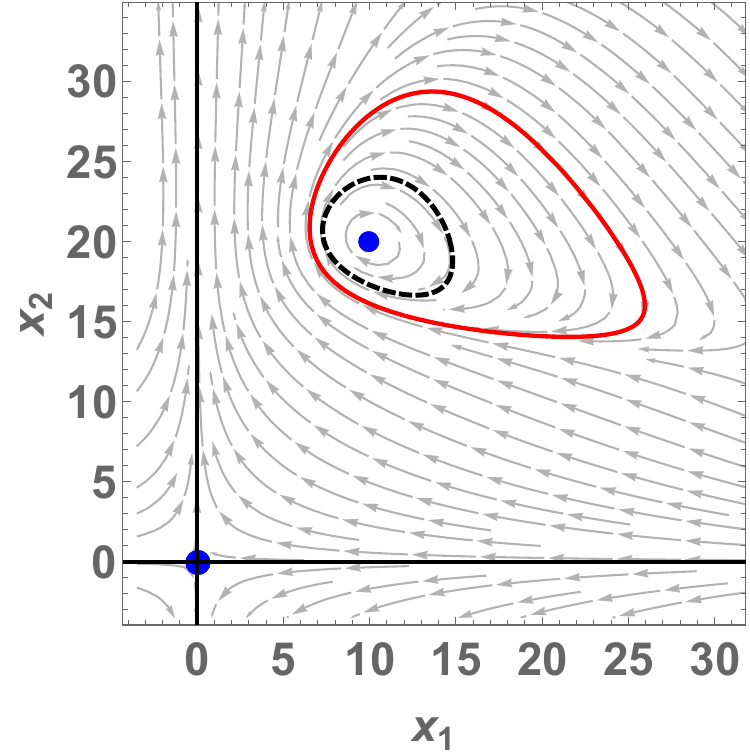}
\hskip 0.3cm
\includegraphics[width=0.283\columnwidth]{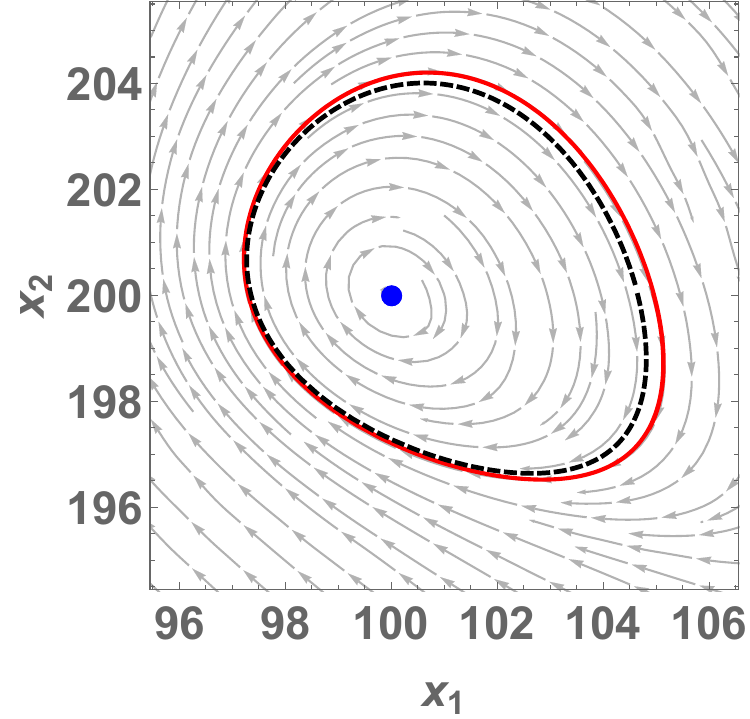}
\hskip 0.3cm
\includegraphics[width=0.35\columnwidth]{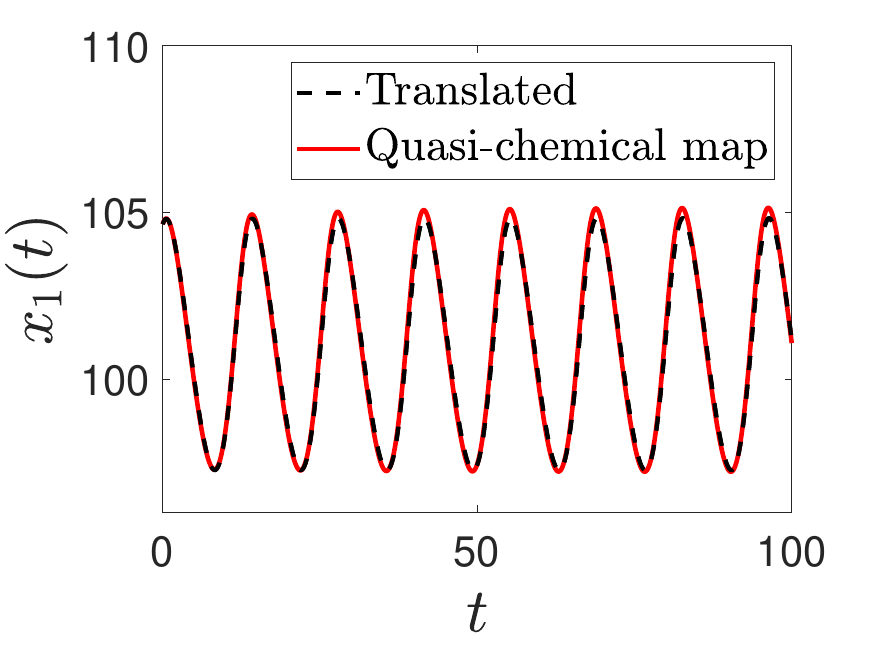}
}
\vskip 0.2cm
\leftline{\hskip 0.6cm 
(d) (\ref{eq:ex1_P2c}) with $\mu = 1/10$ \hskip 1.1cm 
(e) (\ref{eq:ex1_P2c}) with $\mu = 1/100$ \hskip 1.1cm 
(f) (\ref{eq:ex1_P2c}) with $\mu = 1/100$}
\vskip  0.2cm
\centerline{
\hskip -0.6cm
\includegraphics[width=0.27\columnwidth]{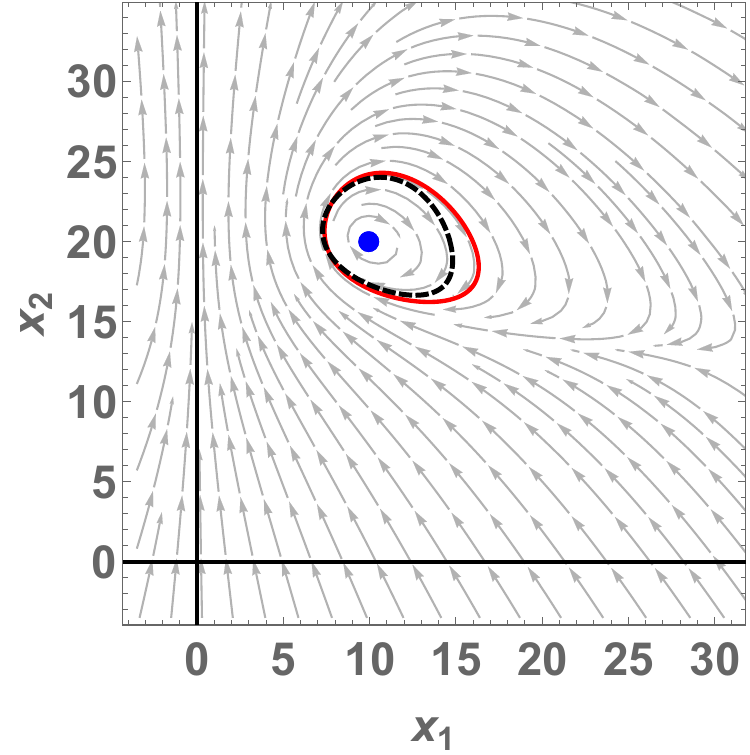}
\hskip 0.3cm
\includegraphics[width=0.283\columnwidth]{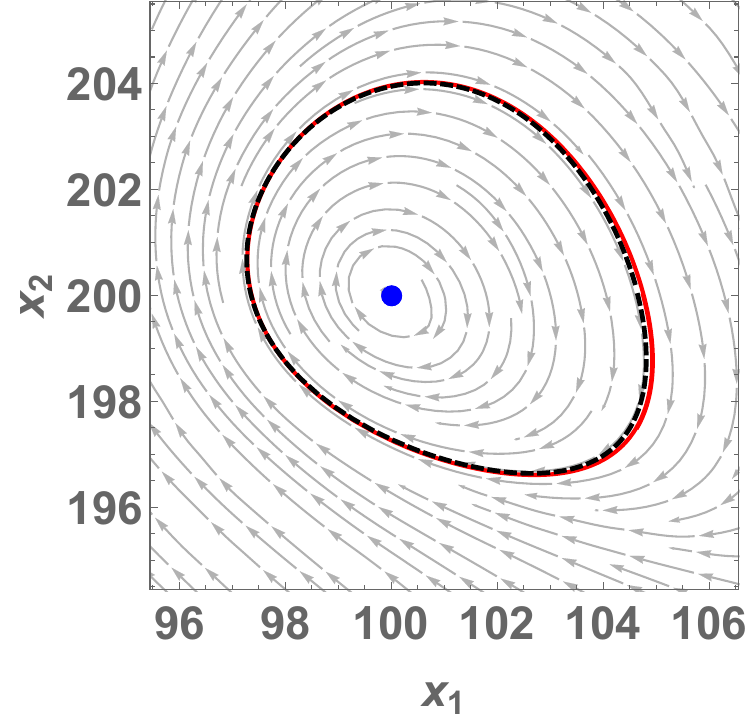}
\hskip 0.3cm
\includegraphics[width=0.35\columnwidth]{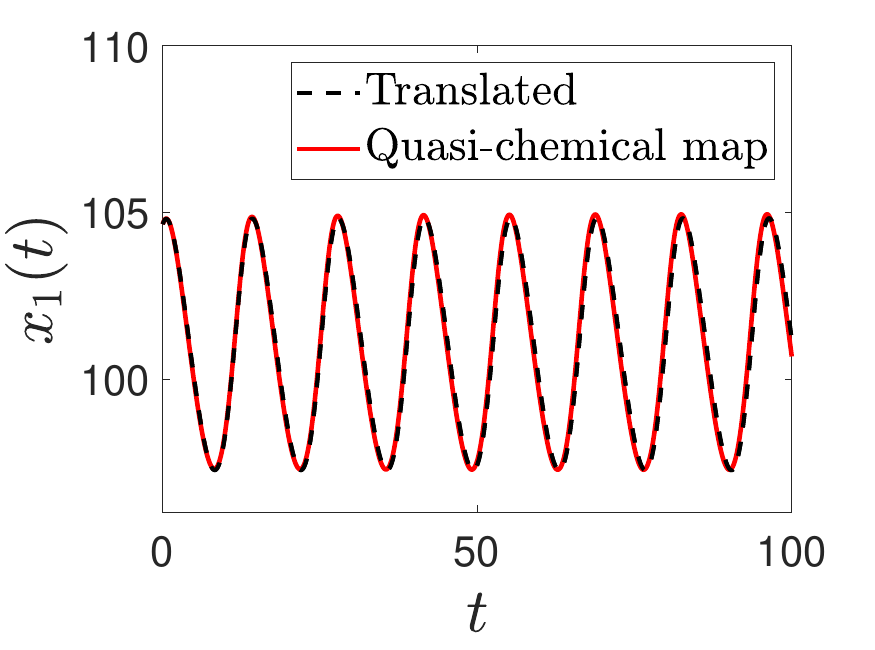}
}
\vskip -0.2cm
\caption{\it{\emph{Application of the 
QCM~(\ref{eq:quasi_chemical}) on DS~(\ref{eq:ex1}).}} 
Panel \emph{(a)} displays state-space
for~$(\ref{eq:ex1_P1c})$ with $\mu = 1/10$,
with the limit cycle shown as the red curve, 
the equilibrium as the blue dot, and vector field
as gray arrows; also shown is the limit cycle 
of~$(\ref{eq:ex1_translated2})$
as the black dashed curve. 
Panel \emph{(b)} displays a similar plot
for~$(\ref{eq:ex1_P1c})$ with $\mu = 1/100$, while panel $(c)$
shows a corresponding time-state space.
Analogous plots are shown in panels \emph{(d)}--\emph{(f)}
for system~$(\ref{eq:ex1_P2c})$.} \label{fig:2}
\end{figure}

Let us now exploit the perturbations to make~(\ref{eq:ex1_P1}) chemical.
To this end, we fix $\mu \in (0,\mu_0)$ and translate 
the variables via $x_1 = (z_1 + 1/\mu)$ and $x_2 = (z_2 + 2/\mu)$,
leading to the qualitatively equivalent CDS
\begin{align}
\frac{\mathrm{d} x_1}{\mathrm{d} t} 
& =  \mu x_1 \left[ \frac{1}{2} \left(x_2 - \frac{2}{\mu} \right) 
+ \frac{1}{8} \left(x_1 - \frac{1}{\mu} \right) \right], 
\nonumber \\
\frac{\mathrm{d} x_2}{\mathrm{d} t} 
& = \frac{\mu}{2} x_2 \left[-\frac{1}{2} \left(x_1 - \frac{1}{\mu} \right)
- \frac{1}{20}  \left(x_1 - \frac{1}{\mu} \right) 
\left(x_2 - \frac{2}{\mu} \right)
+ \frac{1}{20}  \left(x_2 - \frac{2}{\mu} \right)^2 
- \frac{3}{32} \left(x_2 - \frac{2}{\mu} \right) \right]. 
\label{eq:ex1_P1c}
\end{align}
Fixing $\mu = 1/10$ in~(\ref{eq:ex1_P1c}), one obtains 
\begin{align}
\frac{\mathrm{d} x_1}{\mathrm{d} t}
& = \frac{x_1}{10} \left(- \frac{45}{4} + \frac{1}{8} x_1 
+ \frac{1}{2} x_2 \right), \nonumber \\
\frac{\mathrm{d} x_2}{\mathrm{d} t}
& = \frac{x_2}{20} \left(\frac{135}{8} +  \frac{1}{2} x_1 - 
\frac{51}{32} x_2 - \frac{1}{20} x_1 x_2
 + \frac{1}{20} x_2^2 \right).
\label{eq:ex1_P1cb}
\end{align}
In~Figure~\ref{fig:2}(a), we display state-space
for~(\ref{eq:ex1_P1cb}), which contains 
an unstable equilibrium (blue dot) enclosed by a stable
limit cycle (red curve), in qualitative agreement with~(\ref{eq:ex1})
(and, therefore~(\ref{eq:ex1_translated2})).
Note that~(\ref{eq:ex1_P1cb}) can be interpreted
as a correction, via a suitable rescaling of the vector field, 
of $\mathbf{x}$-factorable system~(\ref{eq:xfactorable}), 
which itself fails to qualitatively match the target DS; 
see Figure~\ref{fig:1}(e)--(f).

In~Figure~\ref{fig:2}(a), we also show 
the limit cycle of~(\ref{eq:ex1_translated2})
as dashed black curve, which is in quantitative
disagreement with the limit cycle of~(\ref{eq:ex1_P1cb}).
By design, this mismatch can be made arbitrarily small
if $\mu > 0$ is chosen sufficiently small. For example, 
in Figure~\ref{fig:2}(b), we display state-space 
for~(\ref{eq:ex1_P1c}) with $\mu = 1/100$, 
where one can notice a significantly improved match 
with the target limit cycle.
A corresponding time-state space is shown in Figure~\ref{fig:2}(c), 
demonstrating that the period of oscillations is 
also approximately preserved.

\textbf{Example}. 
Let us again consider DS~(\ref{eq:ex1}) with $\varepsilon = 1/10$, 
and perturb only its first equation:
\begin{align}
\frac{\mathrm{d} z_1}{\mathrm{d} t} 
& = \frac{1}{2} z_2 + \frac{1}{8} z_1
+ \mu z_1 \left[\frac{1}{2} z_2 + \frac{1}{8} z_1 \right], 
\nonumber \\
\frac{\mathrm{d} z_2}{\mathrm{d} t} 
& = -\frac{1}{2} z_1
+ \frac{1}{20}
\left(- z_1 z_2 + z_2^2 \right)
- \frac{3}{32} z_2. \label{eq:ex1_P2}
\end{align}
Translating the variables according to
$x_1 = (z_1 + 1/\mu)$ and $x_2 = (z_2 + 2/\mu)$,
one obtains the qualitatively equivalent system
\begin{align}
\frac{\mathrm{d} x_1}{\mathrm{d} t} 
& =  \mu x_1 \left[ \frac{1}{2} \left(x_2 - \frac{2}{\mu} \right) 
+ \frac{1}{8} \left(x_1 - \frac{1}{\mu} \right) \right], 
\nonumber \\
\frac{\mathrm{d} x_2}{\mathrm{d} t} 
& = -\frac{1}{2} \left(x_1 - \frac{1}{\mu} \right)
- \frac{1}{20}  \left(x_1 - \frac{1}{\mu} \right) 
\left(x_2 - \frac{2}{\mu} \right)
+ \frac{1}{20}  \left(x_2 - \frac{2}{\mu} \right)^2 
- \frac{3}{32} \left(x_2 - \frac{2}{\mu} \right). 
\label{eq:ex1_P2c}
\end{align}
One can readily show that the vector field in the second
equation is chemical for all sufficiently small $\mu > 0$,
making~(\ref{eq:ex1_P2c}) chemical then.
Fixing $\mu = 1/10$ in~(\ref{eq:ex1_P2c}), one obtains CDS
\begin{align}
\frac{\mathrm{d} x_1}{\mathrm{d} t}
& = \frac{x_1}{10} \left(- \frac{45}{4} + \frac{1}{8} x_1 
+ \frac{1}{2} x_2 \right), \nonumber \\
\frac{\mathrm{d} x_2}{\mathrm{d} t}
& =\frac{135}{8} +  \frac{1}{2} x_1 - 
\frac{51}{32} x_2 - \frac{1}{20} x_1 x_2
 + \frac{1}{20} x_2^2.
\label{eq:ex1_P2cb}
\end{align}
In Figure~\ref{fig:2}(d), we display 
 state-space of~(\ref{eq:ex1_P2cb}),
showing an unstable equilibrium and a stable
limit cycle which are in qualitative agreement with~(\ref{eq:ex1_translated2}).
Figure~\ref{fig:2}(e)--(f) displays state and time-state
spaces for~(\ref{eq:ex1_P2c}) when $\mu = 1/100$, 
showing an excellent match with the target DS. 
Let us emphasize an important difference between~(\ref{eq:ex1_P1c})
and~(\ref{eq:ex1_P2c}): the former is cubic, 
while the latter is \emph{quadratic}.

\textbf{Quasi-chemical map}. 
Let us now formalize the map
underlying~(\ref{eq:ex1_P1c}) and~(\ref{eq:ex1_P2c}).

\begin{definition} $($\textbf{Quasi-chemical map}$)$ 
\label{def:quasi_chemical}
Consider \emph{DS}~$(\ref{eq:dyn_components})$. 
Consider also 
\begin{align}
\frac{\mathrm{d} x_i}{\mathrm{d} t} 
& = q_{i} \left(\mathbf{x} - \frac{\mathbf{T}}{\mu} \right)
+ \frac{\mu}{T_i} x_i 
\Big[f_{i} \left(\mathbf{x} - \frac{\mathbf{T}}{\mu} \right) 
- q_{i} \left(\mathbf{x} - \frac{\mathbf{T}}{\mu} \right) \Big], 
\; \; \; i = 1,2,\ldots,N,
\label{eq:quasi_chemical} 
\end{align}
where $\mathbf{T} = (T_1,T_2,\ldots,T_N)^{\top} 
\in \mathbb{R}_{>}^N$ is a fixed vector,
$\mu > 0$ is a free parameter,
and $q_{i}(\mathbf{x} -\mathbf{T}/\mu)
\in \mathbb{P}_n^{\mathcal{C}}(\mathbb{R}^N,\mathbb{R})$
are arbitrary polynomials of degree at most $n$
that are chemical for all sufficiently small $\mu > 0$
and for all $i = 1,2,\ldots,N$; we say that 
$q_{i}(\mathbf{x})$ are \emph{quasi-chemical}.
Then, $\Psi_{\mu} : \mathbb{P}_n(\mathbb{R}^N,\mathbb{R}^N) 
\to \mathbb{P}_{n+1}^{\mathcal{C}}(\mathbb{R}^N,\mathbb{R}^N)$,
mapping the vector field of \emph{DS}~$(\ref{eq:dyn_components})$
to the vector field of \emph{CDS}~$(\ref{eq:quasi_chemical})$
for all sufficiently small $\mu > 0$,
is called a \emph{quasi-chemical map} (\emph{QCM}).
\end{definition}
\noindent
\textbf{Remark}. In Definition~\ref{def:quasi_chemical}, 
we assume that $T_i > 0$. If $T_i = 0$, then we define 
$\mathrm{d} x_i/\mathrm{d} t = 
f_{i}(\mathbf{x} - \mathbf{T}/\mu)$ in~(\ref{eq:quasi_chemical}),
under the additional assumption that 
$f_{i}(\mathbf{x} - \mathbf{T}/\mu)$
with $T_i = 0$ is quasi-chemical. 

\noindent
\textbf{Remark}. In Appendix~\ref{app:quasi_chemical_g}, 
we present a more general definition of the QCM.

The following fundamental lemma is readily provable.
\begin{lemma}
\label{lemma:quasi_chemical}
Under the translation of variables 
$z_i = (x_i - T_i/\mu)$ for any $\mu > 0$, 
\emph{DS}~$(\ref{eq:quasi_chemical})$ becomes
\begin{align}
\frac{\mathrm{d} z_i}{\mathrm{d} t} 
& = f_{i}(\mathbf{z}) + \mu \frac{z_i}{T_i} 
\left[f_{i}(\mathbf{z}) - q_i(\mathbf{z}) \right],
\; \; \; i = 1, 2, \ldots, N.
\label{eq:quasi_chemical_T}
\end{align}
\end{lemma}
\noindent 
\bcol{\textbf{Remark}. Example~(\ref{eq:ex1_P1})--(\ref{eq:ex1_P1c})
is obtained by taking $q_1 = q_2 = 0$, 
$T_1 = 1$ and $T_2 = 2$ in~(\ref{eq:quasi_chemical})--(\ref{eq:quasi_chemical_T}), while example~(\ref{eq:ex1_P2})--(\ref{eq:ex1_P2c})
by taking $q_1 = 0$, $q_2 = f_2$, 
$T_1 = 1$ and $T_2 = 2$.}

Lemma~\ref{lemma:quasi_chemical}
shows that the QCM essentially
just perturbs the target DS~(\ref{eq:dyn_components}),
which allows us to apply the theory of regular perturbations
to deduce a number of results.
We start by showing that the QCM
is a chemical one when it comes to
preserving robust regions in state-space.
To this end, 
for any set $\mathbb{K} \subset \mathbb{R}^N$
and vector $\mathbf{T} \in \mathbb{R}^N$, 
we let $\mathbb{K} + \mathbf{T} 
= \{\mathbf{x} + \mathbf{T} \in \mathbb{R}^N 
\,|\, \mathbf{x} \in \mathbb{K}\}$.

\begin{theorem}$($\textbf{\emph{QCM}}$)$ 
\label{theorem:quasi_chemical}
Assume that \emph{DS}~$(\ref{eq:dyn_components})$ 
is robust in compact set $\mathbb{K} \subset \mathbb{R}^N$. 
Then, for every sufficiently small $\mu > 0$, 
\emph{CDS}~$(\ref{eq:quasi_chemical})$
in $\mathbb{K} + \mathbf{T}/\mu
\subset \mathbb{R}_{\ge}^N$
is qualitatively equivalent to 
\emph{DS}~$(\ref{eq:dyn_components})$ in $\mathbb{K}$.
\end{theorem}

\begin{proof}
The perturbation from~(\ref{eq:quasi_chemical_T})
is given by $(\mathbf{F}(\mathbf{z};\mu))_i
= \mu (z_i/T_i)[f_{i}(\mathbf{z}) - q_i(\mathbf{z})]$,
and satisfies $\mathbf{F} \in C^1$ and 
$\mathbf{F}(\mathbf{z};0) = \mathbf{0}$.
Hence, since~(\ref{eq:dyn_components})
is robust in $\mathbb{K}$, Definition~\ref{def:robustness}
implies existence of $\mu_0 > 0$
such that~(\ref{eq:quasi_chemical_T})
and~(\ref{eq:dyn_components})
are qualitatively equivalent in $\mathbb{K}$
for every $\mu \in (0,\mu_0)$.
By Lemma~\ref{lemma:quasi_chemical}, 
system~(\ref{eq:quasi_chemical})
in $\mathbb{K} + \mathbf{T}/\mu$
is qualitatively equivalent to~(\ref{eq:quasi_chemical_T}) 
in $\mathbb{K}$ for every $\mu > 0$.
Hence, (\ref{eq:quasi_chemical})
in $\mathbb{K} + \mathbf{T}/\mu
\subset \mathbb{R}_{\ge}^N$
is also qualitatively equivalent 
to~(\ref{eq:dyn_components}) in $\mathbb{K}$
for every $\mu \in (0,\mu_0)$.
\end{proof}
\noindent
\textbf{Remark}. 
In Theorem~\ref{theorem:quasi_chemical}, we 
have assumed that DS~(\ref{eq:dyn_components})
is robust, i.e. remains qualitatively equivalent in $\mathbb{K}$
under \emph{every} continuously differentiable perturbation, 
as per Definition~\ref{def:robustness}.
However, the proof shows that Theorem~\ref{theorem:quasi_chemical} holds
under a much less strict assumption:
it is only required that~(\ref{eq:dyn_components})  
remains qualitatively equivalent in $\mathbb{K}$
 under the \emph{special} $(n+1)$-degree polynomial perturbations
$(\mathbf{F}(\mathbf{z};\mu))_i
= \mu (z_i/T_i)[f_{i}(\mathbf{z}) - q_i(\mathbf{z})]$
from~(\ref{eq:quasi_chemical_T}),
which we call \emph{chemical perturbations}.

\bcol{
\noindent
\textbf{Remark}. Theorem~\ref{theorem:quasi_chemical} implies
that hyperbolic solutions (such as hyperbolic equilibria and limit cycles)
are preserved by the QCM, since such solutions are robust~\cite{Wiggins}[Theorem 3.6.4].}

Theorem~\ref{theorem:quasi_chemical} holds under
any choice of the quasi-chemical functions $q_i$ in~(\ref{eq:quasi_chemical}).
Let us now re-state this theorem
under a special choice of $q_i$.
 
\begin{theorem} 
\label{theorem:quasi_chemical2}
Assume that $N$-dimensional $n$-degree 
\emph{DS}~$(\ref{eq:dyn_components})$
is robust in $\mathbb{K} \subset \mathbb{R}^N$.
Then, for every sufficiently small $\mu > 0$,
 $N$-dimensional $(n+1)$-degree \emph{CDS}
\begin{align}
\frac{\mathrm{d} x_i}{\mathrm{d} t} 
& = \frac{\mu}{T_i} x_i f_{i}\left(\mathbf{x} - \frac{\mathbf{T}}{\mu} \right) ,
\; \; \; i = 1,2,\ldots,N,
\label{eq:quasi_chemical2} 
\end{align}
in $\mathbb{K} + \mathbf{T}/\mu
\subset \mathbb{R}_{\ge}^N$
is qualitatively equivalent 
to \emph{DS}~$(\ref{eq:dyn_components})$ in $\mathbb{K}$.
\end{theorem}

\begin{proof}
Fixing the quasi-chemical function $q_i = 0$ for $i = 1, 2, \ldots, N$
in~(\ref{eq:quasi_chemical}) yields~(\ref{eq:quasi_chemical2}).
The statement of the theorem then follows from
Theorem~\ref{theorem:quasi_chemical}.
\end{proof}
\noindent
\textbf{Remark}. Theorem~\ref{theorem:quasi_chemical2}
shows that the $\mathbf{x}$-factorable map~\cite{Samardzija} 
qualitatively preserves robust features if they are translated
uniformly, $T_1 = T_2 = \ldots = T_N$, i.e. the direction
of translation is important; however, 
even in this case, time is distorted.

\subsection{Quantitative properties}
In this section, we augment Theorem~\ref{theorem:quasi_chemical}
with more precise, \emph{quantitative}, results 
for some special robust solutions, namely 
solutions over finite time-intervals, 
hyperbolic equilibria and hyperbolic limit cycles.
For more complicated solutions, 
quantitative results are more difficult to obtain;
in such cases, trapping regions can be useful
(see Appendix~\ref{app:background}). 
In what follows, we write $\mathbf{x}(\mu)
= \mathbf{y} + \mathcal{O}(\mu)$
if $\|\mathbf{x}(\mu) - \mathbf{y}\|
\le c \mu$ for all sufficiently small $\mu > 0$,
where $c > 0$ is a $\mu$-independent constant.

\begin{theorem}$($\textbf{\emph{QCM}: Quantitative properties}$)$ 
\label{theorem:quasi_chemical_EL}
\begin{enumerate}
\item[\emph{(i)}] \textbf{Finite time-intervals}. 
Let $\mathbf{y}(t) = \mathbf{y}(t;\mathbf{y}_0)$
be a solution of~$(\ref{eq:dyn_components})$ 
for all $t \in [0, t_0]$ for some finite $t_0 > 0$. 
Then, for all sufficiently small $\mu$
$(\ref{eq:quasi_chemical})$ has a unique solution 
$\mathbf{x}(t;\mu) = \mathbf{x}(t;\mu,\mathbf{y}_0 + \mathbf{T}/\mu)$ 
for all $t \in [0, t_0]$. Furthermore,
$ \mathbf{x}(t;\mu) = (\mathbf{y}(t) + \mathbf{T}/\mu) 
+ \mathcal{O}(\mu)$
uniformly for $t \in [0, t_0]$.
\item[\emph{(ii)}] \textbf{Equilibria}. 
Let $\mathbf{y}^*$ be a 
hyperbolic equilibrium of~$(\ref{eq:dyn_components})$.
Then, for all sufficiently small $\mu$
$(\ref{eq:quasi_chemical})$ has a 
unique hyperbolic equilibrium $\mathbf{x}^*(\mu)$
in a neighborhood of $(\mathbf{y}^* + \mathbf{T}/\mu)$.
In particular, $\mathbf{x}^*(\mu) 
= (\mathbf{y}^* + \mathbf{T}/\mu) + \mathcal{O}(\mu)$.
Equilibrium $\mathbf{x}^*(\mu)$ is qualitatively 
equivalent to $\mathbf{y}^*$. In particular, 
eigenvalues $\lambda_{1,\mu}, \ldots, \lambda_{N,\mu}$ associated 
to $\mathbf{x}^*(\mu)$ can be ordered so that 
$\lim_{\mu \to 0} \lambda_{i,\mu} = \lambda_{i,0}$
for all $i = 1,\ldots, N$, 
where $\lambda_{1,0}, \ldots, \lambda_{N,0}$ are 
the eigenvalues associated to $\mathbf{y}^*$.
\item[\emph{(iii)}] \textbf{Limit cycles}. 
Let $\mathbf{y}_{\tau_0}$ be a 
hyperbolic limit cycle of~$(\ref{eq:dyn_components})$
with period $\tau_0 > 0$.
Then, for all sufficiently small $\mu$
$(\ref{eq:quasi_chemical})$ has a 
unique hyperbolic limit cycle $\mathbf{x}_{\tau_{\mu}}$
with period $\tau_{\mu} > 0$
in a neighborhood of $(\mathbf{y}_{\tau_0} + \mathbf{T}/\mu)$.
In particular, $\mathbf{x}_{\tau_{\mu}}(t) 
= (\mathbf{y}_{\tau_{0}}(t) + \mathbf{T}/\mu) + \mathcal{O}(\mu)$ 
uniformly over finite time-intervals, 
and $\tau_{\mu} = \tau_0 + \mathcal{O}(\mu)$.
Limit cycle $\mathbf{x}_{\tau_{\mu}}$ is qualitatively 
equivalent to $\mathbf{y}_{\tau_0}$. In particular, 
characteristic exponents 
$\rho_{1,\mu}, \ldots, \rho_{N,\mu}$ associated 
to $\mathbf{x}_{\tau_{\mu}}$ can be ordered so that 
$\lim_{\mu \to 0} \rho_{i,\mu} = \rho_{i,0}$
for all $i = 1, \ldots, N$, 
where $\rho_{1,0}, \ldots, \rho_{N,0}$ are 
the characteristic exponents associated to $\mathbf{y}_{\tau_0}$.
\item[\emph{(iv)}] \textbf{Trapping regions}. 
Let $\mathbb{S}_0 = \{\mathbf{y} \in \mathbb{R}^N | V(\mathbf{y}) \le r
\textrm{ for some }  r > 0\}$ be a 
trapping region for~$(\ref{eq:dyn_components})$. 
Then, $\mathbb{S}_{\mu} = \{\mathbf{x} \in \mathbb{R}^N | 
V(\mathbf{x} - \mathbf{T}/\mu) \le r\}$ is a trapping region 
for~$(\ref{eq:quasi_chemical})$ for all sufficiently small $\mu > 0$.
\end{enumerate}
\end{theorem}

\begin{proof}
To prove statement~(i), we use the fact that,
if $\mathbf{y}(t) = \mathbf{y}(t;\mathbf{y}_0)$
is a solution of~$(\ref{eq:dyn_components})$ 
for all $t \in [0, t_0]$, 
then for all sufficiently small $\mu > 0$
system~(\ref{eq:quasi_chemical_T}) has a unique solution 
$\mathbf{z}(t;\mu) = \mathbf{z}(t;\mu,\mathbf{y}_0)$ 
for all $t \in [0, t_0]$, and 
$\mathbf{z}(t;\mu) = \mathbf{y}(t) + \mathcal{O}(\mu)$
uniformly for $t \in [0, t_0]$~\cite{Coddington}[Chapter 1].
Statement~(i) then follows from Lemma~\ref{lemma:quasi_chemical}
via the translation of variables $\mathbf{x} = \mathbf{z} + \mathbf{T}/\mu$.
Statements~(ii) and~(iii) are proved analogously,
using e.g. the theory
from~\cite{Coddington}[Chapter 14].
To prove statement~(iv), we use Definition~\ref{def:trapping_region}
to deduce that $\mathbf{f}(\mathbf{y}) \cdot \mathbf{n}(\mathbf{y}) < 0$
for all $\mathbf{y} \in \mathbb{R}^N$ on the boundary 
$V(\mathbf{y}) = r$ with outward-pointing normal $\mathbf{n}(\mathbf{y})
\in \mathbb{R}^N$,
where $\cdot$ denotes the dot product.
Therefore, the vector field of~(\ref{eq:quasi_chemical_T}) 
satisfies $\mathbf{f}(\mathbf{z}) \cdot \mathbf{n}(\mathbf{z})
+ \mathcal{O}(\mu) < 0$ on the boundary $V(\mathbf{z}) = r$ 
for every $\mu > 0$ sufficiently small;
hence, $\mathbb{S}_0$ is also a trapping region for~(\ref{eq:quasi_chemical_T}).
Statement~(iv) then follows from Lemma~\ref{lemma:quasi_chemical}.
\end{proof}

Theorem~\ref{theorem:quasi_chemical_EL}(i)
shows that the corresponding solutions of CDS~(\ref{eq:quasi_chemical})
and target DS~(\ref{eq:dyn_components}) 
stay, up to a translation, within $\mathcal{O}(\mu)$-distance
for any finite time-interval. This alignment 
of trajectories holds without a time re-scaling,
i.e. the QCM preserves the time-scale (time-parametrization) along the solutions;
this is in contrast to the time-warp 
map from Section~\ref{sec:chemical_maps}.
Theorem~\ref{theorem:quasi_chemical_EL}(ii)
shows that the corresponding hyperbolic equilibria 
of~(\ref{eq:quasi_chemical}) and~(\ref{eq:dyn_components})
are also within $\mathcal{O}(\mu)$-distance.
Furthermore, the corresponding eigenvalues
are arbitrarily close; consequently,
the QCM preserves, not only stability,
but also the type of hyperbolic equilibria.
Analogous conclusions are reached for hyperbolic limit cycles
in Theorem~\ref{theorem:quasi_chemical_EL}(iii);
in addition, the periods of oscillations of the 
corresponding limit cycles are $\mathcal{O}(\mu)$-close.
Finally, Theorem~\ref{theorem:quasi_chemical_EL}(iv)
shows that the QCM preserves trapping regions.

\subsection{Bifurcations}
\label{sec:bifurcations}
Theorem~\ref{theorem:quasi_chemical} is 
applicable when the target DS~(\ref{eq:dyn})
is robust in $\mathbb{K} \subset \mathbb{R}^N$. 
In this section, we consider the case when~(\ref{eq:dyn}) 
is not robust in $\mathbb{K}$.
To this end, we now explicitly display
dependence of vector fields on some parameters,
which we always assume is continuously differentiable.
In particular, to indicate dependence of~(\ref{eq:dyn}) 
on some parameters $\boldsymbol{\beta}\in \mathbb{R}^P$, we write
\begin{align}
\frac{\mathrm{d} \mathbf{y}}{\mathrm{d} t} & = 
\mathbf{f}(\mathbf{y};\boldsymbol{\beta}), 
\; \; \; \textrm{where } 
\mathbf{f}(\mathbf{y};\cdot) \in C^1(\mathbb{R}^P,\mathbb{R}^N).
\label{eq:dynp} 
\end{align}
Consider also the perturbed DS~(\ref{eq:dyn3}), written as
\begin{align}
\frac{\mathrm{d} \mathbf{z}}{\mathrm{d} t} & = 
\mathbf{f}(\mathbf{z};\boldsymbol{\beta})
+
\mathbf{F}(\mathbf{z};\boldsymbol{\beta},\mu),
\; \; \; \textrm{where } 
\mathbf{F} \in C^1(\mathbb{R}^N \times \mathbb{R}^P \times \mathbb{R}_{\ge},\mathbb{R}^N)
\textrm{ and } \mathbf{F}(\mathbf{z};\boldsymbol{\beta},0) = \mathbf{0}.
\label{eq:dyn3p} 
\end{align}
Let us fix the parameters to
$\boldsymbol{\beta} = \boldsymbol{\beta}^* \in \mathbb{R}^P$, 
and assume that then~(\ref{eq:dynp}) is not robust
in compact set $\mathbb{K} \subset \mathbb{R}^N$.
This means that, for some perturbations $\mathbf{F}$, 
DSs~(\ref{eq:dyn3p}) and~(\ref{eq:dynp}) at 
$\boldsymbol{\beta} = \boldsymbol{\beta}^*$ 
are qualitatively different in $\mathbb{K}$, 
and~(\ref{eq:dynp}) is said to undergo a \emph{bifurcation}~\cite{Perko,Kuznetsov,Wiggins}.
Even though~(\ref{eq:dyn3p}) at $\boldsymbol{\beta} = \boldsymbol{\beta}^*$
is in general qualitatively different from~(\ref{eq:dynp})
at $\boldsymbol{\beta} = \boldsymbol{\beta}^*$,
(\ref{eq:dyn3p}) at a slightly different parameter value,
say $\boldsymbol{\beta} = \boldsymbol{\beta}^{**}  \in \mathbb{R}^P$, 
may be qualitatively equivalent to~(\ref{eq:dynp})
at $\boldsymbol{\beta} = \boldsymbol{\beta}^*$.
In this case, we say that the underlying bifurcation is robust. 

\begin{definition} $($\textbf{Robust bifurcation}$)$ 
\label{def:robust_bifurcation}
Assume that \emph{DS}~$(\ref{eq:dynp})$
at $\boldsymbol{\beta} = \boldsymbol{\beta}^* \in \mathbb{R}^P$
is not robust and undergoes a bifurcation $\mathcal{B}$ 
in $\mathbb{K} \subset \mathbb{R}^N$.
Assume furthermore that for every $\varepsilon > 0$, 
and for every $\mathbf{F} \in C^1$
with $\mathbf{F}(\mathbf{z};\boldsymbol{\beta},0) = \mathbf{0}$,
there exists $\mu_0 > 0$ such that for every $\mu \in (0,\mu_0)$
there exists $\boldsymbol{\beta}^{**}  \in \mathbb{R}^P$, 
with $\|\boldsymbol{\beta}^{**} - \boldsymbol{\beta}^{*}\| < \varepsilon$,
such that~$(\ref{eq:dyn3p})$ at 
$\boldsymbol{\beta} = \boldsymbol{\beta}^{**}$
and~$(\ref{eq:dynp})$ at 
$\boldsymbol{\beta} = \boldsymbol{\beta}^*$
are qualitatively equivalent in $\mathbb{K}$.
Then, \emph{DS}~$(\ref{eq:dynp})$ is said to undergo
at $\boldsymbol{\beta} = \boldsymbol{\beta}^*$
a \emph{robust bifurcation} $\mathcal{B}$ in $\mathbb{K}$.
\end{definition} 
\bcol{ \noindent
\textbf{Remark}. Generic bifurcations
(such saddle-node, generic Hopf and homoclinic bifurcations) 
are robust~\cite{Perko,Kuznetsov,Wiggins}.}

We now show that the QCM preserves robust bifurcations.

\begin{theorem}$($\textbf{\emph{QCM}: Bifurcations}$)$ 
\label{theorem:quasi_chemical_bifurcation}
Assume that \emph{DS}~$(\ref{eq:dynp})$
at $\boldsymbol{\beta} = \boldsymbol{\beta}^* \in \mathbb{R}^P$
undergoes a robust bifurcation $\mathcal{B}$ 
in $\mathbb{K} \subset \mathbb{R}^N$.
Then, for every sufficiently small $\mu > 0$ 
there exists $\boldsymbol{\beta}^{**} \in \mathbb{R}^P$,
arbitrarily close to $\boldsymbol{\beta}^*$,
such that \emph{CDS}~$(\ref{eq:quasi_chemical})$
at $\boldsymbol{\beta} = \boldsymbol{\beta}^{**}$
undergoes bifurcation $\mathcal{B}$ 
in $\mathbb{K} + \mathbf{T}/\mu
\subset \mathbb{R}_{\ge}^N$.
\end{theorem}

\begin{proof}
The perturbation from~(\ref{eq:quasi_chemical_T})
is given by $(\mathbf{F}(\mathbf{z};\boldsymbol{\beta},\mu))_i
= \mu (z_i/T_i) \left[f_{i}(\mathbf{z};\boldsymbol{\beta}) - 
q_i(\mathbf{z};\boldsymbol{\beta}) \right]$,
and satisfies $\mathbf{F} \in C^1$ and 
$\mathbf{F}(\mathbf{z};\boldsymbol{\beta},0) = \mathbf{0}$.
Therefore, using the assumption that~(\ref{eq:dynp})
at $\boldsymbol{\beta} = \boldsymbol{\beta}^*$
undergoes robust bifurcation $\mathcal{B}$ 
in $\mathbb{K} \subset \mathbb{R}^N$,
Definition~\ref{def:robust_bifurcation} 
implies that for every $\varepsilon > 0$ there exists 
$\mu_0 > 0$ such that for all $\mu \in (0,\mu_0)$
there exists $\boldsymbol{\beta}^{**}  \in \mathbb{R}^P$, 
with $\|\boldsymbol{\beta}^{**} - \boldsymbol{\beta}^{*}\| < \varepsilon$,
such that~(\ref{eq:quasi_chemical_T}) at 
$\boldsymbol{\beta} = \boldsymbol{\beta}^{**}$
undergoes bifurcation $\mathcal{B}$ in $\mathbb{K} \subset \mathbb{R}^N$.
Since bifurcations are preserved
 under translations of dependent variables,
statement of the theorem then follows from
Lemma~\ref{lemma:quasi_chemical}.
\end{proof}
\noindent
\textbf{Remark}. 
Analogous to Theorem~\ref{theorem:quasi_chemical},
Theorem~\ref{theorem:quasi_chemical_bifurcation}
holds for all bifurcations that are
robust with respect to the chemical perturbations;
robustness with respect to other types of
perturbations is irrelevant for the theorem.

\bcol{
Analogous to Theorem~\ref{theorem:quasi_chemical2}, 
we now prove that $(n+1)$-degree CDSs can realize
every robust bifurcation displayed by $n$-degree DSs
of the same dimension.

\begin{theorem} 
	\label{theorem:quasi_chemical2_bifurcations}
	Assume that $N$-dimensional $n$-degree 
	\emph{DS}~$(\ref{eq:dynp})$
	at $\boldsymbol{\beta} = \boldsymbol{\beta}^*$
	undergoes a robust bifurcation $\mathcal{B}$ 
	in $\mathbb{K} \subset \mathbb{R}^N$.
	Then, for every sufficiently small $\mu > 0$,
	$N$-dimensional $(n+1)$-degree \emph{CDS}
	\begin{align}
		\frac{\mathrm{d} x_i}{\mathrm{d} t} 
		& = \frac{\mu}{T_i} x_i f_{i}\left(\mathbf{x} - \frac{\mathbf{T}}{\mu};\boldsymbol{\beta}\right) ,
		\; \; \; i = 1,2,\ldots,N,
		\label{eq:quasi_chemical2_bifurcation} 
	\end{align}
	at some $\boldsymbol{\beta} = \boldsymbol{\beta}^{**}$, 
	arbitrarily close to $\boldsymbol{\beta}^*$,
	undergoes bifurcation $\mathcal{B}$ 
	in $\mathbb{K} + \mathbf{T}/\mu
	\subset \mathbb{R}_{\ge}^N$.
\end{theorem}

\begin{proof}
	Choosing $q_i = 0$ in~(\ref{eq:quasi_chemical}), one obtains~(\ref{eq:quasi_chemical2_bifurcation}),
	and the theorem follows from
	Theorem~\ref{theorem:quasi_chemical_bifurcation}.
\end{proof}
}

\section{Quasi-chemical map: Applications} 
\label{sec:quasi_chemical_applications}
In this section, we use the results 
from Section~\ref{sec:quasi_chemical}
to transform some DSs with desired behaviors to CDSs.
These tasks could simply be performed by 
using the special QCM~(\ref{eq:quasi_chemical2}).
However, in what follows, we instead focus on 
the more general QCM~(\ref{eq:quasi_chemical}),
which has a key advantage:
the quasi-chemical functions $q_i$ may be chosen
to reduce the number of higher-degree terms 
in the resulting CDSs, or even 
preserve the degree of the target DSs.
We state most of the results in this section in terms of CRNs; 
see Appendix~\ref{app:background}
for more details on how to construct 
CRNs for a given CDS.

\subsection{CRN with arbitrary many limit cycles} 
\label{sec:oscillations}
Let us consider the two-dimensional DS
\begin{align}
\frac{\mathrm{d} y_1}{\mathrm{d} t} 
& = y_2 + \varepsilon (\beta_0 y_1 + \beta_1 y_1^3 +
 \ldots + \beta_{n} y_1^{2 n + 1}), \nonumber \\
\frac{\mathrm{d} y_2}{\mathrm{d} t} & = - y_1.
\label{eq:HO_1}
\end{align}
Under a suitable choice of parameters
$\beta_0,\beta_1,\ldots,\beta_n$, and for all
$\varepsilon > 0$ sufficiently small,
(\ref{eq:HO_1}) has $n$
nested hyperbolic limit cycles~\cite{Eckweiler,Perko}.
We also prove this fact, and provide explicit 
expressions for a suitable parameter choice,
in Appendix~\ref{app:linearoscillator}. 
We now map DS~(\ref{eq:HO_1}) into a CDS
of the same degree. To this end, 
we denote two irreversible reactions 
$n X \xrightarrow[]{\alpha} m X$
and 
$m X \xrightarrow[]{\beta} n X$
as a single reversible reaction 
$n X \xrightleftharpoons[\beta]{\alpha} m X$.

\begin{theorem}$($\textbf{\emph{CRN} with arbitrary many limit cycles}$)$  \label{theorem:first_order}
Consider the \emph{CRN}
\begin{align}
\varnothing &\xrightleftharpoons[\alpha_1]{\alpha_0} X_1, 
\; \; 2 X_1 \xrightleftharpoons[\alpha_3]{\alpha_2} 3 X_1, 
\; \; \ldots, 
\; \; 2 n X_1 \xrightleftharpoons[\alpha_{2 n + 1}]{\alpha_{2 n}} (2 n + 1) X_1, 
\nonumber \\
X_2 & \xrightarrow[]{\alpha_{2 n + 2}} X_1 + 2 X_2, 
\; \; X_1 + X_2 \xrightarrow[]{\alpha_{2 n + 3}} X_1. 
\label{eq:HO_1_C}
\end{align}
Let $0 < r_1< r_2 < \ldots < r_n$ be arbitrary real numbers.
Then, there exist coefficients $\alpha_i 
= \alpha_i(r_1,r_2,\ldots,r_N)$ for $i = 0, 1, \ldots, 2 N + 3$ 
such that \emph{CRN}~$(\ref{eq:HO_1_C})$ has 
$n$ nested hyperbolic limit cycles.
The $i$th limit cycle is arbitrarily close to
a circle of radius $r_i$,
and has period arbitrarily close to $2 \pi$. 
The $n$th limit cycle is stable, 
and the limit cycles alternate in stability.
\end{theorem}

\begin{proof}
Let us fix $\varepsilon > 0$ to a sufficiently small value, 
and $\beta_0,\beta_1,\ldots,\beta_n$
according to~(\ref{eq:HO_1_b}) from Appendix~\ref{app:linearoscillator}.
It then follows from Lemma~\ref{lemma:first_order}
that~(\ref{eq:HO_1}) has $n$ nested hyperbolic limit cycles;
the $i$th limit cycle is arbitrarily close to 
the $2 \pi$-periodic circle of radius $r_i$ centered at the origin.
Furthermore, the $n$th limit cycle is stable, 
and the limit cycles alternative in stability.

Let us apply on~(\ref{eq:HO_1}) 
the QCM~(\ref{eq:quasi_chemical})
with $T_1 = T_2 = 1$, $q_1 = f_1$ and $q_2 = 0$, leading to
\begin{align}
\frac{\mathrm{d} x_1}{\mathrm{d} t} 
& = \left(x_2 - \frac{1}{\mu} \right) 
+ \varepsilon \sum_{i = 0}^{n} \beta_i 
\left(x_1 - \frac{1}{\mu} \right)^{2 i + 1} 
=  \sum_{i = 0}^{2 n + 1} (-1)^{i} \alpha_{i} x_1^{i} 
+ \alpha_{2 n + 2} x_2, \nonumber \\
\frac{\mathrm{d} x_2}{\mathrm{d} t} & = 
- \mu x_2  \left(x_1 - \frac{1}{\mu} \right)
= \alpha_{2 n + 2} x_2 - \alpha_{2 n + 3} x_1 x_2,
\label{eq:HO_1_RRE}
\end{align}
where the coefficients
$\alpha_0,\alpha_1,\ldots,\alpha_{2 n + 3}$,
are stated explicitly as~(\ref{eq:HO_1_C_coeff}) 
in Appendix~\ref{app:linearoscillator}.
Since $\beta_n < 0$, the zero-degree term in the first 
equation from~(\ref{eq:HO_1_RRE})
is dominated by $1/\mu^{2 n + 1} > 0$ for all sufficiently small $\mu > 0$;
therefore, (\ref{eq:HO_1_RRE}) is then a CDS.
A CRN induced by~(\ref{eq:HO_1_RRE}) is given by~(\ref{eq:HO_1_C});
see Definition~\ref{def:CRN}.
In particular, we fuse 
$X_2 \xrightarrow[]{1} X_1 + X_2$
and 
$X_2 \xrightarrow[]{1} 2 X_2$
into
$X_2 \xrightarrow[]{1} X_1 + 2 X_2$, 
as per Definition~\ref{def:fused}.
Statement of the theorem then follows 
from Theorem~\ref{theorem:quasi_chemical_EL}(iii).
\end{proof}

\begin{figure}[!htbp]
\vskip 0.0cm
\leftline{\hskip 2.5cm (a) \hskip  6.5cm (b)}
\centerline{
\hskip -1.0cm
\includegraphics[width=0.35\columnwidth]{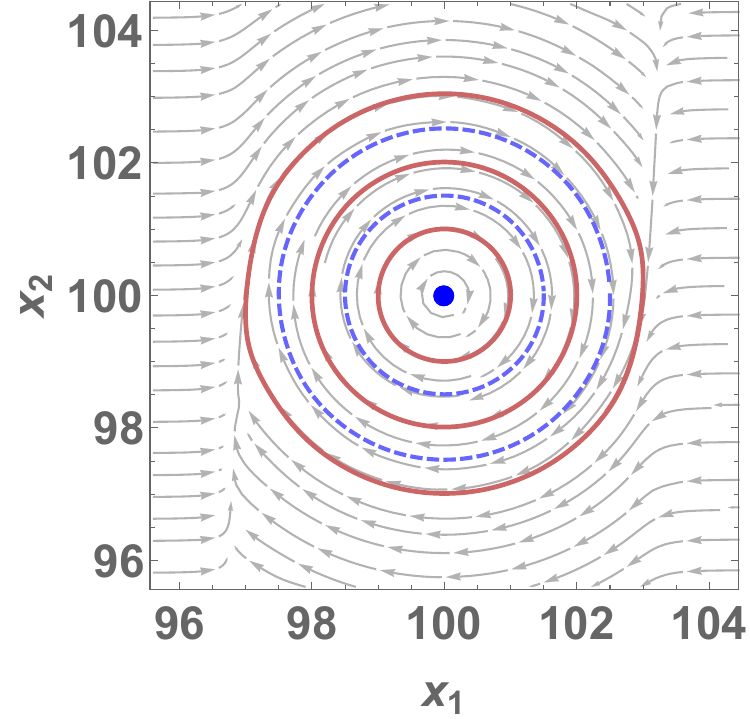}
\hskip 1.5cm
\includegraphics[width=0.3\columnwidth]{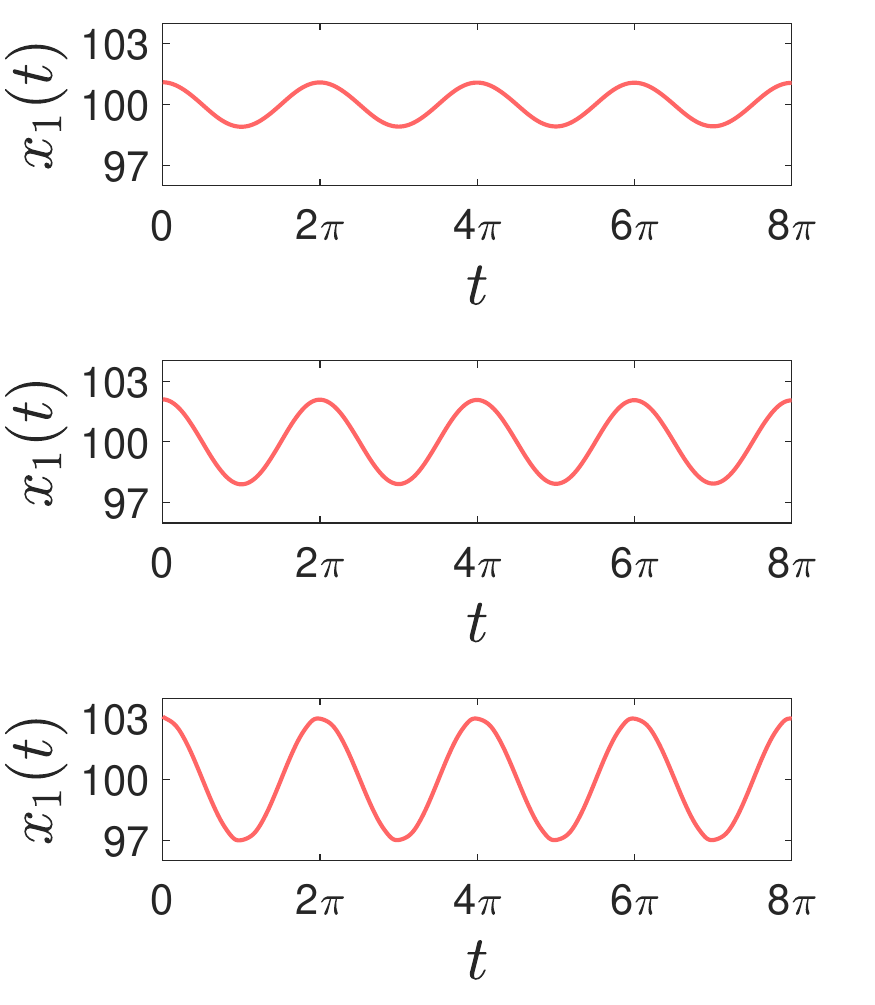}
}
\vskip -0.2cm
\caption{\it{\emph{CDS with arbitrary many limit cycles.}} 
Panel \emph{(a)} displays state-space
for \emph{CDS}~$(\ref{eq:HO_1_RREn5})$,
corresponding to the \emph{CRN}~$(\ref{eq:HO_1_Cn5})$,
with parameters~$(\ref{eq:HO_1_RREn5coeff})$, 
$\varepsilon = 10^{-4}$ and $\mu = 10^{-2}$.
The three stable and two unstable limit cycles are 
respectively shown as solid red and dashed blue curves, 
and the enclosed equilibrium as the blue dot.
Each of the three stable limit cycles
is shown in time-state space 
in one of the sub-panels of panel \emph{(b)}.} 
\label{fig:3}
\end{figure}

\textbf{Example}. 
Consider the CRN~(\ref{eq:HO_1_C}) with $n = 5$:
\begin{align}
\varnothing &\xrightleftharpoons[\alpha_1]{\alpha_0} X_1, 
\; \; 2 X_1 \xrightleftharpoons[\alpha_3]{\alpha_2} 3 X_1, 
\; \; 4 X_1 \xrightleftharpoons[\alpha_5]{\alpha_4} 5 X_1,
\; \; 6 X_1 \xrightleftharpoons[\alpha_7]{\alpha_6} 7 X_1,
\; \; 8 X_1 \xrightleftharpoons[\alpha_9]{\alpha_8} 9 X_1,
\; \; 10 X_1 \xrightleftharpoons[\alpha_{11}]{\alpha_{10}} 11 X_1, \nonumber \\
X_2 & \xrightarrow[]{\alpha_{12}} X_1 + 2 X_2, 
\; \; X_1 + X_2 \xrightarrow[]{\alpha_{13}} X_1, 
\label{eq:HO_1_Cn5}
\end{align}
with the CDS~(\ref{eq:HO_1_RRE}) given by
\begin{align}
\frac{\mathrm{d} x_1}{\mathrm{d} t} 
& = \left(x_2 - \frac{1}{\mu} \right) 
+ \varepsilon \sum_{i = 0}^{5} \beta_i 
\left(x_1 - \frac{1}{\mu} \right)^{2 i + 1}, \nonumber \\
\frac{\mathrm{d} x_2}{\mathrm{d} t} & = 
- \mu x_2  \left(x_1 - \frac{1}{\mu} \right).
\label{eq:HO_1_RREn5}
\end{align}
Let us impose that~(\ref{eq:HO_1_RREn5})
has $5$ limit cycles, close to circles
with radii $r_1 = 1$, $r_2 = 3/2$, 
$r_3 = 2$, $r_4 = 5/2$ and $r_5 = 3$. 
Then, suitable coefficients $\beta_0,\beta_1,\ldots,\beta_5$ are given
\begin{align}
\beta_0 & = \frac{2025}{4}, \; \; 
\beta_1 = -\frac{5307}{4}, \; \; 
\beta_2 = 1039, \; \; 
\beta_3 = - \frac{11652}{35}, \; \; 
\beta_4 = \frac{320}{7}, \; \; 
\beta_5 = - \frac{512}{231}.
\label{eq:HO_1_RREn5coeff}
\end{align}
Coefficients~(\ref{eq:HO_1_RREn5coeff}) can
be obtained directly from~(\ref{eq:HO_1_b})
in Appendix~\ref{app:linearoscillator}; alternatively, 
one can indirectly impose that a suitable polynomial
has the desired radii as roots, see~(\ref{eq:I}) 
in Appendix~\ref{app:linearoscillator}.
In Figure~\ref{fig:3}(a), we display state-space
for CDS~(\ref{eq:HO_1_RREn5}) 
with parameters~(\ref{eq:HO_1_RREn5coeff}), 
 $\varepsilon = 10^{-4}$ and $\mu = 10^{-2}$.
As desired, the system has three stable limit cycles, 
shown as red solid curves, and two unstable ones,
shown as dashed blue curves, all of which are approximately circular.
In Figure~\ref{fig:3}(b), we show three sub-panels, 
each displaying one of the stable limit cycles in 
 time-state space; one can notice that each of
the limit cycles is approximately $2 \pi$-periodic.

\subsection{CRN with chaos}
\label{sec:chaos} 
In this section, we design a CDS
whose long-time dynamics is neither
an equilibrium nor a limit cycle.
To this end, let us consider the
Lorenz system~\cite{Lorenz}:
\begin{align}
\frac{\mathrm{d} y_1}{\mathrm{d} t} & = - 10 y_1 + 10 y_2, \nonumber \\
\frac{\mathrm{d} y_2}{\mathrm{d} t} & = 28 y_1 - y_2 - y_1 y_3, \nonumber \\
\frac{\mathrm{d} y_3}{\mathrm{d} t} & = - \frac{8}{3} y_3 + y_1 y_2. 
\label{eq:Lorenz}
\end{align}
In Figure~\ref{fig:4}(a), we show the butterfly-like 
chaotic attractor of~(\ref{eq:Lorenz}), called the Lorenz attractor, 
projected into $(y_1,y_3)$-space,
while a corresponding $(t,y_1)$-space
is shown in Figure~\ref{fig:4}(b).

\begin{figure}[!htbp]
\vskip 0.2cm
\leftline{\hskip 2.2cm (a) \hskip  6.9cm (c)}
\centerline{
\hskip 0.0cm
\includegraphics[width=0.35\columnwidth]{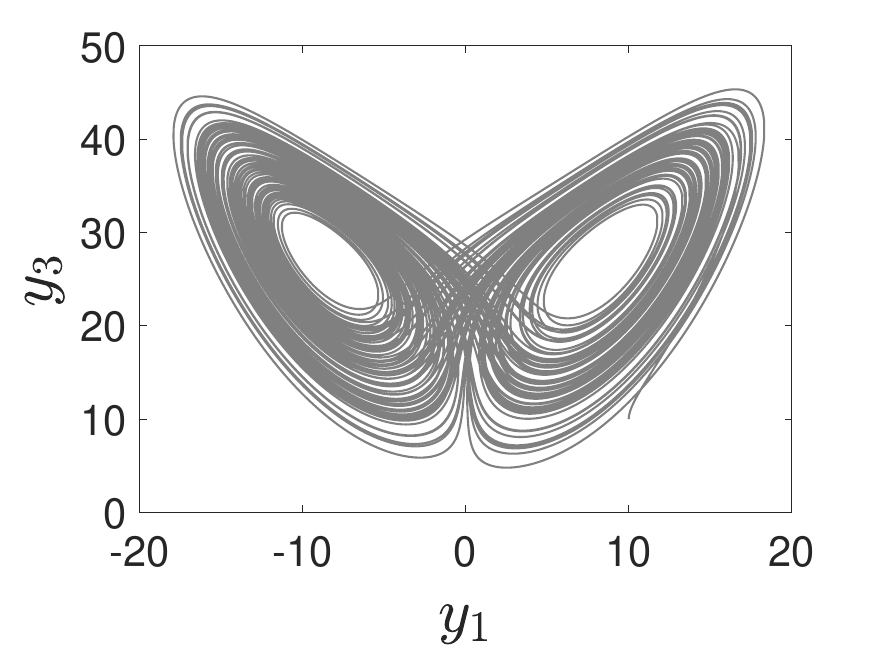}
\hskip 1.5cm
\includegraphics[width=0.35\columnwidth]{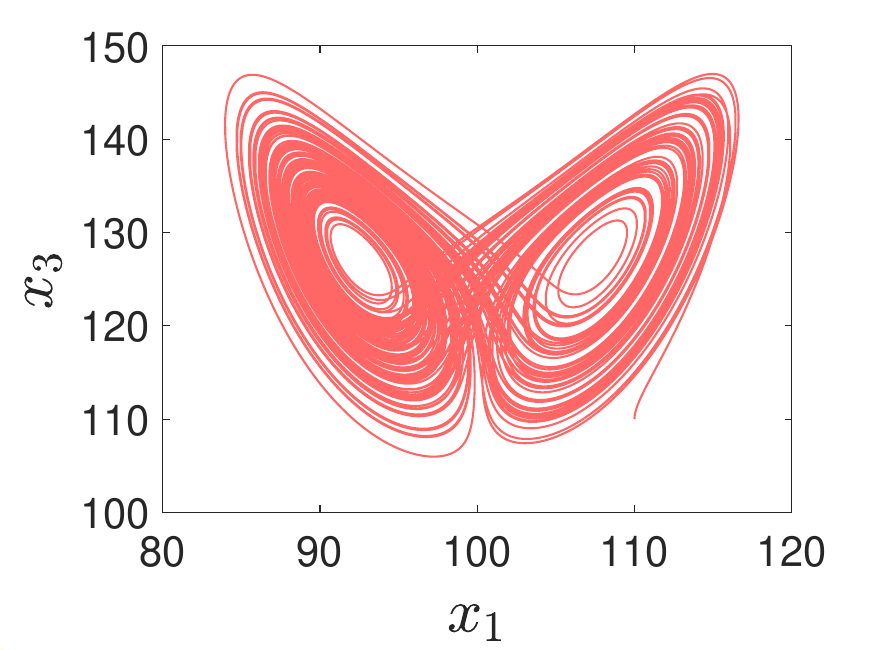}
}
\leftline{\hskip 2.2cm (b) \hskip  6.9cm (d)}
\centerline{
\hskip 0.0cm
\includegraphics[width=0.35\columnwidth]{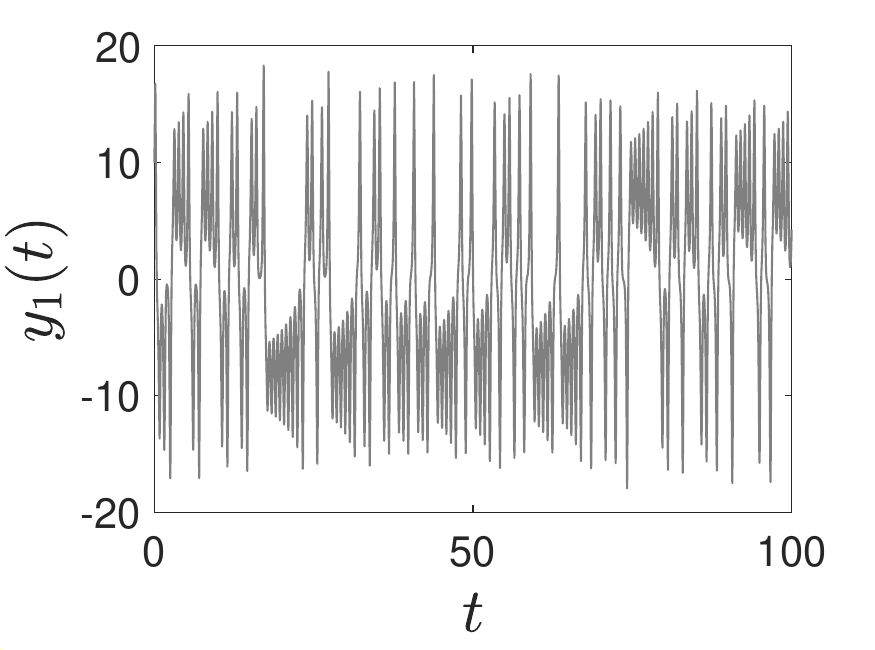}
\hskip 1.5cm
\includegraphics[width=0.35\columnwidth]{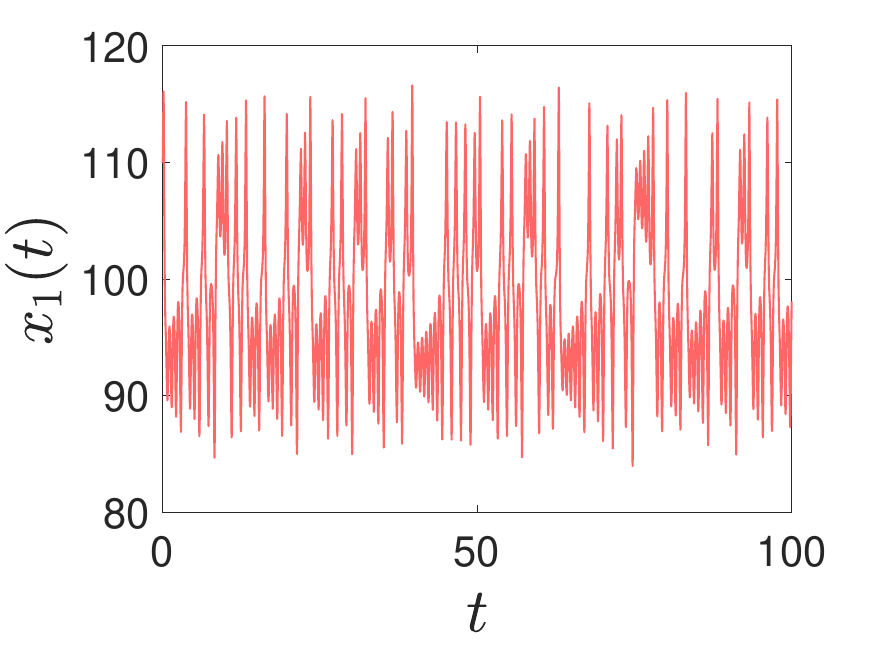}
}
\vskip -0.2cm
\caption{\it{\emph{CDS with Lorenz attractor.} 
Panels \emph{(a)} and \emph{(b)} respectively display $(y_1,y_3)$-
and $(t,y_1)$-space for the Lorenz system~$(\ref{eq:Lorenz})$
with initial condition $y_1(0) = y_2(0) = y_3(0) = 10$.
Analogous plots are shown in panels \emph{(c)} and \emph{(d)}
for \emph{CDS}~$(\ref{eq:Lorenz_chemical})$
with initial condition $x_1(0) = x_2(0) = x_3(0) = 10 + 100$.}} 
\label{fig:4}
\end{figure}

\begin{theorem}$($\textbf{\emph{CRN} with Lorenz attractor}$)$ 
\label{theorem:chaos1}
Consider the \emph{CRN}
\begin{align}
X_1 &\xrightarrow[]{\alpha_1} \varnothing, 
\; \; \; X_2 \xrightarrow[]{\alpha_2} X_1 + X_2, 
\; \; \;  \varnothing \xrightleftharpoons[\alpha_4]{\alpha_3} X_2, 
\; \; \;  X_1 + X_2 \xrightarrow[]{\alpha_5} X_1 + 2 X_2, 
\; \; \;  X_2 + X_3 \xrightarrow[]{\alpha_6} 2 X_2,  \nonumber \\
X_1 & + X_2 + X_3 \xrightarrow[]{\alpha_7} X_1 + 2 X_3, 
\; \; \;  \varnothing \xrightarrow[]{\alpha_8} X_3,
\; \; \;  X_3 \xrightarrow[]{\alpha_9} 2 X_3,
\; \; \;  X_1 + X_3 \xrightarrow[]{\alpha_{10}} X_1.
\label{eq:Lorenz_CRN}
\end{align}
Assume that the rate coefficients are given by 
\begin{align}
\alpha_1 & = 10, \; \; \; 
\alpha_2 = 10,\; \; \; 
\alpha_3 = \frac{1}{\mu}, \; \; \; 
\alpha_4 = \frac{1}{\mu} + 29, \; \; \; 
\alpha_5 = 1 + 28 \mu,  \; \; \; 
\alpha_6 = 1, \nonumber \\
\alpha_7 & = \mu, \; \; \; 
\alpha_8 = \frac{8}{3 \mu}, \; \; \; 
\alpha_{9} = \frac{1}{\mu} - \frac{8}{3}, \; \; \; 
\alpha_{10} = 1. 
\label{eq:Lorenz_coeff}
\end{align}
Then, for every sufficiently small $\mu > 0$, 
\emph{CRN}~$(\ref{eq:Lorenz_CRN})$
has a Lorenz attractor.
\end{theorem}

\begin{proof}
Let us apply on~(\ref{eq:Lorenz}) 
the QCM~(\ref{eq:quasi_chemical})
with $T_1 = T_2 = T_3 = 1$, 
$q_1 = f_1$, $q_2 = -y_2$ and $q_3 = - (8/3) y_3$, 
leading to 
\begin{align}
\frac{\mathrm{d} x_1}{\mathrm{d} t} 
& = - 10 x_1 + 10 x_2, \nonumber \\
\frac{\mathrm{d} x_2}{\mathrm{d} t} 
& = - \left(x_2 - \frac{1}{\mu} \right) 
+ \mu x_2 \left[ 28 \left(x_1 - \frac{1}{\mu} \right) - 
\left(x_1 - \frac{1}{\mu} \right) \left(x_3 - \frac{1}{\mu} \right) \right], 
\nonumber \\
\frac{\mathrm{d} x_3}{\mathrm{d} t} & = 
- \frac{8}{3} \left(x_3 - \frac{1}{\mu} \right) 
+ \mu x_3 \left[\left(x_1 - \frac{1}{\mu} \right) 
\left(x_2 - \frac{T}{\mu} \right) \right].
\label{eq:Lorenz_chemical_g}
\end{align}
DS~(\ref{eq:Lorenz_chemical_g}) is chemical for all $\mu > 0$,
and induces CRN~(\ref{eq:Lorenz_CRN})
with coefficients~(\ref{eq:Lorenz_coeff}).
Since the Lorenz attractor is robust
under small perturbations~\cite{Tucker},
statement of the theorem follows from
Lemma~\ref{lemma:quasi_chemical}.
\end{proof}

Let us fix $\mu = 1/100$ in~(\ref{eq:Lorenz_coeff});
the CDS~(\ref{eq:Lorenz_chemical_g}) 
induced by~(\ref{eq:Lorenz_CRN}) then reads
\begin{align}
\frac{\mathrm{d} x_1}{\mathrm{d} t} 
& = - 10 x_1 + 10 x_2, \nonumber \\
\frac{\mathrm{d} x_2}{\mathrm{d} t} 
& = 100 - 129 x_2 
+ \frac{32}{25} x_1 x_2 + x_2 x_3 - \frac{1}{100} x_1 x_2 x_3, 
\nonumber \\
\frac{\mathrm{d} x_3}{\mathrm{d} t} & = 
\frac{800}{3} + \frac{292}{3} x_3 
- x_1 x_3 - x_2 x_3 + \frac{1}{100} x_1 x_2 x_3. 
\label{eq:Lorenz_chemical}
\end{align}
State-space and time-state space for~(\ref{eq:Lorenz_chemical})
are respectively shown in Figure~\ref{fig:4}(c) and~(d), 
which compare well with Figure~\ref{fig:4}(a)--(b).
One can observe that the Lorenz system~(\ref{eq:Lorenz}) 
and its chemical counterpart~(\ref{eq:Lorenz_chemical}) both 
evolve on a similar time-scale for the given time-interval, which is consistent with Theorem~\ref{theorem:quasi_chemical_EL}(i).
Note that, while Lorenz attractor itself is robust
with respect to perturbations, the underlying trajectories
in the time-state space are sensitive; for this reason, 
the two trajectories from Figures~\ref{fig:4}(b) and~(d)
differ after some time. \bcol{Let us also note that 
trapping regions for the Lorenz attractor are also preserved, 
as per Theorem~\ref{theorem:quasi_chemical_EL}(iv).}
Finally, we note that chemical Lorenz system~(\ref{eq:Lorenz_chemical})
is cubic. In Appendix~\ref{app:chaos}, we present a candidate
chaotic CDS which is quadratic.

\subsection{CRN with homoclinic bifurcation}
\label{sec:CRN_homoclinic}
Consider the two-dimensional DS with a super-critical 
homoclinic bifurcation, given by~\cite{Sandstede,Me1}
\begin{align}
\frac{\mathrm{d} y_1}{\mathrm{d} t} 
& = \left(\beta -\frac{4}{5} \right) y_1 + y_2 - \frac{6}{5} y_1 y_2 
+ \frac{3}{2} y_2^2, \nonumber \\
\frac{\mathrm{d} y_2}{\mathrm{d} t} 
& = y_1 - \frac{4}{5} y_2 - \frac{4}{5} y_2^2.
\label{eq:homoclinic}
\end{align}
In particular, when $\beta = 0$, DS~(\ref{eq:homoclinic})
has a saddle equilibrium at the origin, whose stable and unstable 
manifolds coincide, forming a tear-shaped homoclinic loop; 
see Figure~\ref{fig:6}(b), where we also display
an unstable focus inside the homoclinic loop, 
and a stable node in the first quadrant.
While the saddle is robust, 
the homoclinic loop is not - when 
$\beta$ is slightly varied, 
the saddle manifolds split, 
and a bifurcation occurs.
More precisely, when $\beta < 0$,
the unstable manifold ``undershoots" and 
forms a stable limit cycle around the focus,
as shown in Figure~\ref{fig:6}(a). 
On the other hand, when $\beta > 0$, 
the unstable manifold ``overshoots", 
and is connected to the stable node, 
see Figure~\ref{fig:6}(c).

\begin{figure}[!htbp]
\vskip 0.3cm
\leftline{
\hskip 0.6cm (a) (\ref{eq:homoclinic}) with $\beta = -1/10$
\hskip  1.3cm (b) (\ref{eq:homoclinic}) with $\beta = 0$
\hskip  2.2cm (c) (\ref{eq:homoclinic}) with $\beta = 1/10$} 
\vskip 0.2cm
\centerline{
\hskip 0.0cm
\includegraphics[width=0.3\columnwidth]{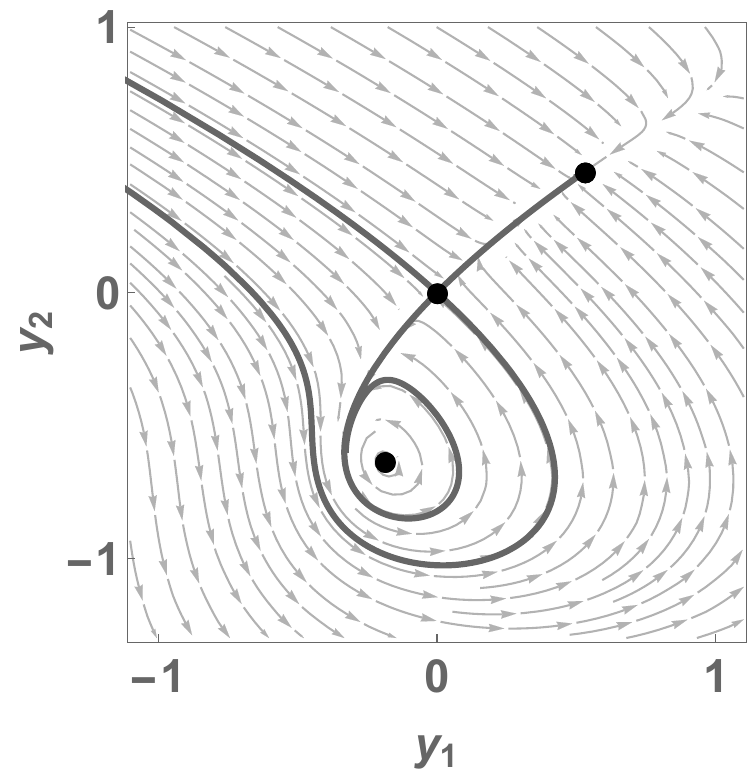}
\hskip 0.5cm
\includegraphics[width=0.3\columnwidth]{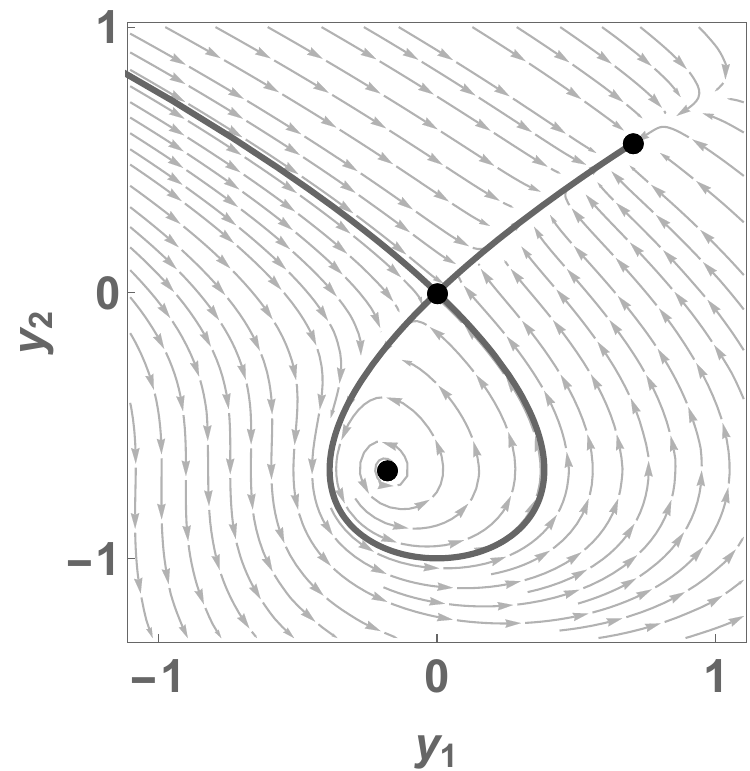}
\hskip 0.5cm
\includegraphics[width=0.3\columnwidth]{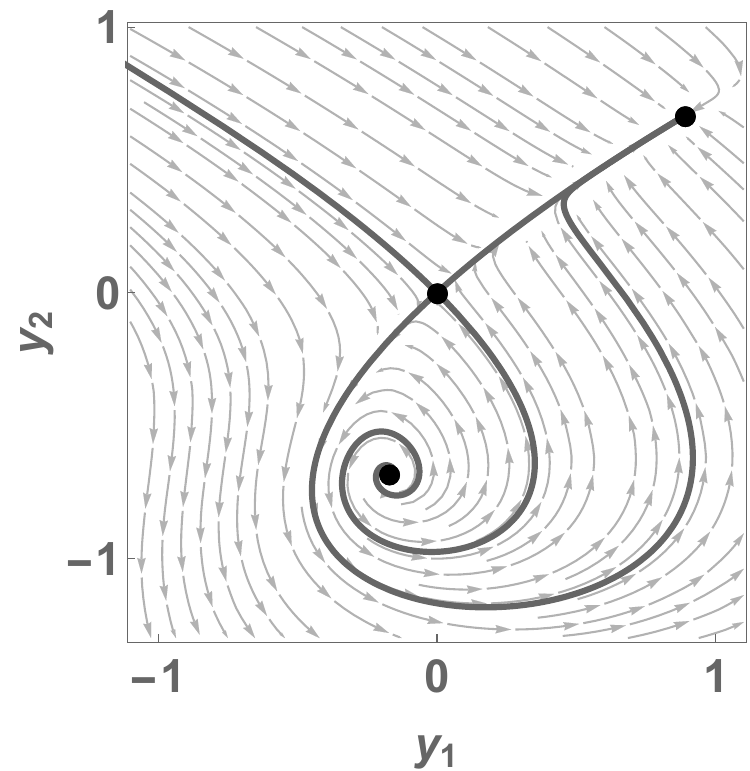}
}
\leftline{
\hskip 0.6cm (d) (\ref{eq:homoclinic_g_fixed}) with $\beta = -1/10$
\hskip  1.3cm (e) (\ref{eq:homoclinic_g_fixed}) with $\beta = 50/5879$
\hskip  1.1cm (f) (\ref{eq:homoclinic_g_fixed}) with $\beta = 1/10$} 
\vskip 0.2cm
\centerline{
\hskip 0.0cm
\includegraphics[width=0.32\columnwidth]{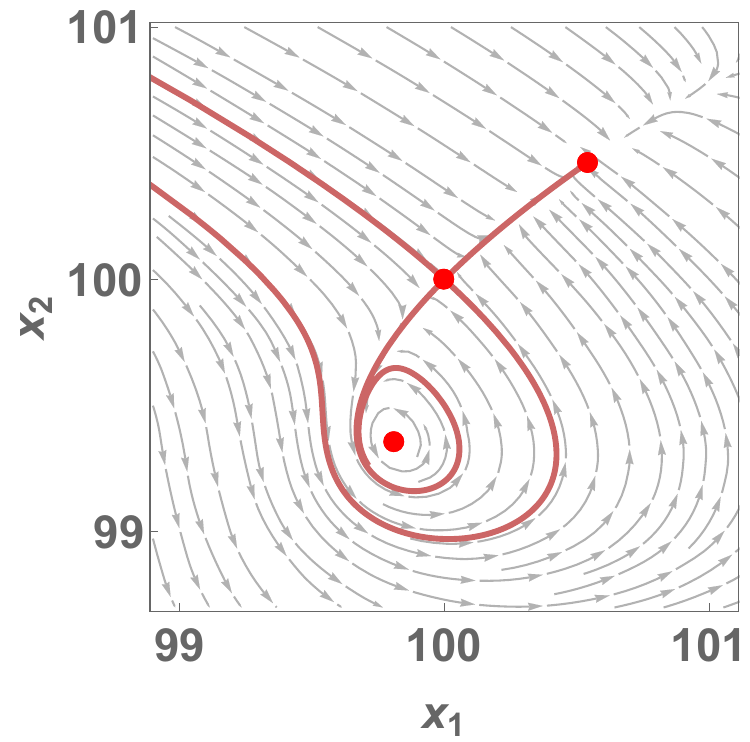}
\hskip 0.1cm
\includegraphics[width=0.32\columnwidth]{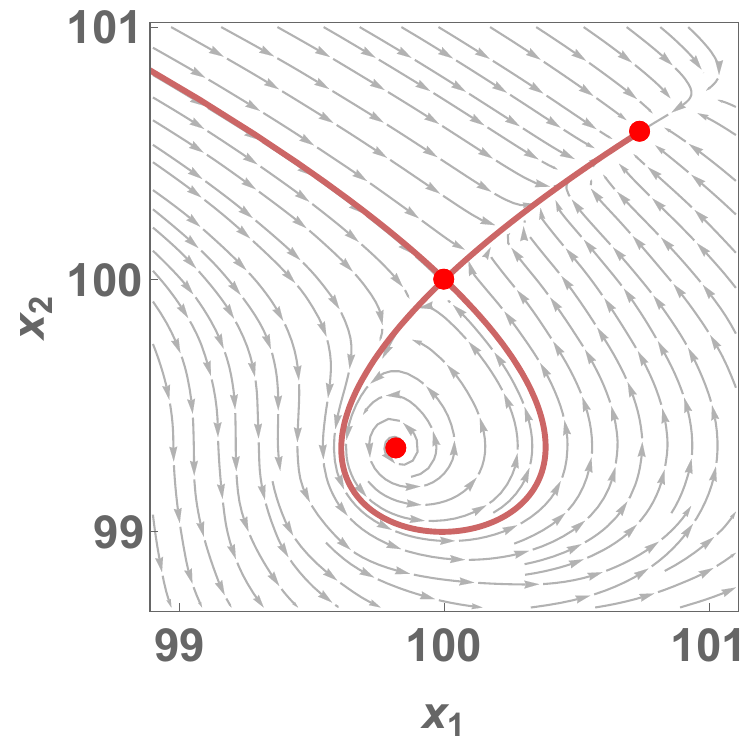}
\hskip 0.1cm
\includegraphics[width=0.32\columnwidth]{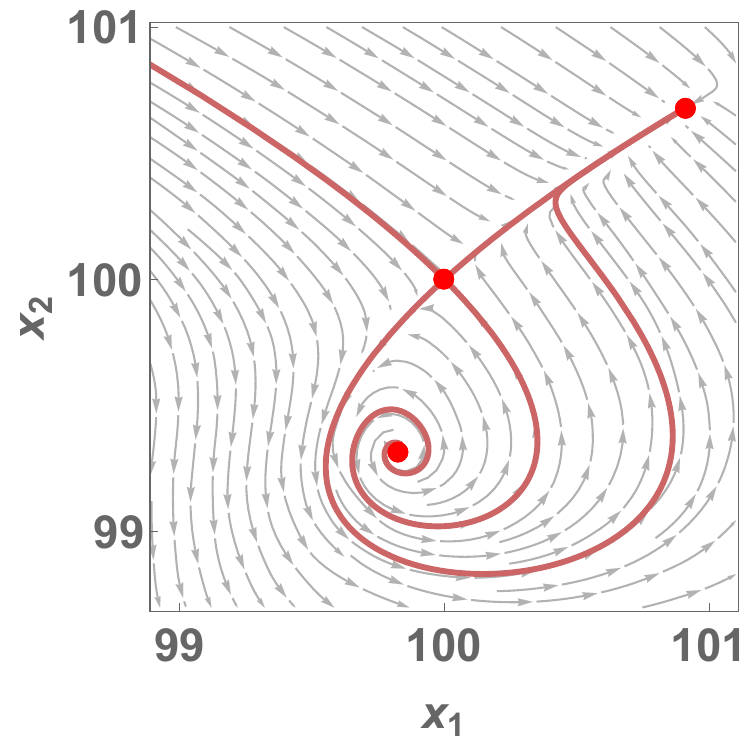}
}
\vskip -0.2cm
\caption{\it{\emph{CDS undergoing homoclinic bifurcation.} 
Panels \emph{(a)}--\emph{(c)} display state-space
for \emph{DS}~$(\ref{eq:homoclinic})$ for three different values
of $\beta$, where at $\beta = 0$ a generic homoclinic 
bifurcation takes place.
Analogous plots are shown for \emph{CDS}~$(\ref{eq:homoclinic_g_fixed})$
in panels \emph{(d)}--\emph{(f)}.}} 
\label{fig:6}
\end{figure}

\begin{theorem}$($\textbf{\emph{CRN} with homoclinic bifurcation}$)$ 
\label{theorem:homoclinic}
Consider the \emph{CRN}
\begin{align}
\varnothing &\xrightarrow[]{\alpha_1} X_1, 
\; \; \; X_1 \xrightarrow[]{\alpha_2} 2 X_1, 
\; \; \; 2 X_1 \xrightarrow[]{\alpha_3} 3 X_1, 
\; \; \; X_1 + X_2 \xrightarrow[]{\alpha_4} X_2, 
\; \; \; 2 X_1 + X_2 \xrightarrow[]{\alpha_5} X_1 + X_2, \nonumber \\ 
X_1 & + 2 X_2 \xrightarrow[]{\alpha_6} 2 X_1 + 2 X_2, 
\; \; \; \varnothing \xrightleftharpoons[\alpha_8]{\alpha_7} X_2,
\; \; \;  X_1 + X_2 \xrightarrow[]{\alpha_9} X_1 + 2 X_2, 
\; \; \;  2 X_2 \xrightleftharpoons[\alpha_{11}]{\alpha_{10}} 3 X_2. 
\label{eq:homoclinic_CRN}
\end{align}
Assume that the rate coefficients are given by 
\begin{align}
\alpha_1 & = \frac{4}{5 \mu} - \frac{\beta}{\mu}, \; \; \; 
\alpha_2 = \frac{3}{10 \mu} - \frac{9}{5} + \beta,\; \; \; 
\alpha_3 = \frac{6}{5}, \; \; \; 
\alpha_4 = \frac{9}{5} - \mu, \; \; \; 
\alpha_5 = \frac{6}{5} \mu, \nonumber \\
\alpha_6 & = \frac{3}{2} \mu, \; \; \; 
\alpha_7 = \frac{4}{5 \mu}, \; \; \; 
\alpha_8 = \frac{4}{5 \mu} + \frac{9}{5}, \; \; \;  
\alpha_9 = \mu, \; \; \;  
\alpha_{10} = \frac{8}{5}, \; \; \;  
\alpha_{11} = \frac{4}{5} \mu.
\label{eq:homoclinic_coeff}
\end{align}
Then, for every sufficiently small $\mu > 0$, 
\emph{CRN}~$(\ref{eq:homoclinic_CRN})$
undergoes a super-critical homoclinic bifurcation
at saddle $(1/\mu,1/\mu) \in \mathbb{R}_{>}^{2}$
at some parameter value $\beta = \beta(\mu)$ arbitrarily close to zero.
\end{theorem}

\begin{proof}
Let us apply on~(\ref{eq:homoclinic})
the QCM~(\ref{eq:quasi_chemical})
with $T_1 = T_2 = 1$, $q_1 = (\beta - 4/5) y_1$ and $q_2 = -(4/5) y_2$, 
leading to the DS
\begin{align}
\frac{\mathrm{d} x_1}{\mathrm{d} t} 
& = \left(\beta -\frac{4}{5} \right)
\left(x_1 - \frac{1}{\mu} \right)
+ \mu x_1 \left[\left(x_2 - \frac{1}{\mu} \right) 
- \frac{6}{5} \left(x_1 - \frac{1}{\mu} \right) \left(x_2 - \frac{1}{\mu} \right)
+ \frac{3}{2} \left(x_2 - \frac{1}{\mu} \right)^2 \right], \nonumber \\
\frac{\mathrm{d} x_2}{\mathrm{d} t} 
& = - \frac{4}{5} \left(x_2 - \frac{1}{\mu} \right) 
+ \mu x_2 \left[ \left(x_1 - \frac{1}{\mu} \right) 
- \frac{4}{5} \left(x_2 - \frac{1}{\mu} \right) ^2 \right],
\label{eq:homoclinic_g}
\end{align}
which is chemical for every $\beta \le 4/5$,
and for every $\mu > 0$; an induced CRN is given by~(\ref{eq:homoclinic_CRN})
with coefficients~(\ref{eq:homoclinic_coeff}).
DS~(\ref{eq:homoclinic}) at $\beta = 0$ undergoes a super-critical 
homoclinic bifurcation which is generic~\cite{Sandstede,Me1}, 
and therefore robust;
hence, the statement of the theorem 
follows from~Theorem~\ref{theorem:quasi_chemical_bifurcation}.
\end{proof}

Let us fix $\mu = 1/100$ in~(\ref{eq:homoclinic_coeff}), 
so that the resulting CDS~(\ref{eq:homoclinic_g}) becomes
\begin{align}
\frac{\mathrm{d} x_1}{\mathrm{d} t} 
& = \left(80 - 100 \beta \right)
+ \left(\frac{141}{5} + \beta \right) x_1
+ \frac{6}{5} x_1^2
- \frac{179}{100} x_1 x_2
- \frac{3}{250} x_1^2 x_2
+ \frac{3}{200} x_1 x_2^2, \nonumber \\
\frac{\mathrm{d} x_2}{\mathrm{d} t} 
& = 80 - \frac{409}{5} x_2 + \frac{1}{100} x_1 x_2
+ \frac{8}{5} x_2^2 - \frac{1}{125} x_2^3.
\label{eq:homoclinic_g_fixed}
\end{align}
It is found numerically that~(\ref{eq:homoclinic_g_fixed})
undergoes the homoclinic bifurcation 
at $\beta \approx 50/5879$; state-space 
is displayed in Figure~\ref{fig:6}(d)--(e),
and agrees well with that of 
the target DS~(\ref{eq:homoclinic}).
Note that CDS~(\ref{eq:homoclinic_g_fixed}) is cubic.
In Appendix~\ref{app:homoclinic},
we apply a generalized QCM to design a quadratic CDS 
undergoing super-critical homoclinic bifurcation.

\subsection{Hilbert's 16th problem in chemistry}
\label{sec:Hilbert}

In 1900, David Hilbert presented a list of 
mathematical problems, 16th of which asks for the maximum number 
of limit cycles in two-dimensional 
$n$-degree polynomial DSs~\cite{Hilbert1}.
This number is denoted by $H(n)$, and remains unknown 
for all $n \ge 2$~\cite{Hilbert2}. 
One can pose a similar question in the chemical context:
what is the maximum number of limit cycles 
in two-dimensional $n$-degree CDSs?
We denote this number by $H_{\mathcal{C}}(n)$.
Let us show that chemical Hilbert number 
$H_{\mathcal{C}}(n)$ is sandwiched
between the two successive general Hilbert numbers
$H(n-1)$ and $H(n)$.

\bcol{
\begin{theorem} \label{theorem:Hilbert}
$H(n-1) \le H_{\mathcal{C}}(n) \le H(n)$
for every integer $n \ge 2$.
\end{theorem}

\begin{proof}
CDSs are a subset of polynomial DSs; 
hence, $H_{\mathcal{C}}(n) \le H(n)$.
If $H(n-1) < \infty$,
then there exists a two-dimensional 
$(n-1)$-degree DS with $H(n-1)$ 
hyperbolic limit cycles~\cite{Hilbert_Robust}[Theorem 2(a)].
Therefore, by Theorem~\ref{theorem:quasi_chemical2}
and~\ref{theorem:quasi_chemical_EL}(iii), 
there exists a two-dimensional $n$-degree 
CDS with $H(n-1)$ limit cycles, 
implying that $H_{\mathcal{C}}(n) \ge H(n-1)$.
On the other hand, if $H(n-1) = \infty$, 
then there exists a two-dimensional 
$(n-1)$-degree DS with arbitrary many
hyperbolic limit cycles~\cite{Hilbert_Robust}[Theorem 2(b)].
Therefore, by Theorem~\ref{theorem:quasi_chemical2}
and~\ref{theorem:quasi_chemical_EL}(iii), 
there exists a two-dimensional $n$-degree 
CDS with arbitrary many limit cycles, 
implying that $H_{\mathcal{C}}(n) = H(n-1)$. 
\end{proof}
}

\begin{corollary}
$H_{\mathcal{C}}(3) \ge 4$, $H_{\mathcal{C}}(4) \ge 13$, 
$H_{\mathcal{C}}(5) \ge 28$, and
$H_{\mathcal{C}}(n) =  \mathcal{O}(n^2 \ln(n))$
as $n \to \infty$. 
\end{corollary}

\begin{proof}
Statement of the corollary follows
from Theorem~\ref{theorem:Hilbert} and 
the fact that $H(2) \ge 4$~\cite{H2}, 
$H(3) \ge 13$~\cite{H3}, $H(4) \ge 28$~\cite{H4},
and $H(n) =  \mathcal{O}(n^2 \ln(n))$ as $n \to \infty$~\cite{Bound}.
\end{proof}

\section{Discussion}
\label{sec:discussion}
\bcol{In this paper, we have investigated some \emph{chemical maps}
- special maps that transform general dynamical systems (DSs) 
into chemical dynamical systems (CDSs), 
while preserving desired features.
We have introduced a novel map, 
called the \emph{quasi-chemical map}
(QCM), which systematically perturbs any given polynomial DS
in such a way that it becomes a CDS under sufficiently large translations. 
The QCM increases degree of DSs by at most one,
and qualitatively preserves 
robust dynamical and bifurcation structures; 
see respectively Theorems~\ref{theorem:quasi_chemical}
and~\ref{theorem:quasi_chemical_bifurcation}.
In particular, this map qualitatively preserves
hyperbolic equilibria and limit cycles, 
some temporal properties, such as periods of oscillations,
and trapping regions;
see Theorem~\ref{theorem:quasi_chemical_EL}.
Using the QCM, we proved that every $n$-degree DS
with robust dynamical or bifurcation structures
can be mapped to an $(n+1)$-degree CDS
of the same dimension displaying the same structures;
see Theorems~\ref{theorem:quasi_chemical2} and~\ref{theorem:quasi_chemical2_bifurcations}.

These properties can make the QCM
more favorable than some alternative ones, such as
the time-warp map~\cite{Time_change1,Time_change2} 
and the $\mathbf{x}$-factorable map~\cite{Samardzija}.
In particular, while the time-warp map preserves
both robust and non-robust dynamical and bifurcation structures, 
it does not preserve temporal properties; 
furthermore, it can significantly increase degree of DSs,
see e.g. Appendix~\ref{app:timechange}.
On the other hand, the $\mathbf{x}$-factorable map
increases the degree by only one, but can fail 
to preserve even hyperbolic equilibria,
as demonstrated in Section~\ref{sec:chemical_maps}.
In a special case, the QCM can be interpreted
as a correction of the $\mathbf{x}$-factorable map; 
see Theorem~\ref{theorem:quasi_chemical2}.
However, let us note that, as opposed to the 
time-warp map, the QCM relies on a small parameter
and sufficiently large translations of desired dynamical features,
which may pose some experimental challenges~\cite{DNA}.
}

In Section~\ref{sec:quasi_chemical_applications}, 
we have applied the QCM to
design some CDSs with exotic behaviors.
In particular, in Section~\ref{sec:oscillations},
we have designed CDS~(\ref{eq:HO_1_C}) with 
arbitrary many limit cycles. Another such CDS 
is designed in~\cite{Radek},
with the goal of providing
a lower bound on the maximum number
of stable limit cycles in two-dimensional 
$n$-degree CDSs.
This Hilbert-like number is denoted by $C(n)$,
and the bound obtained in~\cite{Radek}
is $C(n)  \ge \lfloor (n+2)/6 \rfloor$, 
where $\lfloor \cdot \rfloor$ 
denotes the integer part of a real number.
CDS~(\ref{eq:HO_1_C}) from Section~\ref{sec:oscillations}
of this paper predicts a better bound, 
namely  $C(n) \ge \lfloor (n+1)/4 \rfloor$.

In Section~\ref{sec:chaos}, we have designed 
a cubic CDS with Lorenz attractor. 
A similar system has been put forward in~\cite{Samardzija};
however, since the $\mathbf{x}$-factorable map is applied
and some of the variables are translated differently, 
existence of the chaotic attractor is not rigorously justified.
On the other hand, a CDS with Lorenz attractor
rigorously justified is presented in~\cite{Janos_chaos}.
To design it, the authors use the time-warp map;
consequently, the CDS is quartic and the time-scale 
at which the Lorenz attractor is tracked is distorted.
In Section~\ref{sec:CRN_homoclinic}, we have presented 
a cubic CDS undergoing a super-critical homoclinic bifurcation;
in Appendix~\ref{app:homoclinic}, we have also
designed a quadratic CDS undergoing the same bifurcation.
This problem has also been considered 
using the $\mathbf{x}$-factorable map
and numerical simulations in~\cite{Me1}. 
\bcol{More broadly, bifurcations in two-dimensional 
quadratic CDSs have also been 
studied in~\cite{Bimolecular_bifurcations}; 
see also~\cite{Degn–Harrison} for a bifurcation 
analysis of a chemical system.}

In this paper, we have focused on polynomial DSs.
Assume now that we are given a \emph{non-polynomial} DS
with continuously differentiable vector field;
assume also that this DS is robust in a desired state-space region.
One approach to design a dynamically similar CDS is
to first map the non-polynomial DS
to a qualitatively equivalent polynomial DS;
this can be done in a dimension-preserving manner 
using the Weierstrass approximation theorem, 
see Appendix~\ref{app:nonpoly}.
Once a suitable polynomial DS is found, 
one can then apply the QCM to obtain a qualitatively equivalent CDS.
Another, more direct, approach to map robust non-polynomial
DSs to CDSs is by using special 
chemical systems that execute suitable
algorithms from the theory of artificial neural networks~\cite{RNCRN}.

The chemical maps presented in this paper are
\emph{dimension-preserving}, i.e. they 
introduce no additional dependent variables.
In a follow-up paper, we will focus on
\emph{dimension-expanding} chemical maps 
that introduce auxiliary variables 
to achieve chemical systems.
As part of future research, we also
formulate the following fundamental problem:
\emph{Find classes of polynomial \emph{DS}s that can be mapped to
\emph{CDS}s with the same dimension and degree,
and with desired dynamical or bifurcation structures preserved}.
This general problem is also relevant
to Hilbert's 16th problem in chemistry.
In particular, we have proved that 
$H_{\mathcal{C}}(n) \ge H(n-1)$
in Theorem~\ref{theorem:Hilbert} of Section~\ref{sec:Hilbert},
using the special QCM~(\ref{eq:quasi_chemical2}),
which increases degree of DSs by one.
A natural question is whether the more general QCM~(\ref{eq:quasi_chemical})
can be applied to suitable polynomial DSs in a degree-preserving manner
to probe if e.g. $H_{\mathcal{C}}(n) = H(n)$ for some $n \ge 2$.
As a step forward with the more general problem, we present the following result.

\begin{theorem} \label{theorem:special_case}
Consider the class of $N$-dimensional $n$-degree polynomial \emph{DSs}
\begin{align}
\frac{\mathrm{d} y_i}{\mathrm{d} t} 
& = f_{i}(\mathbf{y}) + \beta_{i} y_i^{n},
\; \; \; \; \; i = 1, 2, \ldots, N,
\label{eq:special_case} 
\end{align}
where $f_{i} \in \mathbb{P}_{n-1}(\mathbb{R}^N,\mathbb{R})$ 
are arbitrary polynomials 
of degree at most $(n-1)$ for all $i = 1, 2, \ldots, N$. 
Assume that for every $i = 1, 2, \ldots, N$
such that $\beta_i \ne 0$ the following condition holds:
$\beta_{i} < 0$ if $n \ge 1$ is odd, 
and $\beta_{i} > 0$ if $n \ge 2$ is even.
Assume furthermore that~$(\ref{eq:special_case})$
is robust in $\mathbb{K} \subset \mathbb{R}^N$.
Then, $N$-dimensional $n$-degree \emph{CDS}
\begin{align}
\frac{\mathrm{d} x_i}{\mathrm{d} t} 
& = \beta_{i} \left(x_i - \frac{T_i}{\mu} \right)^{n}
+ \frac{\mu}{T_i} x_i f_{i}
\left(\mathbf{x} - \frac{\mathbf{T}}{\mu}\right),
\; \; \; \; \; i = 1, 2, \ldots,N, 
\label{eq:special_case_chem} 
\end{align}
in $\mathbb{K} + \mathbf{T}/\mu \subset \mathbb{R}^N$
is qualitatively equivalent to \emph{DS}~$(\ref{eq:special_case})$
in $\mathbb{K}$.
\end{theorem}

\begin{proof}
Applying on~(\ref{eq:special_case}) 
the QCM~(\ref{eq:quasi_chemical}) 
with $q_i = \beta_{i} y_i^{n}$, 
one obtains~(\ref{eq:special_case_chem}).
Statement of the theorem then follows 
from Theorem~\ref{theorem:quasi_chemical}.
\end{proof} 
\noindent
Let us remark that DSs of the form~(\ref{eq:special_case})
can in two dimensions have arbitrary many limit cycles,
and may display chaotic attractors in three dimensions;
see Theorems~\ref{theorem:first_order} and~\ref{theorem:chaos2},
which can be seen as special cases of Theorem~\ref{theorem:special_case}.

\appendix 

\section{Appendix: Background theory} 
\label{app:background}
In this section, we present further background theory.

\begin{definition} $($\textbf{Qualitative equivalence}$)$ \label{def:equivalent}
Consider \emph{DS}~$(\ref{eq:dyn})$
in region $\mathbb{U} \subset \mathbb{R}^N$,
and \emph{DS}~$(\ref{eq:dyn2})$ in region $\mathbb{V} \subset \mathbb{R}^N$.
Let $\boldsymbol{\Phi}$ be a continuous map
from $\mathbb{U}$ onto $\mathbb{V}$, which has a continuous inverse.
Assume that $\boldsymbol{\Phi}$ maps solutions of~$(\ref{eq:dyn})$ 
contained in $\mathbb{U}$ onto the solutions of~$(\ref{eq:dyn2})$ 
contained in $\mathbb{V}$, and preserves the direction of time.
Then, $(\ref{eq:dyn})$ in 
$\mathbb{U}$ is said to be \emph{qualitatively equivalent}
to~$(\ref{eq:dyn2})$ in $\mathbb{V}$.  
\end{definition} 

\begin{definition} $($\textbf{Robust solution}$)$ 
\label{def:robust_solution}
Let $\{\mathbf{y}(t;\mathbf{y}_0) | t \in [0,t_y)\}$ 
be a solution of~$(\ref{eq:dyn})$
for some (possibly infinite) $t_y > 0$.
Assume that for every neighborhood $\mathbb{K} \subset \mathbb{R}^N$
of $\{\mathbf{y}(t;\mathbf{y}_0) | t \in [0,t_y)\}$,
and for every $\mathbf{F} \in C^1$ with 
$\mathbf{F}(\mathbf{z};0) = \mathbf{0}$,
there exists $\mu_0 > 0$
such that for all $\mu \in (0,\mu_0)$
\emph{DS}~$(\ref{eq:dyn3})$ has 
a solution $\{\mathbf{z}(t;\mathbf{z}_0) | t \in [0,t_z) \} \in \mathbb{K}$ 
such that~$(\ref{eq:dyn3})$ in a neighborhood
of $\{\mathbf{z}(t;\mathbf{z}_0) | t \in [0,t_z)\}$
is qualitatively equivalent to~$(\ref{eq:dyn})$ in a neighborhood
of $\{\mathbf{y}(t;\mathbf{y}_0) | t \in [0,t_y)\}$.
Then, solution $\mathbf{y}(t;\mathbf{y}_0)$ 
of~$(\ref{eq:dyn})$ is \emph{robust}.
\end{definition} 

\textbf{Linearization}. 
Let $\nabla \mathbf{f}(\mathbf{y}) \in \mathbb{R}^{N \times N}$ 
be the Jacobian matrix for~(\ref{eq:dyn}), 
where $(\nabla \mathbf{f}(\mathbf{y}))_{i,j} 
= (\partial f_i(\mathbf{y})/\partial y_j)$.
Then, the linearization around the solution $\mathbf{y}(t)$
of~(\ref{eq:dyn}) is given by
\begin{align}
\frac{\mathrm{d} \mathbf{z}}{\mathrm{d} t} 
& = \nabla \mathbf{f}(\mathbf{y}(t)) \mathbf{z}.
\label{eq:dyn_linear} 
\end{align}
If $\mathbf{y}^*$ is a time-independent solution of~(\ref{eq:dyn}), 
then $\mathbf{z}(t) = \exp(\nabla \mathbf{f}(\mathbf{y}^*) t) \mathbf{z}_0$, 
and we associate the eigenvalues of $\nabla \mathbf{f}(\mathbf{y}^*)$
to $\mathbf{y}^*$.

\begin{definition} $($\textbf{Equilibrium}$)$ \label{def:equilibrium}
$\mathbf{y}^* \in \mathbb{R}^N$ is a 
\emph{hyperbolic equilibrium} for~$(\ref{eq:dyn})$ if:
\begin{enumerate}
\item[\emph{(i)}] $\mathbf{f}(\mathbf{y^*}) = \mathbf{0}$, 
and 
\item[\emph{(ii)}] All $N$ eigenvalues associated to $\mathbf{y}^*$ 
have non-zero real parts.
\end{enumerate}
\end{definition} 

Another important class of solutions of~(\ref{eq:dyn}) 
are those that are periodic
with period $\tau > 0$, which we denote by 
$\mathbf{y}_{\tau} = \mathbf{y}_{\tau}(t)$. 
In this case, the solution of the linearization~(\ref{eq:dyn_linear})
is of the form $\mathbf{z}(t) = P(t) \exp(A t) \mathbf{z}_0$, 
where $P(t) \in \mathbb{R}^{N \times N}$ is a $\tau$-periodic matrix, 
and $A \in \mathbb{R}^{N \times N}$ is 
time-independent~\cite{Perko}. 
We say that the eigenvalues of $A$
are characteristic exponents associated to $\mathbf{y}_{\tau}(t)$;
we note that at least one of the exponents has zero real part.

\begin{definition} $($\textbf{Limit cycle}$)$ \label{def:limit_cycle}
$\mathbf{y}_{\tau} : \mathbb{R}_{\ge} \to \mathbb{R}^N$ is 
a \emph{hyperbolic limit cycle} with period $\tau > 0$ 
for~$(\ref{eq:dyn})$ if:
\begin{enumerate}
\item[\emph{(i)}] $\mathbf{y}_{\tau}(t)$ is a solution of~$(\ref{eq:dyn})$
such that $\mathbf{y}_{\tau}(0) = \mathbf{y}_{\tau}(\tau)$ and
$\mathbf{y}_{\tau}(0) \ne \mathbf{y}_{\tau}(t)$ for all $t \in (0,\tau)$,
and
\item[\emph{(ii)}] $(N-1)$ characteristic exponents associated to 
$\mathbf{y}_{\tau}$ have non-zero real parts.
\end{enumerate}
\end{definition}

\textbf{Trapping regions}. 
For solutions $\mathbf{y}(t)$ of~(\ref{eq:dyn}) 
that are neither time-independent nor periodic,
hyperbolicity can still be defined~\cite{Wiggins};
however, linearization~(\ref{eq:dyn_linear})
is more difficult to analyze.
In this more general context, particularly important 
are regions of state-space on whose boundary
vector field of a DS points 
inwards. Consequently, if the state of the DS
is initiated inside such a region, 
then it stays in there for all future times.

\begin{definition} $($\textbf{Trapping region}$)$ \label{def:trapping_region}
Let $\mathbb{S} = \{\mathbf{y} \in \mathbb{R}^N | V(\mathbf{y}) \le r
\textrm{ for some }  r > 0\}$ be a compact set in state-space, 
where $V : \mathbb{R}^N \to \mathbb{R}$ is a continuous function
such that the boundary $V(\mathbf{y}) = r$ has a well-defined
normal at each point. 
If vector field $\mathbf{f}(\mathbf{y})$ 
points inwards on the boundary $V(\mathbf{y}) = r$, 
then $\mathbb{S}$ is called 
a \emph{trapping region} for~$(\ref{eq:dyn})$.
\end{definition}
\noindent
\textbf{Remark}. A sufficient condition for $V(\mathbf{y}) = r$ to have 
a well-defined normal at each point is that  
$V : \mathbb{R}^N \to \mathbb{R}$ is continuously differentiable
and that $\nabla V(\mathbf{y}) \ne \mathbf{0}$
for all $\mathbf{y} \in \mathbb{R}^N$ such that $V(\mathbf{y}) = r$.

\subsection{Chemical reaction networks}
Every CDS induces a canonical set of chemical reactions~\cite{Janos}.
To state this result, given any $x \in \mathbb{R}$, we
define the sign function as
$\textrm{sign}(x) = -1$ if $x < 0$, 
$\textrm{sign}(x) = 0$ if $x = 0$, and 
$\textrm{sign}(x) = 1$ if $x > 0$.

\begin{definition} $($\textbf{Chemical reaction network}$)$ \label{def:CRN}
Assume that~$(\ref{eq:dyn2})$ is a \emph{CDS}, 
i.e. $\mathbf{g} \in \mathbb{P}_n^{\mathcal{C}}(\mathbb{R}^N,\mathbb{R}^N)$. 
Let $\alpha x_1^{\nu_{1}} 
x_2^{\nu_{2}} \ldots x_N^{\nu_{N}}$,
with $\alpha \in \mathbb{R}$
and $\nu_1,\nu_2,\ldots,\nu_N \in \mathbb{Z}_{\ge}$, 
be a monomial from $g_i$. 
Then, monomial $\alpha x_1^{\nu_{1}} 
x_2^{\nu_{2}} \ldots x_N^{\nu_{N}}
= \textrm{\emph{sign}}(\alpha) 
|\alpha| x_1^{\nu_{1}} x_2^{\nu_{2}} \ldots x_N^{\nu_{N}}$
induces the \emph{canonical chemical reaction}
\begin{align}
\sum_{k = 1}^N \nu_{k} X_k 
& \xrightarrow[]{|\alpha|} 
\left(\nu_{i} + \textrm{\emph{sign}}(\alpha) \right) X_i
+ \sum_{k = 1, k \ne i}^N \nu_{k} X_k, 
\label{eq:CR}
\end{align}
where $X_i$ denotes the chemical species whose concentration is $x_i$.
The set of all such chemical reactions,
induced by all the distinct monomials in $\mathbf{g}$,
is called the \emph{canonical chemical reaction network} (\emph{CRN}) 
induced by~$(\ref{eq:dyn2})$. 
\end{definition}
\noindent \textbf{Remark}. 
Species on the left-hand side of~$(\ref{eq:CR})$ 
are called the \emph{reactants}, 
while $|\alpha|$ is called the \emph{rate coefficient}.
Terms of the form $0 X_i$ are denoted by $\varnothing$,
and interpreted as some neglected species.

For any given CDS, the induced canonical CRN 
from Definition~\ref{def:CRN} is unique.
However, a given CDS can also induce other, non-canonical, CRNs.

\begin{definition} $($\textbf{Fused reaction}$)$ \label{def:fused}
Consider $M$ canonical reactions that have 
identical reactants and rate coefficients:
\begin{align}
\sum_{k = 1}^N \nu_{k} X_k & \xrightarrow[]{|\alpha|} 
\left(\nu_{k_j} + \textrm{\emph{sign}}(\alpha) \right) X_{k_j}
+ \sum_{k = 1, k \ne k_j}^N \nu_{k} X_k, 
\; \; \; j = 1,2, \ldots, M. \nonumber 
\end{align}
Then, the corresponding \emph{fused reaction} is given by 
\begin{align}
\sum_{k = 1}^N \nu_{k} X_k & \xrightarrow[]{|\alpha|} 
\left(\nu_{k_1} + \textrm{\emph{sign}}(\alpha) \right) X_{k_1}
+
\ldots
+
\left(\nu_{k_M} + \textrm{\emph{sign}}(\alpha) \right) X_{k_M}
+ \sum_{\substack{k = 1,\\ k \ne k_1,k_2,\ldots,k_M}}^N \nu_{k} X_k.
\nonumber
\end{align}
Any network obtained by fusing reactions in the canonical \emph{CRN}
is called a non-canonical \emph{CRN}.
\end{definition}

\textbf{Example}. The canonical CRN induced by~(\ref{eq:ex_01})
is given by~(\ref{eq:ex_01_CRN1}). 
Since reactions $X_1 \xrightarrow[]{2} 2 X_1$
and $X_1 \xrightarrow[]{2} X_1 + X_2$
from~(\ref{eq:ex_01_CRN1}) have the same reactants
and rate coefficients, they can be fused;
the resulting non-canonical CRN is given~(\ref{eq:ex_01_CRN2}).

\section{Appendix: Time-warp map} 
\label{app:timechange}

\begin{definition} $($\textbf{Time-warp map}$)$ \label{def:timechange}
Consider \emph{DS}~$(\ref{eq:dyn_components})$.
Assume that 
$f_{i} \in \mathbb{P}_n^{\mathcal{N}}(\mathbb{R}^N;\mathbb{R})$
for $i = 1,2,\ldots,K$,
and
$f_{i} \in \mathbb{P}_n^{\mathcal{C}}(\mathbb{R}^N;\mathbb{R})$
for $i = K+1,K+2,\ldots,N$.
Consider \emph{CDS}
\begin{align}
\frac{\mathrm{d} x_i}{\mathrm{d} s} 
& = \left(x_1 x_2 \ldots x_K \right) f_{i}(\mathbf{x}), 
\; \; \;  i = 1,2,\ldots,N.
\label{eq:timechange} 
\end{align}
$\Psi_s: \mathbb{P}_n^{\mathcal{N}}(\mathbb{R}^N; \, \mathbb{R}^N) 
\to \mathbb{P}_{n+K}^{\mathcal{C}}(\mathbb{R}^N; \, \mathbb{R}^N)$,
mapping the vector field of \emph{DS}~$(\ref{eq:dyn_components})$
to the vector field of \emph{CDS}~$(\ref{eq:timechange})$,
is called a \emph{time-warp map}.
\end{definition}

The time-warp map can be interpreted as a 
state-dependent time-rescaling in~$(\ref{eq:dyn_components})$.
\begin{theorem} \label{theorem:timechange}
Consider \emph{DS}~$(\ref{eq:dyn_components})$
with $f_{i} \in \mathbb{P}_n^{\mathcal{N}}(\mathbb{R}^N;\mathbb{R})$
for $i = 1,2,\ldots,K$,
and
$f_{i} \in \mathbb{P}_n^{\mathcal{C}}(\mathbb{R}^N;\mathbb{R})$
for $i = K+1,K+2,\ldots,N$.
Assume that~$(\ref{eq:dyn_components})$
with initial condition $\mathbf{y}_0 = (y_{1,0},\ldots,y_{N,0})^{\top} 
\in \mathbb{R}_{\ge}^N$ with 
$(y_{1,0},\ldots,y_{K,0})^{\top} \in \mathbb{R}_{>}^K$
has a solution $\mathbf{y}(t;\mathbf{y}_0) \in \mathbb{R}_{\ge}^N$
with 
$(y_1(t;\mathbf{y}_0),
\ldots,y_K(t;\mathbf{y}_0))^{\top} \in \mathbb{R}_{>}^K$
for all $t \in [0,t_0]$ for some $t_0 > 0$. 
Then, \emph{CDS}~$(\ref{eq:timechange})$ has the solution
$\mathbf{y}(t(s);\mathbf{y}_0) \in \mathbb{R}_{\ge}^N$
for all $s \in 
[0, \int_0^{t_0} (y_1(\theta;\mathbf{y}_0)\ldots 
y_K(\theta;\mathbf{y}_0))^{-1} \mathrm{d} \theta]$, 
where $t(s)$ is the inverse of 
$s(t) = \int_0^{t} 
(y_1(\theta;\mathbf{y}_0) \ldots 
y_K(\theta;\mathbf{y}_0))^{-1} 
\mathrm{d} \theta$.
\end{theorem}

\begin{proof}
Since $(y_1(t;\mathbf{y}_0),
\ldots,y_K(t;\mathbf{y}_0))^{\top} \in \mathbb{R}_{>}^K$ for all $t \in [0,t_0]$, 
it follows that $s(t)$ is then well-defined and $s(t) > 0$;
furthermore, $s(t)$ 
is monotonically increasing for all $t \in [0,t_0]$
and, hence, has an inverse then, which we denote by $t(s)$. 
Let $x_i(s;\mathbf{y}_0) = y_i(t(s);\mathbf{y}_0)$; then, 
it follows from~(\ref{eq:dyn_components}) 
that for all 
$s \in \left[0, \int_0^{t_0} 
\left(y_1(\theta;\mathbf{y}_0) \ldots y_K(\theta;\mathbf{y}_0) \right)^{-1} 
\mathrm{d} \theta \right]$
variable $x_i = x_i(s;\mathbf{y}_0)$ satisfies 
\begin{align}
\frac{\mathrm{d} x_i}{\mathrm{d} s}
& =  \frac{\mathrm{d} y_i}{\mathrm{d} t} 
\frac{\mathrm{d} t}{\mathrm{d} s}
=  \left(x_1 \ldots x_K \right) f_{i}(\mathbf{x}),
\; \; \;  i = 1,2,\ldots,N,
\label{eq:timechange2} 
\end{align}
which is identical to~(\ref{eq:timechange}).
\end{proof}
\textbf{Remark}. More generally, 
monomial $\left(x_1 x_2 \ldots x_K \right)$
from~(\ref{eq:timechange})
can be replaced with any other polynomial
that ensures~(\ref{eq:timechange}) is chemical;
for example, $\left(x_1^{n_1} x_2^{n_2} \ldots x_K^{n_K} \right)$
with $n_i \in \mathbb{Z}_{>}$ is allowed.

\textbf{Remark}. Assume that a solution of~(\ref{eq:dyn_components})
is bounded for all $t \ge 0$; then, such a solution can be translated
to the positive orthant prior to applying the time-warp map.
Therefore, Theorem~\ref{theorem:timechange} shows that 
such translated time-warp map can preserve any bounded solution, 
such as equilibria and limit cycles.

\section{Appendix: Generalized quasi-chemical map}
\label{app:quasi_chemical_g}
In this section, we present a more general form of the QCM.
To this end, we define vector
$\mathbf{T}_{\mu,a} = (T_1/\mu^{a_1},\ldots,T_N/\mu^{a_N})^{\top}
\in \mathbb{R}_{\ge}^N$
and diagonal matrix $S_{\mu,b} = \textrm{diag}(s_1/\mu^{b_1},
\ldots, s_N/\mu^{b_N}) \in \mathbb{R}_{>}^{N \times N}$, 
where $T_i \ge 0$, $s_i > 0$, and $a_i, b_i \in \mathbb{Z}_{\ge}$
are fixed parameters, while $\mu > 0$ is a free parameter.
Furthermore, we let $p_i(\cdot;\mu) \in  \mathbb{P}_n(\mathbb{R}^N,\mathbb{R})$
be polynomials that are continuously differentiable in $\mu$ 
with $p_i(\mathbf{z};0) = 0$.

\begin{definition} $($\textbf{Generalized quasi-chemical map}$)$ 
\label{def:quasi_chemical_g}
Consider \emph{DS}~$(\ref{eq:dyn_components})$. 
Consider also 
\begin{align}
\frac{\mu^{b_i}}{s_i} \frac{\mathrm{d} x_i}{\mathrm{d} t} 
& = \left[q_{i} \left(S_{\mu,b}^{-1}(\mathbf{x} - \mathbf{T}_{\mu,\mathbf{a}}) \right)
+ p_i \left(S_{\mu,b}^{-1}(\mathbf{x} - \mathbf{T}_{\mu,\mathbf{a}});\mu \right) \right]
\nonumber \\
& + \frac{\mu^{a_i}}{T_i} x_i 
\Big[f_{i} \left(S_{\mu,b}^{-1}(\mathbf{x} - \mathbf{T}_{\mu,\mathbf{a}}) \right) 
- q_{i} \left(S_{\mu,b}^{-1}(\mathbf{x} - \mathbf{T}_{\mu,\mathbf{a}}) \right) \Big],
\; \; \; i = 1,2,\ldots,N.
\label{eq:quasi_chemical_g} 
\end{align}
Assume that $[q_{i}(S_{\mu,b}^{-1}(\mathbf{x} - \mathbf{T}_{\mu,\mathbf{a}}))
+ p_i(S_{\mu,b}^{-1}(\mathbf{x} - \mathbf{T}_{\mu,\mathbf{a}});\mu)]
\in \mathbb{P}_n^{\mathcal{C}}(\mathbb{R}^N,\mathbb{R})$
is chemical for all sufficiently small $\mu > 0$
and for all $i = 1,2,\ldots,N$.
Assume also that $a_i > b_i$ for all $i = 1,2,\ldots,N$.
Then, $\Psi_{\mu} : \mathbb{P}_n(\mathbb{R}^N,\mathbb{R}^N) 
\to \mathbb{P}_{n+1}^{\mathcal{C}}(\mathbb{R}^N,\mathbb{R}^N)$,
mapping the vector field of \emph{DS}~$(\ref{eq:dyn_components})$
to the vector field of \emph{CDS}~$(\ref{eq:quasi_chemical_g})$
for all sufficiently small $\mu > 0$,
is called a \emph{quasi-chemical map}.
\end{definition}

\textbf{Remark}. 
Definition~\ref{def:quasi_chemical_g} can be further generalized by 
replacing the diagonal matrix $S_{\mu,b} \in \mathbb{R}^{N \times N}$
with general matrix $A_{\mu} \in \mathbb{R}^{N \times N}$
that is non-singular for all sufficiently small $\mu > 0$.

\begin{lemma}
\label{lemma:quasi_chemical_g}
Under the affine change of variables $z_i = (\mu^{b_i}/s_i) (x_i - T_i/\mu^{a_i})$
for every $\mu > 0$,
\emph{DS}~$(\ref{eq:quasi_chemical_g})$ becomes 
\begin{align}
\frac{\mathrm{d} z_i}{\mathrm{d} t} 
& = f_{i}(\mathbf{z}) + p_i(\mathbf{z};\mu) 
+ \mu^{a_i - b_i} \frac{s_i}{T_i} z_i
\left[f_{i}(\mathbf{z}) - q_i(\mathbf{z}) \right],
\; \; \; i = 1, 2, \ldots, N.
\label{eq:quasi_chemical_T_g}
\end{align}
\end{lemma}

Comparing Lemma~\ref{lemma:quasi_chemical_g} 
with Lemma~\ref{lemma:quasi_chemical},
one can notice that, in addition to translation, 
the generalized QCM also involves
scaling of the dependent variables, 
as well as additional perturbations
$p_i(\mathbf{z};\mu)$; furthermore, 
different variables can be scaled differently.
It follows from Lemma~\ref{lemma:quasi_chemical_g}, and 
the assumption $a_i > b_i$, that the 
generalized QCM obeys analogous results as those presented in 
Section~\ref{sec:quasi_chemical}.
For example, Theorem~\ref{theorem:quasi_chemical_EL}
holds, but with a different error estimates,
namely $\mathcal{O}(\mu^{\kappa})$, where $\kappa > 0$
is the minimum value in $\{a_1 - b_1, a_2 - b_2, \ldots, a_N - b_N\}$.

\textbf{Special case}. 
To obtain an analogue of~(\ref{eq:quasi_chemical2}),
we choose $q_i = p_i = 0$ in~(\ref{eq:quasi_chemical_g}), leading to
\begin{align}
 \frac{\mathrm{d} x_i}{\mathrm{d} t} 
& = \mu^{a_i - b_i} \frac{s_i}{T_i} x_i 
f_{i} \left(S_{\mu}^{-1}(\mathbf{x} - \mathbf{T}_{\mu,\mathbf{a}}) \right), 
\; \; \; i = 1, 2, \ldots, N.
\label{eq:quasi_chemical_g2} 
\end{align}

\begin{figure}[!htbp]
\vskip  -2.3cm
\leftline{\hskip 0.6cm 
(a) (\ref{eq:ex_QF_3}) with $\mu = 1/100$ \hskip 1.0cm 
(b) (\ref{eq:ex_QF_3}) with $\mu = 1/1000$ \hskip 0.9cm 
(c) (\ref{eq:ex_QF_3}) with $\mu = 1/1000$}
\vskip  0.2cm
\centerline{
\hskip -0.6cm
\includegraphics[width=0.27\columnwidth]{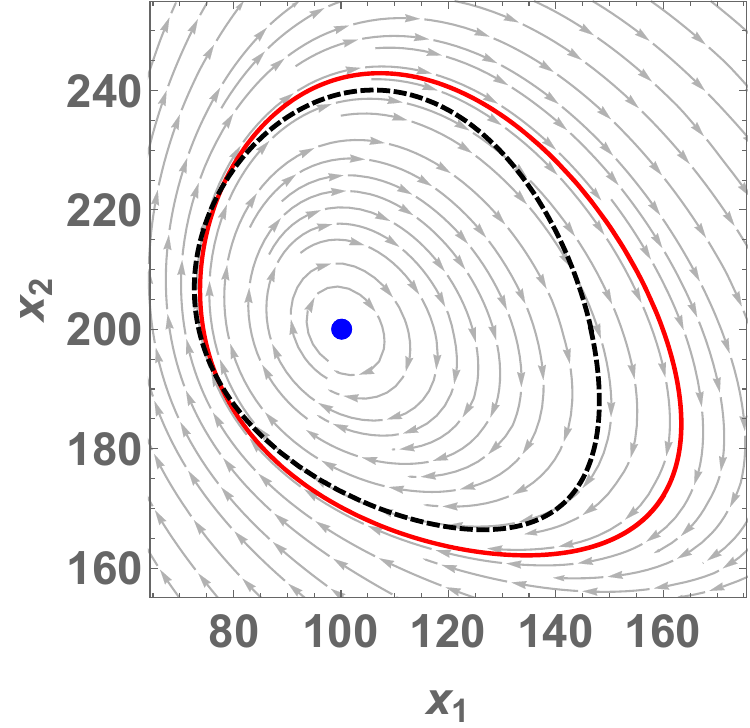}
\hskip 0.3cm
\includegraphics[width=0.283\columnwidth]{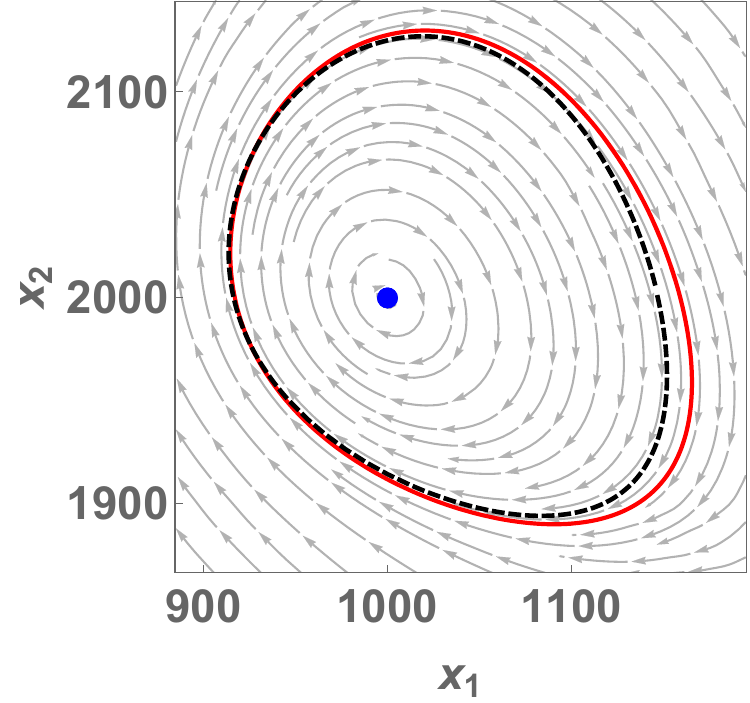}
\hskip 0.3cm
\includegraphics[width=0.35\columnwidth]{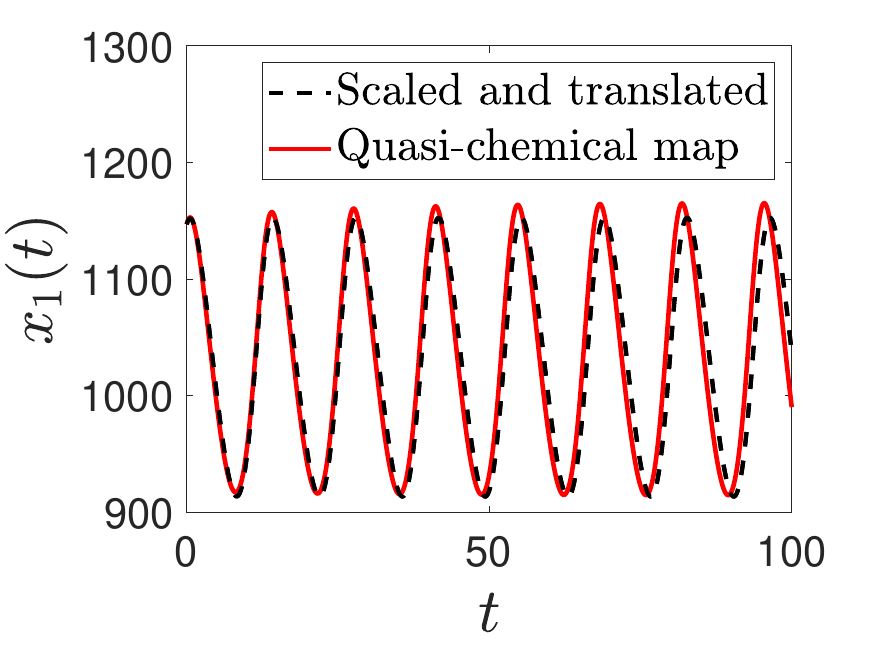}
}
\vskip -0.2cm
\caption{\it{\emph{Application of the generalized 
QCM~(\ref{eq:quasi_chemical_g}) on DS~(\ref{eq:ex1}).}} 
Panel \emph{(a)} displays state-space
for~$(\ref{eq:ex_QF_3})$ with $\mu = 1/100$,
with the limit cycle shown as the solid red curve, 
and the equilibrium as the blue dot; 
also shown as dashed black curve is the limit cycle of~$(\ref{eq:ex1})$
under the same scaling and translation as in~$(\ref{eq:ex_QF_3})$.
Panel \emph{(b)} displays the same plot for $\mu = 1/1000$, 
while panel \emph{(c)} displays a corresponding 
time-state space.} \label{fig:2_app}
\end{figure}

\textbf{Example}.
Let us apply on~(\ref{eq:ex1}) the QCM~(\ref{eq:quasi_chemical_g})
with $q_1 = 0$, $q_2 = f_2$, $p_1 = p_2 = 0$, 
$T_1 = 1$, $T_2 = 2$, $a_1 = a_2 = 1$, 
$s_1 = s_1 = 1$ and $b_1 = b_2 = 1/2$,
which leads to
\begin{align}
\frac{\mathrm{d} x_1}{\mathrm{d} t} 
& = \mu x_1 \left[ \frac{1}{2} \left(x_2 - \frac{2}{\mu} \right) 
+ \frac{1}{8} \left(x_1 - \frac{1}{\mu} \right) \right], \nonumber \\
\frac{\mathrm{d} x_2}{\mathrm{d} t} 
& = -\frac{1}{2} \left(x_1 - \frac{1}{\mu} \right)
- \mu^{1/2} \frac{1}{20}  \left(x_1 - \frac{1}{\mu} \right) 
\left(x_2 - \frac{2}{\mu} \right)
+ \mu^{1/2} \frac{1}{20}  \left(x_2 - \frac{2}{\mu} \right)^2 
- \frac{3}{32} \left(x_2 - \frac{2}{\mu} \right). \label{eq:ex_QF_3}
\end{align}
One can readily show that~(\ref{eq:ex_QF_3})
is chemical for all sufficiently small $\mu > 0$. 
In Figure~\ref{fig:2_app}(a), we display state-space 
for~(\ref{eq:ex_QF_3}) when $\mu = 1/100$.
Analogous plot is shown in Figure~\ref{fig:2_app}(b)
for $\mu = 1/1000$, with a corresponding time-state space
shown in Figure~\ref{fig:2_app}(c). These plots 
confirm that the limit cycle is now of order
$\mathcal{O}(\mu^{-1/2})$, as opposed to 
$\mathcal{O}(1)$ as in Figure~\ref{fig:2}.
Note that the error
between the limit cycle of~(\ref{eq:ex_QF_3}) 
and that of suitably scaled and translated target system, 
as well as their periods, is $\mathcal{O}(\mu^{1/2})$, 
as opposed to $\mathcal{O}(\mu)$
as in Figure~\ref{fig:2}; in particular, the convergence
order is now lower due to the magnification of the limit cycle.

\section{Appendix: Perturbed harmonic oscillator} 
\label{app:linearoscillator}

\begin{lemma} \label{lemma:first_order}
Consider the perturbed harmonic oscillator given by~$(\ref{eq:HO_1})$.
Let $0 < r_1< r_2 < \ldots < r_n$ be arbitrary real numbers.
Then, there exists coefficients $\beta_i 
= \beta_i(r_1,r_2,\ldots,r_N)$ for $i = 0, 1, \ldots, N$ 
such that for every sufficiently small $\varepsilon > 0$
system~$(\ref{eq:HO_1})$ has $n$ nested limit cycles. 
The $i$th limit cycle is arbitrarily close to
$y_1(t) = r_i \cos(t)$
and
$y_2(t) = -r_i \sin(t)$,
and has period arbitrarily close $2 \pi$. 
The $n$th limit cycle is stable, 
and the limit cycles alternate in stability.
\end{lemma}

\begin{proof}
Substituting $f_1 = \beta_0 y_1 + \beta_1 y_1^3 + 
\beta_2 y_1^5 + \ldots + \beta_{n} y_1^{2 n + 1}$ 
and $f_2 = 0$ into
\begin{align}
I(r) & = \int_0^{2 \pi} 
\Big[f_1(r \cos(s),-r \sin(s)) \cos(s) 
- f_2(r \cos(s),-r \sin(s))  \sin(s) \Big] 
\mathrm{d} s, \nonumber
\end{align}
one obtains
\begin{align}
I(r) & = \pi r \sum_{i = 0}^n 2^{-2 i} 
\binom{2 i + 1}{i + 1} \beta_{i} r^{2 i},
\label{eq:I}
\end{align}
where we use $\int_0^{2 \pi} \cos^{2 n}(t) \mathrm{d} t 
= 2^{-2 n + 2} \binom{2 n - 1}{n} \pi$.
Using Vieta's formulas, it follows that 
$0 < r_1^2 < r_2^2 < \ldots < r_n^2$ are the roots
of $I(r)$ if 
\begin{align}
\beta_i & = (-1)^{n-i+1} 2^{2 i} \binom{2 i + 1}{i + 1}^{-1}
\sum_{l \in C_{n-i}^n} 
\left(r_{l(1)} r_{l(2)} \ldots r_{l(n-i)} \right)^2
\; \; \; \textrm{for all } i = 0,1,2,\ldots,n-1, \nonumber \\
\beta_n & = -2^{2 n} \binom{2 n + 1}{n + 1}^{-1},
\label{eq:HO_1_b}
\end{align}
where $C_{k}^{n}$ denotes the set of all combinations 
of length $1 \le k \le n$
of the elements from $\{1,2,\ldots,n\}$.
Simple zeros of $I(r)$ correspond to 
the hyperbolic limit cycles of~(\ref{eq:HO_1}), 
and their stability is 
determined by the slope of $I(r)$~\cite{Perko}[Chapter 4.9],
implying the statement of the lemma.
\end{proof}
\noindent
\textbf{Remark}. We impose that $2^{-2 n} \binom{2 n + 1}{n + 1} \beta_{n} < 0$
to ensure stability of the limit cycle with radius $r_n$;
we arbitrarily let $2^{-2 n} \binom{2 n + 1}{n + 1} \beta_{n} = -1$.

\textbf{Rate coefficients for~(\ref{eq:HO_1_C})}. 
The coefficients are given by
\begin{align}
\alpha_i & = \left|\delta_{0}(i) \frac{1}{\mu} 
+ \varepsilon \sum_{j = 0}^n \beta_j
\binom{2 j + 1}{i} \left(\frac{1}{\mu} \right)^{2 j + 1 - i}\right|,
\; \; i = 0,1,\ldots, 2 n + 1, 
\; \; \alpha_{2 n + 2} = 1, 
\; \; \alpha_{2 n + 3} = \mu,  
\label{eq:HO_1_C_coeff}
\end{align}
where $\beta_0,\beta_1,\ldots,\beta_{n}$ are given by~$(\ref{eq:HO_1_b})$,
and $\delta_{i}(j)$ be the Kronecker-delta, i.e. $\delta_i(j) = 1$ if $j = i$, 
and $\delta_{i}(j) = 0$ otherwise.

\section{Appendix: Candidate quadratic CRN with chaos} 
\label{app:chaos}
Chemical Lorenz system~(\ref{eq:Lorenz_chemical})
contains a cubic term, i.e. CRN~(\ref{eq:Lorenz_CRN}) contains 
one tri-molecular reaction.
Let us now focus on deriving a quadratic CDS with chaos.
To this end, consider the DS
\begin{align}
\frac{\mathrm{d} y_1}{\mathrm{d} t} & = \frac{27}{10} y_2 + y_3, \nonumber \\
\frac{\mathrm{d} y_2}{\mathrm{d} t} & = -y_1 + y_2^2, \nonumber \\
\frac{\mathrm{d} y_3}{\mathrm{d} t} & = y_1 + y_2,
\label{eq:Sprott_P}
\end{align}
which has been put forward
as displaying a chaotic attractor in~\cite{Sprott} as ``Case P".
This attractor has been numerically investigated, 
see also Figure~\ref{fig:5}(a)--(b);
however, no rigorous proofs are given. 
Note that proving
existence and properties of chaotic attractors
can be a difficult task. For example, 
there is a gap of more than three decades between
numerical evidence of the Lorenz attractor~\cite{Lorenz} 
and its rigorous proof~\cite{Tucker}.
To proceed, in this paper
we simply assume that~(\ref{eq:Sprott_P})
has a robust chaotic attractor, which allows us
to state the following result. 

\begin{theorem}$($\textbf{Quadratic \emph{CRN} with a chaotic attractor}$)$ 
\label{theorem:chaos2}
Assume that \emph{DS}~$(\ref{eq:Sprott_P})$
has a chaotic attractor which is robust. 
Consider the \emph{CRN}
\begin{align}
X_1 &\xrightarrow[]{\alpha_1} \varnothing, 
\; \; \; X_1 + X_2 \xrightarrow[]{\alpha_2} 2 X_1, 
\; \; \; X_1 + X_3 \xrightarrow[]{\alpha_3} 2 X_1 + 2 X_3, 
\; \; \; \varnothing \xrightleftharpoons[\alpha_5]{\alpha_4} X_2, \nonumber \\ 
2 X_2 & \xrightarrow[]{\alpha_6} 3 X_2, 
\; \; \;  X_3 \xrightarrow[]{\alpha_7} \varnothing,
\; \; \;  X_2 + X_3 \xrightarrow[]{\alpha_8} X_2 + 2 X_3.
\label{eq:quadratic_chaos_CRN}
\end{align}
Assume that the rate coefficients are given by 
\begin{align}
\alpha_1 & = 2, \; \; \; 
\alpha_2 = \frac{27}{10} \mu,\; \; \; 
\alpha_3 = \mu, \; \; \; 
\alpha_4 = \frac{100}{729 \mu^2}, \; \; \; 
\alpha_5 = \frac{20}{27 \mu} - \frac{27}{10}, \; \; \; 
\alpha_6 = 1, \; \; \; 
\alpha_7 = \frac{37}{27}, \; \; \; 
\alpha_{8} = \mu. 
\label{eq:quadratic_chaos_coeff}
\end{align}
Then, for every sufficiently small $\mu > 0$, 
\emph{CRN}~$(\ref{eq:quadratic_chaos_CRN})$
has a chaotic attractor which is qualitatively
equivalent to the robust chaotic attractor of~$(\ref{eq:Sprott_P})$.
\end{theorem}

\begin{figure}[!htbp]
\vskip -0.0cm
\leftline{\hskip 2.2cm (a) \hskip  6.9cm (c)}
\centerline{
\hskip 0.0cm
\includegraphics[width=0.35\columnwidth]{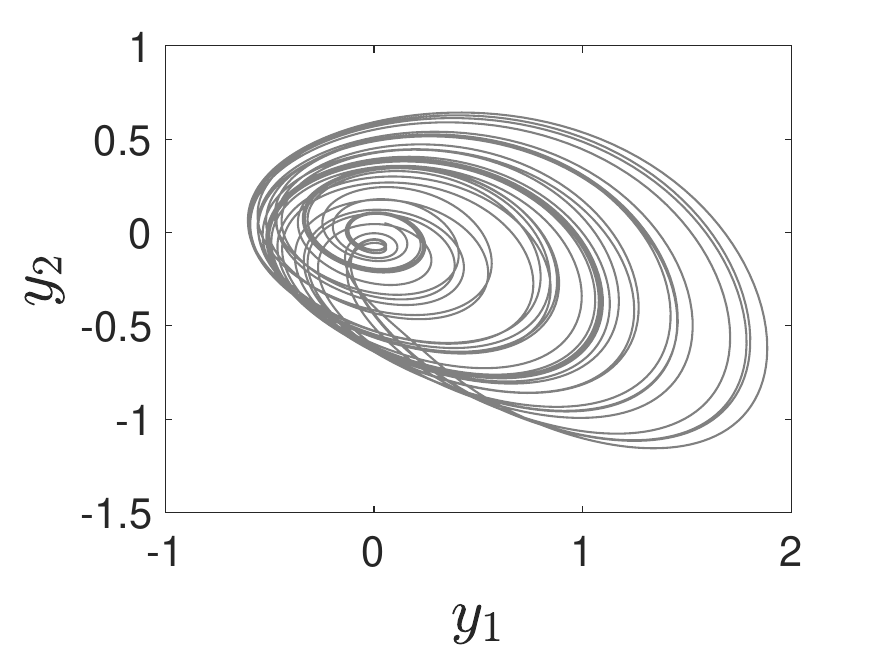}
\hskip 1.5cm
\includegraphics[width=0.35\columnwidth]{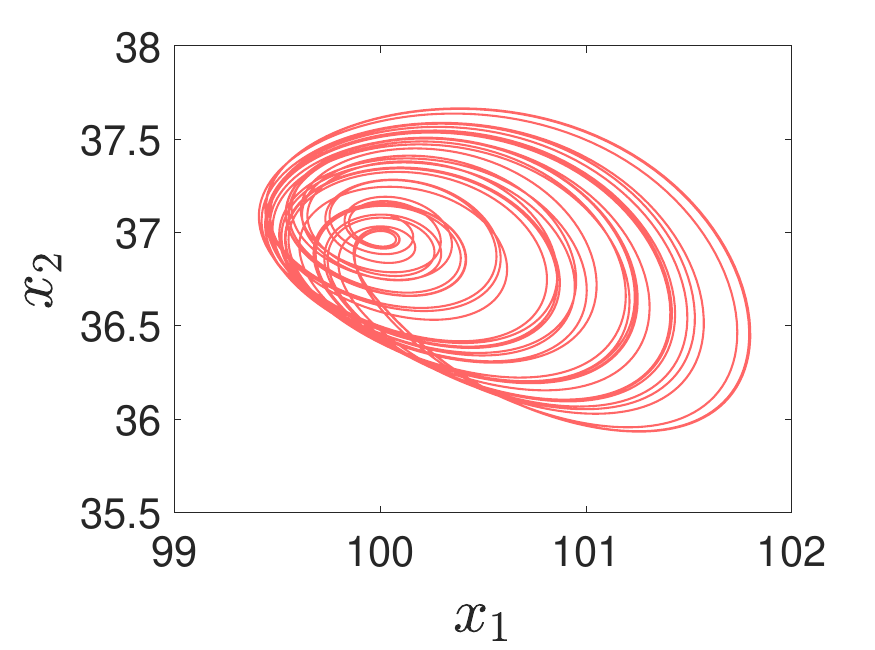}
}
\leftline{\hskip 2.2cm (b) \hskip  6.9cm (d)}
\centerline{
\hskip 0.0cm
\includegraphics[width=0.35\columnwidth]{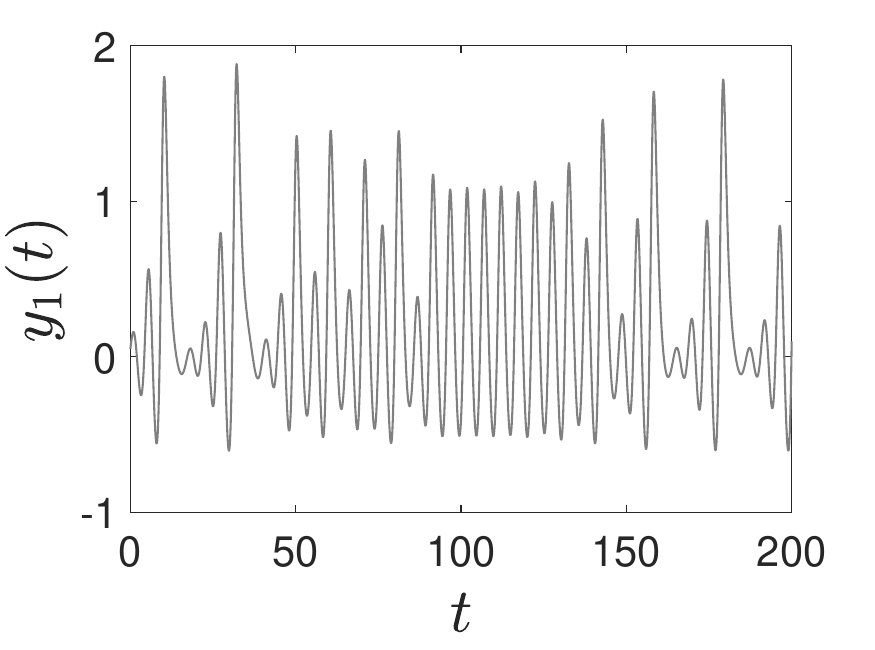}
\hskip 1.5cm
\includegraphics[width=0.35\columnwidth]{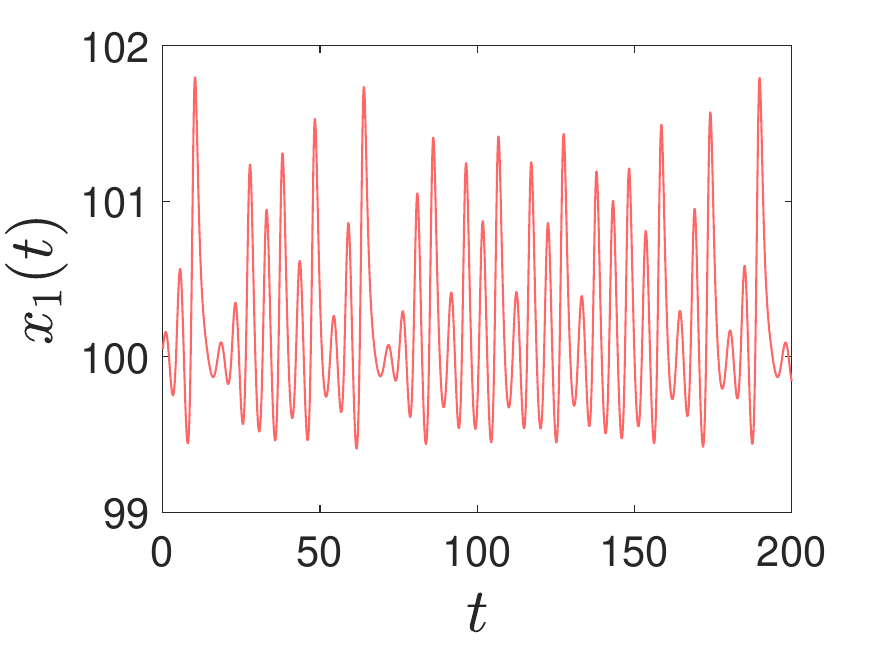}
}
\vskip -0.2cm
\caption{\it{\emph{Quadratic CDS with chaos.} 
Panels \emph{(a)}--\emph{(b)} display $(y_1,y_2)$-
and $(t,y_1)$-space for \emph{DS}~$(\ref{eq:Sprott_P})$,
with initial condition $y_1(0) = y_2(0) = y_3(0) = 5/100$.
Panels \emph{(c)}--\emph{(d)} show analogous plots 
for \emph{CDS}~$(\ref{eq:Sprott_P_chemical})$
with initial condition $x_1(0) = x_3(0) = 5/100 + 100$
and $x_2(0) = 5/100 + 1000/27$.}} 
\label{fig:5}
\end{figure}

\begin{proof}
Let us apply on DS~(\ref{eq:Sprott_P}) 
the QCM~(\ref{eq:quasi_chemical})
with $T_1 = T_3 = 1$, $T_2 = 10/27$, 
$q_1 = q_3 = 0$ and $q_2 = y_2^2$, leading to the CDS
\begin{align}
\frac{\mathrm{d} x_1}{\mathrm{d} t} 
& = \mu x_1 \left[ \frac{27}{10} \left(x_2 - \frac{10}{27 \mu} \right)
+  \left(x_3 - \frac{1}{\mu} \right)\right],
\nonumber \\
\frac{\mathrm{d} x_2}{\mathrm{d} t} & =  \left(x_2 - \frac{10}{27 \mu} \right)^2
+ \frac{27}{10} \mu x_2 \left[ -\left(x_1 - \frac{1}{\mu} \right) \right], \nonumber \\
\frac{\mathrm{d} x_3}{\mathrm{d} t} & = 
\mu x_3 \left[\left(x_1 - \frac{1}{\mu} \right) 
+ \left(x_2 - \frac{10}{27 \mu} \right)\right].
\label{quadratic_chaos_g}
\end{align}
An induced CRN is given by~(\ref{eq:quadratic_chaos_CRN}),
with rate coefficients~(\ref{eq:quadratic_chaos_coeff}).
Using the assumption that~(\ref{eq:Sprott_P})  has a robust attractor,
the statement of the theorem follows from
Theorem~\ref{theorem:quasi_chemical}.
\end{proof}
\noindent 
\textbf{Remark}. 
The choice of the translation parameters has been made
so that the monomial $x_1 x_2$ appears with the 
same coefficient in the first two equations 
in~(\ref{quadratic_chaos_g}), allowing for
fusing of the underlying canonical reactions; 
see Definition~\ref{def:fused}.  

\noindent 
\textbf{Remark}. The assumption in Theorem~\ref{theorem:chaos2}
can be weakened: it is sufficient to assume that~$(\ref{eq:Sprott_P})$
has a chaotic attractor which is robust
with respect to only 
the chemical perturbations underlying~(\ref{quadratic_chaos_g});
see also the remark below Theorem~\ref{theorem:quasi_chemical}.

Let us fix $\mu = 1/100$ in~(\ref{eq:quadratic_chaos_coeff}),
under which CDS~(\ref{quadratic_chaos_g}) becomes
\begin{align}
\frac{\mathrm{d} x_1}{\mathrm{d} t} 
& = - 2 x_1 + \frac{27}{1000} x_1 x_2 + \frac{1}{100} x_1 x_3,
\nonumber \\
\frac{\mathrm{d} x_2}{\mathrm{d} t} 
& = \frac{10^6}{729} - \frac{19271}{270} x_2 - \frac{27}{1000} x_1 x_2 + x_2^2, 
\nonumber \\
\frac{\mathrm{d} x_3}{\mathrm{d} t} 
& = -\frac{37}{27} x_3 + \frac{1}{100} x_1 x_3 + \frac{1}{100} x_2 x_3.
\label{eq:Sprott_P_chemical}
\end{align}
Dynamics of~(\ref{eq:Sprott_P_chemical})
is shown in Figure~\ref{fig:5}(c)--(d), 
which compares well with the original DS~(\ref{eq:Sprott_P}).

\section{Appendix: Quadratic CRN with homoclinic bifurcation} 
\label{app:homoclinic}
Consider the DS
\begin{align}
\frac{\mathrm{d} y_1}{\mathrm{d} t} & = y_2, \nonumber \\
\frac{\mathrm{d} y_2}{\mathrm{d} t} & = y_1 + \beta y_2
+ y_1^2 - y_1 y_2,
\label{eq:homoclinic2}
\end{align}
which undergoes a super-critical homoclinic 
bifurcation~\cite{Perko}[Chapter 4.8].

\begin{theorem}$($\textbf{Quadratic \emph{CRN} with homoclinic bifurcation}$)$ 
\label{theorem:bifurcation2}
Consider the \emph{CRN}
\begin{align}
X_1 &\xrightarrow[]{\alpha_1} \varnothing,  
\; \; \; X_1 + X_2 \xrightarrow[]{\alpha_2} 2 X_1 + X_2, 
\; \; \; \varnothing \xrightleftharpoons[\alpha_4]{\alpha_3} X_2,
\; \; \; X_1 \xrightarrow[]{\alpha_5} X_1 + X_2, \nonumber \\
2 X_1 & \xrightarrow[]{\alpha_6} 2 X_1 + X_2, 
\; \; \; X_1 + X_2 \xrightarrow[]{\alpha_7} X_1, 
\; \; \; 2 X_2 \xrightarrow[]{\alpha_8} 3 X_2.
\label{eq:homoclinic_CRN2}
\end{align}
Assume that the rate coefficients are given by 
\begin{align}
\alpha_1 & = \frac{1}{\mu^{1/2}}, \; \; \; 
\alpha_2 = \mu^{1/2},\; \; \; 
\alpha_3 = \frac{1}{\mu^{5/3}} - \frac{1}{\mu^{3/2}} + \frac{1}{\mu} 
- \frac{1}{\mu^{1/2}}, \; \; \; 
\alpha_4 = \frac{2}{\mu^{2/3}} - \frac{1}{\mu^{1/2}}, \nonumber \\
\alpha_5 & = \frac{1}{\mu} - \frac{2 + \beta}{\mu} + 1, \; \; \; 
\alpha_6 = 1, \; \; \; 
\alpha_7 = 1 - \beta \mu^{1/2}, \; \; \; 
\alpha_8 =  \mu^{1/3}.
\label{eq:homoclinic_coeff2}
\end{align}
Then, for every sufficiently small $\mu > 0$, 
\emph{CRN}~$(\ref{eq:homoclinic_CRN2})$
undergoes a super-critical homoclinic bifurcation
at saddle $(1/\mu^{1/2},1/\mu) \in \mathbb{R}_{>}^{2}$
at some parameter value $\beta = \beta(\mu) \in (-1,0)$.
\end{theorem}

\begin{proof}
Consider the perturbed DS
\begin{align}
\frac{\mathrm{d} z_1}{\mathrm{d} t} & = 
f_1(z_1,z_2) = z_2 + \mu^{1/2} z_1 z_2, 
\nonumber \\
\frac{\mathrm{d} z_2}{\mathrm{d} t} & =
f_2(z_1,z_2; \beta) = z_1 + 
\beta \left(z_2 + \mu^{1/2} z_1 z_2 \right)
+ z_1^2 - z_1 z_2 + \mu^{1/3} z_2^2.
\label{eq:homoclinic2p}
\end{align}
The angle between the vector field of~(\ref{eq:homoclinic2p})
and the $z_1$-axis is given by 
$\theta = \tan^{-1} [f_2(z_1,z_2)/f_1(z_1,z_2; \beta)]$,
and hence satisfies the differential equation
\begin{align}
\frac{\mathrm{d} \theta}{\mathrm{d} \beta} 
& = \frac{f_1(z_1,z_2) \frac{\partial f_2(z_1,z_2; \beta)}{\partial \beta}}
{f_1^2(z_1,z_2) + f_2^2(z_1,z_2; \beta)}
= \frac{f_1^2(z_1,z_2)}{f_1^2(z_1,z_2) + f_2^2(z_1,z_2; \beta)}.
\label{eq:angle}
\end{align}
Hence, for points in the state-space which are
not equilibria, the vector field rotates anti-clockwise
as parameter $\beta$ is increased.  
If $\mu = 0$, then~(\ref{eq:homoclinic2p})
has a stable hyperbolic limit cycle with clock-wise orientation
for a range of values $\beta > -1$~\cite{Perko}[Chapter 4.8];
hence, the same is true for~(\ref{eq:homoclinic2p})
for all sufficiently small $\mu > 0$. 
It follows from~(\ref{eq:angle}) and~\cite{Perko}[Chapter 4.6]
that this limit cycle monotonically expands as $\beta$ is increased, 
until it forms a stable homoclinic loop at the saddle $(0,0)$
at some bifurcation value $\beta(\mu) > -1$.
Since the homoclinic loop is stable, 
it follows that the trace of Jacobian for~(\ref{eq:homoclinic2p})
at $(0,0)$ is negative at the bifurcation, 
i.e. $\beta(\mu) < 0$~\cite{Perko}[Chapter 4.8].
Choosing a sufficiently small $\mu > 0$, 
 restricting $\beta \in (-1,0)$,
and applying translation $x_1 = (z_1 + 1/\mu^{1/2})$
and $x_2 = (z_2 + 1/\mu)$, DS~(\ref{eq:homoclinic2p}) 
becomes the CDS
\begin{align}
\frac{\mathrm{d} x_1}{\mathrm{d} t} 
& = \mu^{1/2} x_1 \left(x_2 - \frac{1}{\mu} \right), 
\nonumber \\
\frac{\mathrm{d} x_2}{\mathrm{d} t} & =
\left(x_1 - \frac{1}{\mu^{1/2}} \right) + 
\beta \left[\left(x_2 - \frac{1}{\mu} \right)
 + \mu^{1/2} \left(x_1 - \frac{1}{\mu^{1/2}} \right) 
 \left(x_2 - \frac{1}{\mu} \right)\right]
+ \left(x_1 - \frac{1}{\mu^{1/2}} \right)^2 \nonumber \\
& - \left(x_1 - \frac{1}{\mu^{1/2}} \right)
 \left(x_2 - \frac{1}{\mu} \right) + 
 \mu^{1/3} \left(x_2 - \frac{1}{\mu} \right)^2,
\label{eq:homoclinic_g2}
\end{align}
with a CRN given by~(\ref{eq:homoclinic_CRN2})--(\ref{eq:homoclinic_coeff2}).
\end{proof}
\noindent
\textbf{Remark}. CDS~(\ref{eq:homoclinic_g2})
can also be obtained by applying on~(\ref{eq:homoclinic2})
the QCM~(\ref{eq:quasi_chemical_g})
with $q_1 = 0$, $p_1 = 0$, $T_1 = 1$, $a_1 = 1/2$, $s_1 = 1$, $b_1 = 0$, 
$q_2 = f_2$, $p_2 = \mu^{1/2} \beta y_1 y_2 + \mu^{1/3} y_2^2$, 
$T_2 = 1$, $a_2 = 1$, $s_2 = 1$, $b_2 = 0$.
One can introduce parameter
$\varepsilon \equiv \mu^{1/6}$ to obtain
 continuously differentiable perturbations 
in~(\ref{eq:homoclinic2p}). 

Fixing $\mu= 1/100$ in~(\ref{eq:homoclinic_g2}), 
one obtains the CDS
\begin{align}
\frac{\mathrm{d} x_1}{\mathrm{d} t} 
& = -10 x_1 + \frac{1}{10} x_1 x_2, 
\nonumber \\
\frac{\mathrm{d} x_2}{\mathrm{d} t} & =
\left(-910 + 10^{10/3} \right) 
+ (81 - 10 \beta) x_1 
+ (10 - 2 \times 10^{4/3}) x_2 
+ x_1^2 + \left(-1 + \frac{\beta}{10} \right) x_1 x_2 
+ \frac{1}{10^{2/3}} x_2^2,
\label{eq:homoclinic_g2_fixed}
\end{align}
which is found numerically to undergo the 
homoclinic bifurcation at saddle
$(10,100)$ when $\beta \approx \frac{-191}{200}$.

\section{Appendix: Non-polynomial dynamical systems} 
\label{app:nonpoly}
Let us consider DSs
with continuously differentiable vector fields, 
\begin{align}
\frac{\mathrm{d} \mathbf{y}}{\mathrm{d} t} & = 
\mathbf{h}(\mathbf{y}),
\; \; \; \textrm{where } \mathbf{h} \in C^1(\mathbb{R}^N,\mathbb{R}^N).
\label{eq:dyn_nonpoly} 
\end{align} 

\begin{theorem}
\label{theorem:non_polynomial}
Assume that vector field $\mathbf{h} \in C^1(\mathbb{R}^N,\mathbb{R}^N)$
is non-polynomial.
Assume furthermore that~$(\ref{eq:dyn_nonpoly})$ 
is robust in $\mathbb{K} \subset \mathbb{R}^N$. 
Then, there exists an $N$-dimensional polynomial \emph{DS}~$(\ref{eq:dyn})$
that is in $\mathbb{K}$ qualitatively equivalent to
$N$-dimensional non-polynomial 
\emph{DS}~$(\ref{eq:dyn_nonpoly})$ in $\mathbb{K}$.
\end{theorem}

\begin{proof}
The Weierstrass approximation 
theorem~\cite{Weierstrass}[Theorem 1.6.2]
guarantees that for every $\varepsilon > 0$ 
there exists a polynomial function
$\mathbf{f}_{\varepsilon} : \mathbb{K} \to \mathbb{R}^N$
and a continuously differentiable function 
$\mathbf{F}_{\varepsilon} \in C^1(\mathbb{R}^N,\mathbb{R}^N)$
such that 
$\mathbf{h}(\mathbf{y}) = \mathbf{f}_{\varepsilon}(\mathbf{y}) + 
\mathbf{F}_{\varepsilon} (\mathbf{y})$
for all $\mathbf{y} \in \mathbb{K}$,
where $\textrm{max}_{\mathbf{y} \in \mathbb{K}} 
(\|\mathbf{F}_{\varepsilon} (\mathbf{y})\| 
+ \|\nabla \mathbf{F}_{\varepsilon}(\mathbf{y})\|)
< \varepsilon$. Using the assumption that~(\ref{eq:dyn_nonpoly})
is robust in $\mathbb{K}$, it follows that 
for every sufficiently small $\varepsilon > 0$
DS~(\ref{eq:dyn}) with vector field
$\mathbf{f}_{\varepsilon}(\mathbf{y}) = \mathbf{h}(\mathbf{y}) 
- \mathbf{F}_{\varepsilon} (\mathbf{y})$ is 
qualitatively equivalent to DS~(\ref{eq:dyn_nonpoly}) in $\mathbb{K}$,
implying the statement of the theorem.
\end{proof}


\begin{thebibliography}{9}
\bibitem{Perko} Perko, L.: 
Differential Equations and Dynamical Systems. 
3rd Edition, Springer-Verlag (2001).

\bibitem{Feinberg} Feinberg, M.: Lectures on chemical reaction networks. 
Mathematics Research Center, University of Wisconsin (1979).

\bibitem{Janos} \'{E}rdi, P., T\'{o}th, J.:
Mathematical models of chemical reactions. Theory and applications of deterministic and stochastic Models. 
Manchester University Press, Princeton University Press (1989).

\bibitem{Synthetic} Endy D.: 
Foundations for engineering biology. 
Nature, 484, 449--453 (2005).

\bibitem{DNA} Soloveichik, D., Seeing G., Winfree E.: 
DNA as a universal substrate for chemical kinetics.
Proceedings of the National Academy of Sciences, 
107(12), 5393--5398 (2010).

\bibitem{Me1} Plesa, T., Vejchodsk\'{y}, T., Erban, R.: 
Chemical reaction systems with a homoclinic bifurcation: 
An inverse problem. Journal of Mathematical Chemistry, 54(10): 1884--1915 (2016).

\bibitem{Me2}  Plesa, T., Vejchodsk\'{y}, T., Erban, R.: 
Test models for statistical inference: Two-dimensional reaction systems displaying limit cycle bifurcations and bistability. 
In Stochastic Dynamical Systems, Multiscale Modeling, Asymptotics and 
Numerical Methods for Computational Cellular Biology (2017). 

\bibitem{Me3} Plesa, T., Zygalakis, K. C., Anderson, D. F., Erban, R.: Noise control for molecular computing.
Journal of the Royal Society Interface, 
15(144), 20180199 (2018).

\bibitem{Me4} Plesa, T., Stan, G. B., Ouldridge, T. E., Bae, W.:
Quasi-robust control of biochemical reaction networks via stochastic morphing. Journal of the Royal Society Interface, 18, 1820200985 (2021).

\bibitem{Me5} Plesa, T., Dack, A., Ouldridge, T. E.:
Integral feedback in synthetic biology: Negative-equilibrium catastrophe.
Journal of Mathematical Chemistry, 61, 1980--2018 (2023).

\bibitem{RNCRN} Dack, A., Qureshi, B., Ouldridge, T. E., Plesa, T.: 
Recurrent neural chemical reaction networks that approximate arbitrary dynamics. Available as: https://arxiv.org/abs/2406.03456 (2025).

\bibitem{Time_change1}  Figueiredo, A., Gleria, I.M., Rocha, T.M.: 
Boundedness of solutions and Lyapunov functions in quasi-polynomial systems. Phys. Lett. A, 268, 335--341 (2000). 

\bibitem{Time_change2} Hangos, K.M., Szederk\'enyi, G.:
Mass action realizations of reaction kinetic system models on various time scales. J. Phys.: Conf. Ser., 268, 012009 (2011).

\bibitem{Me6} Plesa, T.: 
Stochastic approximations of higher-molecular by bi-molecular reactions.
Journal of Mathematical Biology, 86(2), 28 (2023).

\bibitem{Samardzija} Samardzija, N., Greller, L.D., Wasserman, E.: 
Nonlinear chemical kinetic schemes derived from mechanical 
and electrical dynamical systems. 
J. Chem. Phys. 90, 2296--2304 (1989).

\bibitem{Kuznetsov} Kuznetsov, Y. A.: 
Elements of applied bifurcation theory. 
Springer, Berlin (1998).

\bcol{

\bibitem{Wiggins} Wiggins, S.: Introduction to Applied Nonlinear Dynamical Systems and Chaos.
New York : Springer-Verlag (1990).

\bibitem{Hyperbolicity} Ara\'ujo, V., Viana, M.: 
Hyperbolic Dynamical Systems. In: Meyers, R. (eds) Encyclopedia of Complexity and Systems Science. Springer, New York, NY (2009).
}

\bibitem{Escher} Escher, C.: 
Bifurcation and coexistence of several limit cycles in models 
of open two-variable quadratic mass-action systems. 
Chem. Phys., 63(3), 337--348 (1981).

\bibitem{Coddington} Coddington, A., Levinson, N.: 
Theory of Ordinary Differential Equations. 
McGraw-Hill, New-York (1955).

\bibitem{Eckweiler} Eckweiler, H. J.: 
Nonlinear differential equations of the van der Pol type 
with a variety of periodic solutions. 
Studies in Nonlinear Vibration Theory, Institute of Mathematics and Mechanics, New York University (1946).

\bibitem{Lorenz} Lorenz, E. N.: 
Deterministic nonperiodic flow. 
Journal of Atmospheric Sciences, 20(2), 130--141 (1963).

\bibitem{Tucker} Tucker, W.: 
The Lorenz attractor exists. 
Comptes Rendus de l’Académie des Sciences-Series I-Mathematics, 
328(12), 1197--1202 (1999).

\bibitem{Sandstede} Sandstede, B.: 
Constructing dynamical systems having homoclinic bifurcation points of codimension two. Journal of Dynamics and Differential Equations, Vol. 9, No. 2., 4, 296--288 (1997).

\bibitem{Hilbert1} Hilbert, D.:
Mathematical problems.
Bulletin of the American Mathematical Society, 
80, 437--479 (1902).

\bibitem{Hilbert2} Ilyashenko, Y.: 
Centennial history of Hilbert’s 16th problem.
Bulletin of the American Mathematical Society, 
39(3), 301--354 (2002).

\bibitem{Hilbert_Robust} Gasull, A., Santana, P.: 
A note on Hilbert 16th Problem.
Proc. Amer. Math. Soc. 153, 669--677 (2025).

\bibitem{H2} Chen, L., Wang, M.:
The relative position, and the number, of limit cycles of a quadratic
differential system. Acta Math. Sinica (Chin. Ser.) 22, 
751--758 (1979).

\bibitem{H3} Li, C., Liu, C., Yang, J.:
A cubic system with thirteen limit cycles.
J. Differ. Equations 246(9), 3609--3619 (2009).

\bibitem{H4} Prohens, R., Torregrosa, J.:
New lower bounds for the Hilbert numbers using reversible centers. 
Nonlinearity 32(1), 331--355 (2019).

\bibitem{Bound} Christopher, C.J., Lloyd, N.G.:
Polynomial Systems: A Lower Bound for the Hilbert Numbers. 
Proceedings: Mathematical and Physical Sciences, 
450, 219--224 (1995).

\bibitem{Radek} Erban, R., Kang, H.W.: 
Chemical Systems with Limit Cycles. 
Bull Math Biol 85, 76 (2023). 

\bibitem{Janos_chaos} Susits, M., T\'{o}th, J.:
Rigorously proven chaos in chemical kinetics.
Chaos 34, 103130 (2024).

\bcol{
\bibitem{Bimolecular_bifurcations}
Banaji, M., Boros, B., Hofbauer, J.: Bifurcations in planar, 
quadratic mass-action networks with few reactions and low molecularity.
Nonlinear Dyn 112, 21425--21448 (2024).

\bibitem{Degn–Harrison} Vismaya, V. S, Muni, S. S., 
Panda, A. K., Mondal, B.: Degn–Harrison map: 
Dynamical and network behaviours with applications in image encryption.
Chaos, Solitons and Fractals, 192, 115987 (2025).
}

\bibitem{Sprott} Sprott, J. C.:
Some simple chaotic flows.
Phys. Rev. E, 50, R647(R) (1994). 

\bibitem{Weierstrass} Narasimhan, R.: 
Analysis on Real and Complex Manifolds, Second Edition.
Elsevier Science (1985).
\end{thebibliography}
\end{document}